\documentclass[11pt]{article}

\usepackage[utf8]{inputenc}

\usepackage{geometry}
\usepackage{graphicx}

\usepackage{paralist} 

\usepackage{subfig} 
\usepackage{amsmath,amsfonts,amsthm,mathrsfs,amssymb,cite}

\usepackage[usenames]{color}

\newtheorem{thm}{Theorem}[section]

\newtheorem{prop}{Proposition}[section]
\theoremstyle{definition}

\theoremstyle{remark}
\newtheorem{rem}{Remark}[section]
\numberwithin{equation}{section}

\title{ Locating Multiple Multi-scale Electromagnetic Scatterers by A Single Far-field Measurement}
\author{Jingzhi Li\thanks{Faculty of Science, South University of Science and
Technology of China, 518055 Shenzhen, P.~R.~China. Email: {\tt li.jz@sustc.edu.cn}},\quad Hongyu Liu\thanks{Department of Mathematics and Statistics, University of North Carolina, Charlotte, NC 28223, USA.   Email:  {\tt hongyu.liuip@gmail.com}},\quad  Qi Wang\thanks{Department of Computing Sciences, School of Mathematics and Statistics, Xi{'}an Jiaotong University, Xi{'}an 710049, China.
Email:  {\tt qi.wang.xjtumath@gmail.com}}
}

\begin{document}
\maketitle

\begin{abstract}

Two inverse scattering schemes were recently developed in \cite{LiLiuShangSun} for locating multiple electromagnetic (EM) scatterers, respectively, of small size and regular size compared to the detecting EM wavelength. Both schemes make use of a single far-field measurement. The scheme of locating regular-size scatterers requires the {\it a priori} knowledge of the possible shapes, orientations and sizes of the underlying scatterer components. In this paper, we extend that imaging scheme to a much more practical setting by relaxing the requirement on the orientations and sizes. We also develop an imaging scheme of locating multiple multi-scale EM scatterers, which may include at the same time, both components of regular size and small size. For the second scheme, a novel local re-sampling technique is developed. Furthermore, more robust and accurate reconstruction can be achieved for the second scheme if an additional far-field measurement is used. Rigorous mathematical justifications are provided and numerical results are presented to demonstrate the effectiveness and the promising features of the proposed imaging schemes.

\end{abstract}

\section{Introduction}

We shall be concerned with the time-harmonic electromagnetic (EM) wave scattering.
Let
\begin{equation}\label{eq:plane waves}
E^i(x)=pe^{ik x\cdot d},\quad H^i(x)=\frac{1}{i k} \nabla\wedge E^i(x),\quad x\in\mathbb{R}^3\,,
\end{equation}
be a pair of time-harmonic EM plane waves, where $E^i$ and $H^i$ are, respectively, the electric and magnetic fields, and $k\in\mathbb{R}_+$, $d\in\mathbb{S}^2$, $p\in\mathbb{R}^3$ with $p\perp d$ are, respectively, the wave number, incident direction and polarization vector. In the homogeneous background space $\mathbb{R}^3$, where the EM medium is characterized by the electric permittivity $\varepsilon_0=1$, magnetic permeability $\mu_0=1$ and conductivity $\sigma_0=0$, the plane waves $(E^i, H^i)$ propagate indefinitely. If an EM inhomogeneity is presented in the homogeneous space, the propagation of the plane waves will be perturbed, leading to the so-called {\it scattering}. Throughout, we assume that the EM inhomogeneity is compactly supported in a bounded Lipschitz domain $\Omega\subset \mathbb{R}^3$ with $\mathbb{R}^3\backslash\overline{\Omega}$ connected. The inhomogeneity is referred to as a {\it scatterer}, and it is also characterized by the EM medium parameters including the electric permittivity $\varepsilon(x)$, the magnetic permeability $\mu(x)$ and the conductivity $\sigma(x)$. The medium parameters $\varepsilon(x)$, $\mu(x)$ and $\sigma(x)$ for $x\in\Omega$ are assumed to be $C^2$-smooth functions with $\varepsilon(x), \mu(x)>0$ and $\sigma(x)\geq 0$. The propagation of the total EM fields $(E,H)\in\mathbb{C}^3\wedge\mathbb{C}^3$ in the medium is governed by the Maxwell equations
\begin{equation}\label{eq:pp1}
\nabla\wedge E(x)-ik\mu(x) H(x)=0,\quad \nabla\wedge H(x)+(ik\varepsilon(x)-\sigma(x)) E(x)=0,\quad x\in\Omega,
\end{equation}
whereas in the background space it is governed by
\begin{equation}\label{eq:pp2}
\nabla\wedge E(x)-ik H(x)=0,\quad \nabla\wedge H(x)+ik E(x)=0,\quad x\in\mathbb{R}^3\backslash\overline{\Omega}.
\end{equation}
The total EM wave fields outside the inhomogeneity, namely in $\mathbb{R}^3\backslash \overline{\Omega}$, are composed of two parts: the incident wave fields $E^i, H^i$ and the scattered wave fields $E^+, H^+$. That is, we have
\begin{equation}\label{eq:total fields}
E(x)=E^i(x)+E^+(x),\quad H(x)=H^i(x)+H^+(x),\quad x\in\mathbb{R}^3\backslash\overline{\Omega}.
\end{equation}
The scattered EM fields are radiating, characterized by the Silver-M\"uller radiation condition
\[
\displaystyle{\lim_{|x|\rightarrow+\infty}|x|\left| (\nabla\wedge E^+)(x)\wedge\frac{x}{|x|}- ik E^+(x) \right|=0},
\]
which holds uniformly in all directions $\hat{x}:=x/|x|\in\mathbb{S}^2$, $x\in\mathbb{R}^3 \backslash \{0\}$. In the extreme situation where the conductivity of the inhomogeneity goes to infinity, the scatterer becomes perfectly conducting and the EM fields cannot penetrate inside $\Omega$. Moreover, the tangential component of the total electric field vanishes on the boundary of the scatterer, namely,
\begin{equation}
\nu\wedge E=0\quad \mbox{on\ \ $\partial\Omega$},
\end{equation}
where $\nu$ is the outward unit normal vector to $\partial\Omega$. In the perfectly conducting case, the scatterer is usually referred to as a {\it PEC obstacle}.

In summary, let us consider the scattering due to an inhomogeneous EM medium $(M; \varepsilon, \mu, \sigma)$ and a PEC obstacle $O$, and denote the combined scatterer $\Omega:=M\cup O$. $M$ and $O$ are assumed to be bounded Lipschitz domains with $\overline{M}\cap\overline{O}=\emptyset$ and $\mathbb{R}^3\backslash(\overline{M}\cup\overline{O})$ connected. The EM scattering is governed by the following Maxwell system
\begin{equation}\label{eq:Maxwell general}
\begin{cases}
\displaystyle{\nabla\wedge E-ik \bigg(1+(\mu-1)\chi_{M} \bigg) H=0}\ &  \mbox{in\ \ $\mathbb{R}^3\backslash\overline{O}$},\\
\displaystyle{\nabla\wedge H+\bigg(ik(1+(\varepsilon-1)\chi_{M})- \sigma\chi_{M} \bigg)E=0}\ & \mbox{in\ \ $\mathbb{R}^3\backslash\overline{O}$},\\
E^-=E|_{M},\ E^+=(E-E^i)|_{\mathbb{R}^3\backslash\overline{M\cup O}},\\
H^-=H|_M,\ H^+=(H-H^i)|_{\mathbb{R}^3\backslash\overline{M\cup O}}, \\
\nu\wedge E^+=-\nu\wedge E^i\hspace*{1cm}\mbox{on \ \ $\partial O$},\\
\displaystyle{\lim_{|x|\rightarrow+\infty}|x|\left| (\nabla\wedge E^+)(x)\wedge\frac{x}{|x|}- ik E^+(x) \right|=0}.
\end{cases}
\end{equation}
We refer to \cite{Lei,Ned} for the well-posedness study of the forward scattering problem \eqref{eq:Maxwell general}. There exists a unique pair of solutions $(E, H)\in H_{loc}(\text{curl}; \mathbb{R}^3\backslash\overline{O})\wedge H_{loc}(\text{curl}; \mathbb{R}^3\backslash\overline{O})$ to the system \eqref{eq:Maxwell general}. Moreover, $E^+$ admits the following asymptotic expansion  (cf.\!\cite{CK})
\begin{equation}\label{eq:farfield}
E^+(x)=\frac{e^{ik|x|}}{|x|} A\left(\frac{x}{|x|}; (M;\varepsilon,\mu, \sigma), O, d,p,k\right)+\mathcal{O}\left(\frac{1}{|x|^2}\right),
\quad \mathrm{as}~~ |x|\rightarrow +\infty\,,
\end{equation}
where $A(\hat{x}; (M;\varepsilon,\mu, \sigma), O, d, p, k)$,  (or for short $A(\hat{x})$ or $A(\hat{x};  d, p, k)$ with emphasis on the dependence), is known as the electric far-field pattern. Due to the real analyticity of   $A(\hat{x})$ on the unit sphere $\mathbb{S}^2$, if $A(\hat{x})$ is known on any open subset of $\mathbb{S}^2$, it is known on $\mathbb{S}^2$ by the analytic continuation.

The inverse problem that we shall consider is to recover the medium inclusion $(M;\varepsilon,\mu,\sigma)$ and/or the PEC obstacle $O$ by the knowledge of $A(\hat{x};d,p,k)$. In the physical situation, the inhomogeneous EM medium and the obstacle are the unknown/inaccessible target objects. One sends detecting EM plane waves and collects the corresponding scattered data produced by the underlying scatterer, and from which to infer knowledge of the target objects. This inverse scattering problem is of critical importance in many areas of science and technology, such as radar and sonar, geophysical exploration, medical imaging and non-destructive testing, to name just a few (cf.\!\cite{AK1,AK2,CK,Mar,Uhl}). If one introduces an operator $\mathcal{F}$ which maps the EM scatterer to the corresponding far-field pattern, the inverse scattering problem can be formulated as the following operator equation
\begin{equation}\label{eq:operator equation}
\mathcal{F}((M;\varepsilon,\mu,\sigma),O)=A(\hat{x};d,p,k),\quad \hat{x},d\in\mathbb{S}^2,\ p\in\mathbb{R}^3,\ k\in\mathbb{R}_+.
\end{equation}
It is widely known that the operator equation~\eqref{eq:operator equation} is non-linear and ill-posed (cf.\!\cite{CK}).

In this work, we are mainly concerned with the numerical reconstruction algorithms for the inverse scattering problem aforementioned. There are many results in the literature and various imaging schemes have been developed; see, e.g. \cite{AILP,AK1,CCM,CK,IJZ,KG,Pot,SZC,UZ,Zha} and the references therein. It is remarked that most existing schemes involve inversions, and in order to tackle the ill-posedness, regularizations are always utilized. For the present study, we are particularly interested in the reconstruction by making use of a single far-field measurement, namely $A(\hat{x};d,p,k)$ for all $\hat{x}\in\mathbb{S}^2$ but fixed $d\in\mathbb{S}^2, p\in\mathbb{R}^3$ and $k\in\mathbb{R}_+$. Here, we note that in \eqref{eq:operator equation}, the unknown scatterer depends on a 3D parameter $x\in\Omega$, whereas the far-field pattern depends on a 2D parameter $\hat{x}\in\mathbb{S}^2$, and hence, much less information is used for the proposed reconstruction scheme. The inverse electromagnetic scattering problem with minimum measurement data is extremely challenging with very limited theoretical and computational progress in the literature (cf.\!\cite{CK,Isa,LZr}). Furthermore, we shall conduct our study in a very general and complex environment. The target scatterer may consist of multiple components with an unknown number, and each component could be either an inhomogeneous medium inclusion or a PEC obstacle. 

Two imaging schemes   using a single measurement  were recently proposed in \cite{LiLiuShangSun}, namely Schemes S and R, for locating multiple EM scatterer components, respectively,  of small size and regular size compared to the wavelength of the incident EM plane waves. The schemes rely on certain new indicator functions, which can be directly calculated from the measured far-field  data. In calculating the indicator functions, there are no inversions or regularizations and hence the proposed schemes are shown to be very efficient and robust to noisy data. However, Scheme R, the locating scheme for regular-size scatterers, requires the {\it a priori} knowledge of the possible shapes, orientations and sizes of the underlying scatterer components. It is our {first goal} to extend Scheme R in \cite{LiLiuShangSun} to an improved Scheme AR by relaxing the requirement on orientations and sizes. This is achieved in light of the idea  that the admissible reference space can be augmented by more data carrying  information about orientations and sizes. Next, based on the newly developed Scheme AR and Scheme S in \cite{LiLiuShangSun}, in a certain generic practical setting, we develop a novel imaging procedure (Scheme M) in locating multiple multi-scale EM scatterers, which include both regular-size and small-size components. A novel local re-sampling technique is proposed and plays a key role in tackling the challenging multi-scale reconstruction in Scheme M. Furthermore, for the multi-scale locating scheme, if one additional set of far-field data is used, more robust and accurate reconstruction can be achieved. To our best knowledge, this is the first reconstruction scheme in the literature on recovering multi-scale EM scatterers by using such less scattering information. For all the proposed imaging schemes, we provide rigorous mathematical justifications. We also conduct systematical numerical experiments to demonstrate the effectiveness and the promising features of the schemes.


The rest of the paper is organized as follows. Section~\ref{sect:2} is devoted to the description of multi-scale EM scatterers and the two locating schemes in \cite{LiLiuShangSun}. In Section~\ref{sect:3}, we develop techniques on relaxing the requirement on knowledge of the orientation and size for locating regular-size scatterers. In Section~\ref{sect:4}, we present the imaging schemes of locating multiple multi-scale scatterers. Finally, in Section~\ref{sect:5}, numerical experiments are given to demonstrate the effectiveness and the promising features of the proposed imaging schemes.

\section{Multi-scale EM scatterers and two locating schemes}\label{sect:2}

Throughout the rest of the paper, we assume that $k\sim 1$. That is, the wavelength of the EM plane waves is given by $\lambda=2\pi/k\sim 1$ and hence the size of a scatterer can be expressed in terms of its Euclidean diameter.

\subsection{Scheme S}

We first introduce the class of small scatterers for our study. Let $l_s\in\mathbb{N}$ and $D_j$, $1\leq j\leq l_s$ be bounded Lipschitz domains in $\mathbb{R}^3$. It is assumed that all $D_j$'s are simply connected and contain the origin. For $\rho\in\mathbb{R}_+$, we let $\rho D_j:=\{\rho x; x\in D_j\}$ and set
\begin{equation}\label{eq:small component}
\Omega_j^{(s)}=z_j+\rho D_j,\quad z_j\in\mathbb{R}^3,\ \ 1\leq j\leq l_s.
\end{equation}
Each $\Omega_j^{(s)}$ is referred to as a scatterer component and its content is endowed with $\varepsilon_j, \mu_j$ and $\sigma_j$. The parameter $\rho\in\mathbb{R}_+$ represents the relative size of the scatterer (or, more precisely, each of its components).  The scatterer components $(\Omega_j^{(s)}; \varepsilon_j, \mu_j, \sigma_j)$, $1\leq j\leq l_s$, are assumed to satisfy:~i).~if for some $j$, $0\leq \sigma_j<+\infty$, then $\varepsilon_j, \mu_j$ and $\sigma_j$ are all real valued $C^2$-smooth functions in the closure of ${\Omega_j^{(s)}}$; ii).~in the case of i), the following condition is satisfied, $|\varepsilon_j(x)-1|+|\mu_j(x)-1|+|\sigma_j(x)|>c_0>0$ for all $x\in\Omega_j^{(s)}$ and some positive constant $c_0$; iii).~if for some $j$, $\sigma_j=+\infty$, then disregarding the parameters $\varepsilon_j$ and $\mu_j$, $\Omega_j^{(s)}$ is regarded as a PEC obstacle. Condition ii) means that if $(\Omega_j^{(s)}; \varepsilon_j, \mu_j, \sigma_j)$ is a medium component, then it is inhomogeneous from the homogeneous background space. We set
\begin{equation}\label{eq:small scatterer}
\Omega^{(s)}:=\bigcup_{j=1}^{l_s} \Omega_j^{(s)}\quad \mbox{and}\quad (\Omega^{(s)};\varepsilon,\mu,\sigma):=\bigcup_{j=1}^{l_s} (\Omega_j^{(s)}; \varepsilon_j,\mu_j,\sigma_j).
\end{equation}
and  make the following qualitative assumption,
\begin{equation}\label{eq:qualitative assumptions}
\rho\ll 1\qquad \mbox{and}\qquad \mbox{dist}(z_j, z_{j'})\gg 1\quad \mbox{for\ $j\neq j'$, $1\leq j, j'\leq l_s$}.
\end{equation}
The assumption \eqref{eq:qualitative assumptions} implies that compared to the wavelength of the incident plane waves, the relative size of each scatterer component is small and if there are multiple components, they are sparsely distributed. It is numerically shown in \cite{LiLiuShangSun} that if the relative size is smaller than half a wavelength and the distance between two different components is bigger than half a wavelength, the scheme developed there works well for locating the multiple components of $\Omega^{(s)}$. 
Let $0\leq l_s'\leq l_s$ be such that when $1\leq j\leq l_s'$, $\sigma_j=+\infty$, and when $l_s'+1\leq j\leq l_s$, $0\leq \sigma_j<+\infty$. That is, if $1\leq j\leq l_s'$, $\Omega_j^{(s)}$ is a PEC obstacle component, whereas if $l_s'+1\leq j\leq l_s$, $(\Omega_j^{(s)}; \varepsilon_j, \mu_j, \sigma_j)$ is a medium component. If $l_s'=0$, then all the components of the small scatterer $\Omega^{(s)}$ are of medium type and if $l_s'=l_s$, then all the components are PEC obstacles. The EM scattering corresponding to $\Omega^{(s)}$ due to a single pair of incident waves $(E^i, H^i)$ is governed by \eqref{eq:Maxwell general} with $O=\bigcup_{j=1}^{l_s'} \Omega_j^{(s)}$ and $(M; \varepsilon, \mu, \sigma)=\bigcup_{j=l_s'+1}^{l_s} (\Omega_j^{(s)}; \varepsilon_j, \mu_j, \sigma_j)$. We denote the electric far-field pattern by $A(\hat{x}; \Omega^{(s)})$.

In order to locate the multiple components of $\Omega^{(s)}$ in \eqref{eq:small scatterer}, the following indicator function is introduced in \cite{LiLiuShangSun},
\begin{equation}\label{eq:indicator function}
\begin{split}
I_s(z):=\frac{1}{\|A(\hat x;\Omega^{(s)})\|^2_{T^2(\mathbb{S}^2)}}&\sum_{m=-1,0,1}\bigg( {\bigg|\left\langle A(\hat x;\Omega^{(s)}), e^{ik (d-\hat x)\cdot z}\, U_1^m(\hat x)  \right\rangle_{T^2(\mathbb{S}^2)}\bigg|^2}\\
& +{\bigg|\left\langle A(\hat x;\Omega^{(s)}), e^{i k (d-\hat x)\cdot z}\, V_1^m(\hat x)  \right\rangle_{T^2(\mathbb{S}^2)}\bigg|^2}         \bigg),\ \ z\in\mathbb{R}^3,
\end{split}
\end{equation}
where
\[
T^2(\mathbb{S}^2):=\{\mathbf{a}\in\mathbb{C}^3|\ \mathbf{a}\in L^2(\mathbb{S}^2)^3, \ \hat{x}\cdot \mathbf{a}=0\ \ \mbox{for a.e. $\hat{x}\in\mathbb{S}^2$}\}.
\]
and
\begin{equation*}
U_n^m(\hat{x}):=\frac{1}{\sqrt{n(n+1)}}\text{Grad}\, Y_n^m(\hat{x}),\ \
V_n^m(\hat x):=\hat x\wedge U_n^m(\hat x),
 \  n\in\mathbb{N},\  \ m=-n,\cdots,n,
\end{equation*}
with $Y_n^m(\hat x)$, $m=-n,\ldots,n$ the spherical harmonics of order $n\geq 0$ (cf.\!\cite{CK}). It is shown in \cite{LiLiuShangSun} that $z_j$ (cf.\!~\eqref{eq:small component}), $1\leq j\leq l_s$, is a local maximum point for $I_s(z)$. Based on such indicating behavior, the following scheme is proposed in \cite{LiLiuShangSun} for locating the multiple components of the small scatterer $\Omega^{(s)}$.

\medskip

\hrule

\medskip
\noindent {\bf Algorithm: Locating Scheme S}
\medskip

\hrule
\begin{enumerate}[1)]

\item For an unknown EM scatterer $\Omega^{(s)}$ in \eqref{eq:small scatterer}, collect the far-field data by sending a
 single pair of detecting EM plane waves specified in \eqref{eq:plane waves}.

\item Select a sampling region with a mesh $\mathcal{T}_h$ containing $\Omega^{(s)}$.

\item For each point $z\in \mathcal{T}_h$, calculate $I_s(z). $

\item Locate all the significant local maxima of $I_s(z)$ on $\mathcal{T}_h$, which represent the locations of the scatterer components.
\end{enumerate}

\hrule
\medskip

\subsection{Scheme R}

Next, we consider the locating of multiple obstacles of regular size. For this locating scheme, one must require the following generic uniqueness result holds for the inverse scattering problem. Let $O_1$ and $O_2$ be obstacles and both of them are assumed to be bounded simply connected Lipschitz domains in $\mathbb{R}^3$ containing the origin. Then
\begin{equation}\label{eq:uniqueness}
\text{$A(\hat{x}; O_1)=A(\hat{x}; O_2)$ if and only if $O_1=O_2$. }
\end{equation}
This result implies that by using a single far-field measurement, one can uniquely determine an obstacle. There is a widespread belief that such a uniqueness result holds, but there is only limited progress in the literature, see, e.g., \cite{L,LZr,LYZ}. Throughout the present study, we shall assume that such a generic uniqueness holds true.

We now briefly recall {\it Scheme R} in \cite{LiLiuShangSun}  for locating multiple regular-size obstacles. Let $l_r\in\mathbb{N}$ and let $G_j$, $1\leq j\leq l_r$ be bounded simply connected Lipschitz domains containing the origin in $\mathbb{R}^3$. Set
\begin{equation}\label{eq:regular component}
\Omega_j^{(r)}=z_j+G_j,\quad z_j\in\mathbb{R}^3,\ \ 1\leq j\leq l_r.
\end{equation}
Each $\Omega_j^{(r)}$ denotes a PEC obstacle located at the position $z_j\in\mathbb{R}^3$. It is required that
\begin{equation}\label{eq:regular condition}
\text{diam}(\Omega_j^{(r)})=\text{diam}(G_j)\sim 1,\ 1\leq j\leq l_r;\ \ \ L=\min_{1\leq j, j'\leq l_r, j\neq j'}\text{dist}(z_j, z_{j'})\gg 1.
\end{equation}
Furthermore, there exists an admissible reference obstacle space
\begin{equation}\label{eq:reference space}
\mathscr{S}:=\{\Sigma_j\}_{j=1}^{l'},
\end{equation}
where each $\Sigma_j\subset\mathbb{R}^3$ is a bounded simply connected Lipschitz domain that contains the origin and
\begin{equation}\label{eq:assumption 1}
\Sigma_j\neq \Sigma_{j'},\quad \mbox{for}\ \ j\neq j',\  1\leq j, j'\leq l',
\end{equation}
such that
\begin{equation}\label{eq:shape known}
G_j\in\mathscr{S},\quad j=1,2,\ldots,l_r.
\end{equation}
The admissible class $\mathscr{S}$ is required to be known in advance, and by reordering if necessary, it is assumed that
\begin{equation}\label{eq:assumption2}
\| A(\hat{x};\Sigma_{j}) \|_{T^2(\mathbb{S}^2)}\geq \| A(\hat x;\Sigma_{j+1}) \|_{T^2(\mathbb{S}^2)},\quad j=1,2,\ldots,l'-1.
\end{equation}
Let
\begin{equation}\label{eq:regular scatterer}
\Omega^{(r)}:=\bigcup_{j=1}^{l_r} \Omega_j^{(r)}.
\end{equation}
Then $\Omega^{(r)}$ denotes the regular-size scatterer for our current study, which may consist of multiple obstacle components. The second condition in \eqref{eq:regular condition} means that the components are sparsely distributed. It is numerically observed in \cite{LiLiuShangSun}   that if the distance is larger than a few numbers of wavelength, then Scheme R works effectively. The assumption \eqref{eq:shape known} indicates that certain {\it a priori} knowledge of the target scatterer is required. It is remarked that $l_r$ is not necessarily the same as $l'$. Define $l'$ indicator functions as follows,
\begin{equation}\label{eq:indicator regular}
I^j_r(z)=\frac{\bigg| \langle A(\hat x;\Omega^{(r)}), e^{ik(d-\hat x)\cdot z} A(\hat x; \Sigma_j)  \rangle_{T^2(\mathbb{S}^2)}  \bigg|}{\| A(\hat x; \Sigma_j) \|^2_{T^2(\mathbb{S}^2)}},\quad j=1,2,\ldots, l',\quad z\in\mathbb{R}^3.
\end{equation}
The following indicating behavior of $I^j_r(z)$'s is proved in \cite{LiLiuShangSun} and summarized  below.
\begin{thm}\label{thm:main2}
Consider the indicator function $I_r^1(z)$ introduced in \eqref{eq:indicator regular}. Suppose there exists $J_0\subset\{1,2,\ldots,l_r\}$ such that for $j_0\in J_0$, $G_{j_0}=\Sigma_1$, whereas $G_j\neq \Sigma_1$ for $j\in \{1,2,\ldots,l_r\}\backslash J_0$. Then for each $z_j$, $j=1,2,\ldots,l_r$, there exists an open neighborhood of $z_j$, $neigh(z_j)$, such that
\begin{enumerate}
\item[(i).]~if $j\in J_0$, then
\begin{equation}\label{eq:further ind 1}
\widetilde{I}_r^1(z):=|I_r^1(z)-1|\leq \mathcal{O}\left( \frac 1 L  \right ),\quad z\in neigh(z_{j}),
\end{equation}
and moreover, $z_{j}$ is a local minimum point for $\widetilde{I}_r^1(z)$;

\item[(ii).]~if $j\in \{1,2,\ldots,l_r\}\backslash J_0$, then there exists $\epsilon_0>0$ such that
\begin{equation}\label{eq:further ind 2}
\widetilde{I}_r^1(z):=|I_r^1(z)-1|\geq \epsilon_0+\mathcal{O}\left( \frac 1 L \right ),\quad z\in neigh(z_j).
\end{equation}
\end{enumerate}
\end{thm}

Based on Theorem~\ref{thm:main2}, the {\it Scheme R} for locating the multiple components in $\Omega^{(r)}$ can be successively formulated as follows.

\medskip


\hrule

\medskip
\noindent {\bf Algorithm: Locating Scheme R}
\medskip

\hrule
\begin{enumerate}[1)]

\item For an unknown EM scatterer $\Omega^{(r)}$ in \eqref{eq:regular scatterer}, collect the far-field data by sending a single pair of detecting EM plane waves specified in \eqref{eq:plane waves}.

\item Select a sampling region with a mesh $\mathcal{T}_h$ containing $\Omega^{(r)}$.

\item Collect in advance the far-field patterns associated with the admissible reference scatterer space $\mathscr{S}$
                      in \eqref{eq:reference space}, and reorder $\mathscr{S}$  if necessary to make it satisfy
\eqref{eq:assumption2}, and also verify the generic assumption \eqref{eq:uniqueness}.

\item Set $j=1$.

\item For each point $z\in \mathcal{T}_h$, calculate $I_r^j(z)$ (or $\widetilde{I}_r^j(z)=|I_r^j(z)-1|$).

\item Locate all those significant local maxima of $I_r^j(z)$ such that $I_r^j(z)\sim 1$ (or the minima of $\widetilde{I}_r^j(z)$ on $\mathcal{T}_h$ such that $\widetilde{I}_r^j(z)\ll 1$), where scatterer components of the form $z+\Sigma_j$ is located.

\item Trim all those $z+\Sigma_j$ found in 6) from $\mathcal{T}_h$.

\item If $\mathcal{T}_h=\emptyset$ or $j=l'$, then Stop; otherwise, set $j=j+1$, and go to 5).

\end{enumerate}

\hrule

\medskip

\begin{rem}\label{rem:MediumObstacle}
By \eqref{eq:uniqueness} and \eqref{eq:assumption 1}, it is readily seen that
\begin{equation}\label{eq:rerell}
A(\hat{x}; \Sigma_j)\neq A(\hat{x}; \Sigma_{j'}),\quad j\neq j',\ 1\leq j, j'\leq l'. 
\end{equation}
\eqref{eq:rerell} plays a critical role in justifying the indicating behavior of $I_r^j(z)$ in Theorem~\ref{thm:main2}. Nevertheless, since the reference space
\eqref{eq:reference space} is given, one can verify \eqref{eq:rerell} in advance. On the other hand, one can also include inhomogeneous medium components into the admissible reference space provided the relation \eqref{eq:rerell} is satisfied. For the inhomogeneous medium component in $\mathscr{S}$, its content is required to be known in advance; see Remark~\ref{rem:MediumObstacle2} in the following. 
\end{rem}

{\it Scheme R} could find important practical applications, e.g., in radar technology in locating an unknown group of aircrafts, where one has the {\it a priori} knowledge on the possible models of the target airplanes. However, we note here some important practical situations that Scheme R does not cover. Indeed, in Scheme R, it is required that each component, say $\Omega_1^{(r)}$, is a translation of the reference obstacle $\Sigma_1$, namely $\Omega_1^{(r)}=z+\Sigma_1$. This means that, in addition to the shape of the obstacle component $\Omega_1^{(r)}$, one must also know its orientation and size in advance (two concepts to be mathematically specified in Section~\ref{sect:3}). In the radar technology, this means that in addition to the model of each aircraft, one must also know which direction the aircraft is heading to. Clearly, this limits the applicability of the locating scheme. In the next section, we shall propose strategies to relax the limitations about the requirement on orientation and size. Furthermore, we shall consider the locating of multiple multi-scale scatterers, which may include, at the same time, small- and regular-size scatterers. To that end, we introduce the multiple multi-scale scatterer for our subsequent study
\begin{equation}\label{eq:multiscale scatterer}
\Omega^{(m)}:=\Omega^{(s)}\cup\Omega^{(r)},
\end{equation}
where $\Omega^{(s)}$ and $\Omega^{(r)}$ are, respectively, given in \eqref{eq:small scatterer} and \eqref{eq:regular scatterer}. 

\section{Scheme R with augmented reference spaces}\label{sect:3}

In this section, we propose an enhanced version of Scheme R with augmented reference spaces to image a regular-size scatterer with multiple components
of different shapes, orientations and sizes. This goal is achieved through collecting more reference far field data of 
a set of a priori known components, in particular associated with their possible orientations and sizes.

Let $\Pi_{\theta,\phi,\psi}$ denote the 3D rotation whose Euler angles are $\theta,\phi$ and $\psi$ with the $x_1-x_2-x_3$ convention for $x=(x_1,x_2,x_3)\in\mathbb{R}^3$. That is, $\Pi_{\theta,\phi,\psi} x=U(\theta,\phi,\psi) x$, where $U\in SO(3)$ is given by
\begin{equation}\label{eq:U}
\!\!\!\!\!U=\begin{pmatrix}
\cos\theta\cos\psi & -\cos\theta\sin\psi+\sin\phi\sin\theta\cos\psi & \sin\phi\sin\psi+\cos\phi\sin\theta\cos\psi\\
\cos\theta\sin\psi & \cos\phi\cos\psi+\sin\phi\sin\theta\sin\psi & -\sin\phi\cos\psi+\cos\phi\sin\theta\sin\psi \\
-\sin\theta & \sin\phi\cos\theta & \cos\phi\cos\theta
\end{pmatrix} 
\end{equation}
with $0\leq \theta,\phi\leq 2\pi$ and $0\leq \psi\leq \pi$. In the sequel, we suppose there exist triplets $(\theta_j,\phi_j,\psi_j)$, $j=1,2,\ldots, l_r$ such that
\begin{equation}\label{eq:regular scatterer 2}
\Omega_j^{(r)}=z_j+\Pi_{\theta_j,\phi_j,\psi_j} G_j,
\end{equation}
where $G_j\in\mathscr{S}$ defined in \eqref{eq:reference space}.
Now, we let
\begin{equation}\label{eq:rot2}
\Omega^{(r)}=\bigcup_{j=1}^{l_r} (z_j+\Pi_{\theta_j,\phi_j,\psi_j} G_j):=\bigcup_{j=1}^{l_r} (z_j+\widetilde G_j)
\end{equation}
denote the regular-size target scatterer for our current study. Compared to the regular-size scatterer in \eqref{eq:regular scatterer} considered in \cite{LiLiuShangSun} (cf. \eqref{eq:regular component}--\eqref{eq:regular scatterer}), the scatterer introduced in \eqref{eq:rot2} possesses the new feature that each component is allowed to be rotated. In the sequel, the Euler angles $(\theta_j,\phi_j,\psi_j)$ will be referred to as the {\it orientation} of the scatterer component $\Omega_j^{(r)}$ in \eqref{eq:regular scatterer 2}.

Next, we also introduce a scaling/dilation operator $\Lambda_{\tau_j}$, $\tau_j\in\mathbb{R}_+$, and for $\Omega_j^{(r)}=z_j+G_j$, $G_j\in\mathscr{S}$, we set
\begin{equation}\label{eq:sca1}
\Omega_j^{(r)}:=z_j+\Lambda_{\tau_j} G_j,
\end{equation}
where $\Lambda_{\tau_j} G_j:=\{{\tau_j} x\,; \ x\in G_j\}$. Now, for a sequence of $\{\tau_j\}_{j=1}^{l_r}$ 
we set
\begin{equation}\label{eq:sca2}
\Omega^{(r)}=\bigcup_{j=1}^{l_r} (z_j+\Lambda_{\tau_j} G_j).
\end{equation}
We shall call $\tau_j$ the {\it size} or {\it scale} of the component $\Omega_j^{(r)}$ relative to the reference one $G_j$.  

For our subsequent study, we would consider locating a regular-size scatterer with its components both possibly orientated and scaled,
\begin{equation}\label{eq:regular scatterer 4}
\Omega^{(r)}=\bigcup_{j=1}^{l_r} (z_j+\Pi_{\theta_j,\phi_j,\psi_j}\Lambda_{\tau_j} G_j):=\bigcup_{j=1}^{l_r} (z_j+\widehat G_j).
\end{equation}
Compared to the scatterer in \eqref{eq:regular scatterer} considered in \cite{LiLiuShangSun}, the scatterer introduced in \eqref{eq:sca2} is scaled relatively. 
To that end, we first show a relation of the far-field pattern when the underlying scatterer is rotated and scaled.

\begin{prop}\label{prop:rot1}
Let $G$ be a bounded simply connected Lipschitz domain containing the origin, which represents a PEC obstacle. Then, we have that
\begin{equation}\label{eq:rot3}
A(\hat{x}; \Pi_{\theta,\phi,\psi}G, d, p, k)=UA(U^T\hat{x}; G, U^Tp, U^Td, k),
\end{equation}
where $U=U(\theta,\phi,\psi)$ is the rotation matrix corresponding to $\Pi_{\theta,\phi,\psi}$; and
\begin{equation}\label{eq:sca3}
A(\hat{x}; \Lambda_\tau G, d, p, k)=\tau A(\hat{x}; G, d, p, k\tau)
\end{equation}
\end{prop}
\begin{proof}
Let $E\in H_{loc}^1(\mathbb{R}^3\backslash\Pi_{\theta,\phi,
\psi} G)$ and $H\in H_{loc}^1(\mathbb{R}^3\backslash\Pi_{\theta,\phi,
\psi} G)$ be the solutions to the following Maxwell system
\begin{equation}\label{eq:eee1}
\begin{split}
& \nabla\wedge E-ik H=0,\qquad\qquad \nabla\wedge H+ikE=0\quad \mbox{in\ \ $\mathbb{R}^3\backslash\overline{\Pi_{\theta,\phi,\psi}G}$ },\\
& \nu\wedge E=0\quad\mbox{on\ \ $\partial (\Pi_{\theta,\phi,\psi}G)$},\quad\, \ E=E^i+E^+\quad \mbox{in\ \ $\mathbb{R}^3\backslash\overline{\Pi_{\theta,\phi,\psi}G}$ },\\
& {\lim_{|x|\rightarrow+\infty}|x|\left| (\nabla\wedge E^+)(x)\wedge\frac{x}{|x|}- ik E^+(x) \right|=0},
\end{split}
\end{equation}
where $E^i(x)=p e^{ikx\cdot d}$ and $\nu$ is the outward unit normal vector to $\partial(\Pi_{\theta,\phi,\psi} G)$.
Set
\begin{equation}\label{eq:ttt1}
\begin{split}
\widetilde{E}=&\Pi_{\theta,\phi,\psi}^* E:=\Pi_{\theta,\phi,\psi}^{-1}\circ E\circ\Pi_{\theta,\phi,
\psi}=U^TE\circ U\\
\widetilde{H}=&\Pi_{\theta,\phi,\psi}^* H:=\Pi_{\theta,\phi,\psi}^{-1}\circ H\circ\Pi_{\theta,\phi,
\psi}=U^TH\circ U
\end{split}\qquad \mbox{in\ \ $\mathbb{R}^3\backslash\overline{G}$},
\end{equation}
and
\begin{equation}\label{eq:ttt2}
\widetilde{E}^i(x):=(U^Tp) e^{ik x\cdot (U^Td)}
\end{equation}
Then, by the transformation properties of Maxwell's equations (see, e.g., \cite{LiuZhou}), it is straightforward to verify that
\begin{equation}\label{eq:ttt3}
\begin{split}
&\nabla\wedge\widetilde{E}-ik\widetilde{H}=0,\quad \nabla\wedge\widetilde{H}+ik\widetilde{E}=0\quad\mbox{in\ \ $\mathbb{R}^3\backslash \overline{G}$},\\
& \widetilde\nu\wedge\widetilde E=0\quad\mbox{on\ \ $\partial G$},\qquad \widetilde{E}=\widetilde{E}^i+\widetilde{E}^+\quad \mbox{in\ \ $\mathbb{R}^3\backslash\overline{G}$},\\
& {\lim_{|x|\rightarrow+\infty}|x|\left| (\nabla\wedge \widetilde{E}^+)(x)\wedge\frac{x}{|x|}- ik \widetilde{E}^+(x) \right|=0},
\end{split}
\end{equation}
where $\widetilde\nu$ is the outward unit normal vector to $\partial G$. Clearly, $A(\hat{x}; \Pi_{\theta,\phi,\psi} G)$ can be read-off from the large $|x|$ asymptotics of $E(x)$ in \eqref{eq:eee1},
\begin{equation}\label{eq:ffnn1}
E(x)=p e^{ikx\cdot d}+\frac{e^{ik|x|}}{|x|} A\left(\frac{x}{|x|}; \Pi_{\theta,\phi,\psi} G,d,p,k\right)+\mathcal{O}\left(\frac{1}{|x|^2}\right).
\end{equation}
Hence, by \eqref{eq:ttt1} and \eqref{eq:ffnn1}, we have
\begin{equation}\label{eq:ffnn2}
\begin{split}
\widetilde E(x)=&U^TE(Ux)\\
=&U^Tpe^{ik Ux\cdot d}+\frac{e^{ik|Ux|}}{|Ux|} U^T A\left(\frac{Ux}{|Ux|}; \Pi_{\theta,\phi,\psi} G, d, p, k\right)+\mathcal{O}\left(\frac{1}{|Ux|^2}\right)\\
=&U^Tp e^{ikx\cdot U^Td}+\frac{e^{ik|x|}}{|x|} U^TA(U\hat{x}; \Pi_{\theta,\phi,\psi}G, d, p, k)+\mathcal{O}\left(\frac{1}{|x|^2}\right).
\end{split}
\end{equation}
By \eqref{eq:ttt3} and \eqref{eq:ffnn2}, one can readily see that
\[
U^TA(U\hat{x};\Phi_{\theta,\phi,\psi} G, d, p, k)=A(\hat{x}; G, \widetilde{E}^i)=A(\hat{x}; G, U^Td,U^Tp,k).
\]
which immediately implies \eqref{eq:rot3}.

In a completely similar manner, one can show \eqref{eq:sca3}. The proof is complete.

\end{proof}

Proposition~\ref{prop:rot1} suggests that in order to locate a scatterer $\Omega^{(r)}$ in \eqref{eq:rot2} by using the Scheme R, one can make use of the multi-polarization and multi-incident-direction far-field data, namely $A(\hat{x}; p, d, k)$ for all $p\in\mathbb{R}^3$, $d\in\mathbb{S}^2$ and a fixed $k\in\mathbb{R}_+$. On the other hand, in order to still make use of a single far-field for the locating, one can augment the reference space $\mathscr{S}$ by letting
\begin{equation}\label{eq:aug1}
\widetilde{\mathscr{S}}=\Pi_{\theta,\phi,\psi}\mathscr{S}:=\{\Pi_{\theta,\phi,\psi}\Sigma_j\}_{j=1}^{l'},\quad (\theta,\phi,\psi)\in [0,2\pi]^2\times [0,\pi].
\end{equation}
Furthermore, from a practical viewpoint, we introduce a discrete approximation of $\widetilde{\mathscr{S}}$ and set
\begin{equation}\label{eq:discrete aug1}
\widetilde{\mathscr{S}}_h:=\{\Pi_{\theta^h,\phi^h,\psi^h}\Sigma_j\}_{j=1}^{l'}=\{\widetilde{\Sigma}_j\}_{j=1}^{\widetilde{l}_h},
\end{equation}
where $(\theta^h,\phi^h,\psi^h)$ denotes an equal distribution over $[0,2\pi]^2\times[0,\pi]$ with an angular mesh-size $h\in\mathbb{R}_+$ and its cardinality $N_h$, and $\widetilde{l}_h:=l'\times N_h$. By reordering if necessary, we assume the non-increasing relation \eqref{eq:assumption2} also holds for those components. Next, based on the same single far-field data for Scheme S, one can calculate $\widetilde{l}_h$ indicator functions according to \eqref{eq:indicator regular}, but with the reference scatterers taken from $\widetilde{\mathscr{S}}_h$. We denote the $\widetilde{l}_h$ indicator functions by $I_{h}^j(z)$, $1\leq j\leq \widetilde{l}_h$. Then, we have

\begin{thm}\label{thm:main23}
Consider the multiple scatterers introduced in \eqref{eq:rot2} and the indicator function $I_h^1(z)$ introduced above. Let $\widetilde\Sigma_l\in\widetilde{\mathscr{S}}_h$ be such that
\begin{equation}\label{eq:main231}
\widetilde\Sigma_1=\Pi_{\theta_1^h,\phi_1^h,\psi_1^h}\Sigma_{m_0}\quad \mbox{with}\quad \Sigma_{m_0}\in\mathscr{S}.
\end{equation}
Suppose there exists $J_0\subset\{1,2,\ldots,l_r\}$ such that for $j_0\in J_0$,
\[
\Omega_{j_0}^{(r)}=z_{j_0}+\widetilde{G}_{j_0}=z_{j_0}+\Pi_{\theta_{j_0},\phi_{j_0},\psi_{j_0}} G_{j_0}
\]
with
\begin{equation}\label{eq:cond1}
G_{j_0}=\Sigma_{m_0}\quad\mbox{and}\quad \|(\theta_{j_0},\phi_{j_0}, \psi_{j_0})-(\theta_1^h,\phi_1^h,\psi_1^h)\|_{l^\infty}=\mathcal{O}(h);
\end{equation}
whereas for the other components $\Omega_{j}^{(r)}$, $j\in\{1,2,\ldots,l_r\}\backslash J_0$, either of the two conditions in \eqref{eq:cond1} is violated. Then for each $z_j$, $j=1,2,\ldots,l_r$, there exists an open neighborhood of $z_j$, $neigh(z_j)$, such that
\begin{enumerate}
\item[(i).]~if $j\in J_0$, then
\begin{equation}\label{eq:main232}
\widetilde{I}_h^1(z):=|I_h^1(z)-1|\leq \mathcal{O}\left( \frac 1 L+h  \right ),\quad z\in neigh(z_{j}),
\end{equation}
and moreover, $z_{j}$ is a local minimum point for $\widetilde{I}_h^1(z)$;

\item[(ii).]~if $j\in \{1,2,\ldots,l_r\}\backslash J_0$, then there exists $\epsilon_0>0$ such that
\begin{equation}\label{eq:main233}
\widetilde{I}_h^1(z):=|I_h^1(z)-1|\geq \epsilon_0+\mathcal{O}\left( \frac 1 L+h \right ),\quad z\in neigh(z_j).
\end{equation}
\end{enumerate}
\end{thm}

\begin{proof}
Let
\[
\widetilde{\Gamma}_1:=\Pi_{\theta_{j_0},\phi_{j_0},\psi_{j_0}}\Sigma_{m_0},
\]
and
\[
H^1_r(z)=\frac{\bigg| \langle A(\hat x;\Omega^{(r)}), e^{ik(d-\hat x)\cdot z} A(\hat x; \widetilde\Gamma_1)  \rangle_{T^2(\mathbb{S}^2)}  \bigg|}{\| A(\hat x; \widetilde\Gamma_1) \|^2_{T^2(\mathbb{S}^2)}},\quad z\in\mathbb{R}^3.
\]
By a completely similar argument to the proof of Theorem~2.1 in \cite{LiLiuShangSun}, one can show that $H_r^1(z)$ possesses the two indicating behaviors given in \eqref{eq:further ind 1} and \eqref{eq:further ind 2}. Next, by Proposition~\ref{prop:rot1}, we have
\begin{equation}\label{eq:pA}
A(\hat{x}; \widetilde\Gamma_1)=A(\hat{x}; \Pi_{\theta_0,\phi_0,\psi_0}\Sigma_{m_0})=U_0A(U_0^T\hat{x}; \Sigma_{m_0}, U_0^Tp, U_0^Td, k),
\end{equation}
and
\begin{equation}\label{eq:pB}
 A(\hat{x}; \widetilde\Sigma_1)=A(\hat{x}; \Pi_{\theta_1^h,\phi_1^h,\psi_1^h}\Sigma_{m_0})=U_hA(U_h^T\hat{x}; \Sigma_{m_0}, U_h^Tp, U_h^Td, k),
\end{equation}
where $U_0$ and $U_h$ are the rotation matrices corresponding to $\Pi_{\theta_0,\phi_0,\psi_0}$ and $\Pi_{\theta_1^h,\phi_1^h,\psi_1^h}$, respectively. By the second assumption in \eqref{eq:cond1}, it is straightforward to show that
\begin{equation}\label{eq:pp3}
\|A(\hat{x}; \widetilde\Gamma_1)-A(\hat{x};\widetilde\Sigma_1)\|_{T^2(\mathbb{S}^2)}=\mathcal{O}(h).
\end{equation}
Finally, by \eqref{eq:pp3}, one has by direct verification that
\begin{equation}\label{eq:pp4}
|I_h^1(z)-H_r^1(z)|=\mathcal{O}(h),\quad z\in neigh(z_j),\quad j=1,2,\ldots, l_r.
\end{equation}
It is remarked that the estimate in \eqref{eq:pp4} is independent of $neigh(z_j)$, $j=1,\ldots, l_r$.
By \eqref{eq:pp4} and the indicating behaviors of $H_r^1(z)$, one immediately has \eqref{eq:main232} and \eqref{eq:main233}.
\end{proof}

Based on Theorem~\ref{thm:main23}, we propose the following enhanced locating scheme for locating the multiple components  of $\Omega^{(r)}$ in \eqref{eq:rot2}.

\medskip

\hrule

\medskip

\noindent {\bf Algorithm: Locating Scheme AR}
\medskip

\hrule
\medskip

This scheme is the same as Scheme R in Section~\ref{sect:2} with steps~3), 5), 7), respectively modified as

\smallskip

\noindent 3)~Augment the reference space $\mathscr{S}$ to be $\widetilde{\mathscr{S}}_h$ in \eqref{eq:discrete aug1}, and reorder the elements in $\widetilde{\mathscr{S}}_h$ such that
\begin{equation}\label{eq:agu 1assumption2}
\| A(\hat{x};\widetilde{\Sigma}_{j}) \|_{T^2(\mathbb{S}^2)}\geq \| A(\hat x; \widetilde{\Sigma}_{j+1}) \|_{T^2(\mathbb{S}^2)},\quad j=1,2,\ldots,\widetilde{l}_h-1.
\end{equation}

\smallskip

\noindent 5) Replace $I_r^j(z)$ by $I_h^j(z)$.

\smallskip

\noindent 7) Trim all those $z+\widetilde\Sigma_j$ found in Step 6) from $\mathcal{T}_h$.

\medskip
\hrule

\medskip
\medskip

\begin{rem}\label{rem:ar1}
We remark that in Scheme AR, if certain {a priori} information is available about the possible range of the orientations of the scatterer components, it is sufficient for the augmented reference space $\widetilde{\mathscr{S}}_h$ to cover that range only. Clearly, Scheme AR can not only locate the multiple components of $\Omega^{(r)}$ in \eqref{eq:rot2}, but can also recover the orientation of each scatterer component.
\end{rem}

\begin{rem}\label{rem:MediumObstacle2}
Similar to Remark~\ref{rem:MediumObstacle}, our Scheme AR can be extended to include inhomogeneous medium components as long as the relation \eqref{eq:rerell} holds for the reference scatterers in $\widetilde{\mathscr{S}}_h$. Indeed, in our numerical experiments in Section~\ref{sect:5}, we consider the case that the reference scatterers are composed of two inhomogeneous mediums, $(\Sigma_j; \varepsilon_j, \mu_j, \sigma_j),\ j=1,2$ with $\varepsilon_j, \mu_j$ and $\sigma_j$ all constants that are known in advance. For this case, we would like to remark that by following the same argument, Proposition~\ref{prop:rot1} remains the same, which in turn guarantees that Theorem~\ref{thm:main23} remains the same as well. Furthermore, we would like to emphasize that Scheme AR could be straightforwardly extended to work in a much more general setting where there might be both inhomogeneous medium components with variable contents and PEC obstacles presented in the reference space, as long as the generic relation \eqref{eq:rerell} is satisfied.  
\end{rem}

In an analogous manner, for a scatterer described in \eqref{eq:regular scatterer 4}, Scheme AR can be modified that the reference space is augmented by the sizes of 
components to be
\begin{equation}\label{eq:discrete aug2}
\widetilde{\mathscr{S}}_h:=\{\widetilde{\Sigma}_j\}_{j=1}^{\widetilde{l}_{h,m}}
=\cup_{h,m}\{\Pi_{\theta^h,\phi^h,\psi^h}\Lambda_{\tau^m}\Sigma_j\}_{j=1}^{l'},
\end{equation}
where $\tau^m$ is an equal distribution of an interval $[s_1, s_2]$ with its cardinality $N_k$, or some other discrete distribution depending on the availability of certain {a priori} information of relative sizes, and ${\widetilde{l}_{h,m}}=l'\times N_h\times N_m$. Here, $s_1, s_2$ are positive numbers such that $[s_1,s_2]$ contains the scales/sizes of all the scatterer components. With such an augmented reference space, Scheme AR can be used to locate the multiple components and also recover both  orientations and relative sizes of the scatterer $\Omega^{(r)}$ in \eqref{eq:regular scatterer 4}.

\section{Locating multiple multi-scale scatterers}\label{sect:4}

In this section, we shall consider locating a multi-scale scatterer $\Omega^{(m)}$ as described in \eqref{eq:multiscale scatterer} with multiple components. In addition to the requirements imposed on the small component $\Omega^{(s)}$ and the regular-size component $\Omega^{(r)}$ in Section~\ref{sect:2}, we shall further assume that
\begin{equation}\label{eq:dist multiscale}
L_m:=\text{dist}(\Omega^{(s)}, \Omega^{(r)})\gg 1.
\end{equation}
By Lemmas~3.1 and 3.2 in \cite{LiLiuShangSun}, one has, respectively,
\begin{equation}\label{eq:mm1}
A(\hat{x}; \Omega^{(m)}, k)=A(\hat{x}; \Omega^{(s)}, k)+A(\hat{x}; \Omega^{(r)}, k)+\mathcal{O}\left(L_m^{-1}\right)\,,
\end{equation}

\begin{equation}\label{eq:mm2}
A(\hat{x}; \Omega^{(s)}, k)=\mathcal{O}((k\rho)^3).
\end{equation}
That is, if $k\sim 1$, in the far-field pattern $A(\hat{x}; \Omega^{(m)})$, the scattering information from the regular-size component $\Omega^{(r)}$ is dominant and the scattering contribution from the small component $\Omega^{(s)}$ can be taken as small perturbation. Hence, a primitive way to locate the components of $\Omega^{(m)}$ can be proceeded in two stages as follows. First, using the single far-field pattern $A(\hat{x}; \Omega^{(m)})$ as the measurement data, one utilizes Scheme AR to locate the components of the regular-size scatterer $\Omega^{(r)}$. After the recovery of the regular-size scatterer $\Omega^{(r)}$, the far-field pattern from $\Omega^{(r)}$, namely $A(\hat{x}; \Omega^{(r)})$ becomes known. By subtracting $A(\hat{x};\Omega^{(r)})$ from $A(\hat{x};\Omega^{(m)})$, one then has $A(\hat{x}; \Omega^{(s)})$ (approximately). Finally, by applying Scheme S with the far-field data $A(\hat{x};\Omega^{(s)})$, one can then locate the components of $\Omega^{(s)}$. However, if the size contrast between $\Omega^{(r)}$ and $\Omega^{(s)}$ is too big, the scattering information of $\Omega^{(s)}$ will be hidden in the
noisy far-field data of $\Omega^{(r)}$. Hence, in order for the above two-stage scheme to work in locating $\Omega^{(m)}$, the size contrast between $\Omega^{(s)}$ and $\Omega^{(r)}$ cannot be excessively big. But if it is this case, the scattering effect from $\Omega^{(s)}$ would be a significant constituent part to $A(\hat{x};\Omega^{(m)})$, and this will deteriorate the recovery in the first stage and then the second-stage recovery will be deteriorated consequently as well. In order to overcome such a dilemma for this multi-scale locating, we shall develop a subtle local re-sampling technique.

\medskip

\hrule

\medskip

\noindent {\bf Algorithm: Locating Scheme M}
\medskip

\hrule

\begin{enumerate}[1)]

\item Collect a single far-field measurements $A(\hat{x}; \Omega^{(m)}, k)$ corresponding to the multi-scale scatterers $\Omega^{(m)}$.

\item Select a sampling region with a mesh $\mathcal{T}_h$ containing $\Omega^{(m)}$.

\item Suppose that
\[
\Omega^{(r)}=\bigcup_{j=1}^{l_r} ( z_{j}+\widetilde{\Sigma}_j ),\quad \widetilde{\Sigma}_j\in \widetilde{\mathscr{S}}_h,
\]
as described in \eqref{eq:regular scatterer 4} of Section~\ref{sect:3}.
Using $A(\hat{x}; \Omega^{(m)}, k)$ as the measurement data, one locates the rough locations  $\widetilde{z}_j\in\mathcal{T}_h$, $j=1,2,\ldots, l_r$, shapes 
and orientations of each scatterer component following Scheme AR. Here $\widetilde{z}_j$, $j=1,2,\ldots, l_r$, are the approximate position points to the exact ones $z_j$, $j=1,2,\ldots, l_r$.

\item Apply the \emph{local re-sampling technique} following the next sub-steps to update $\widetilde{z}_j$'s and to locate the components of the small-size scatterer $\Omega^{(s)}$.

\begin{enumerate}
\item[a)] For each point $\widetilde{z}_j$ found in Step~3), one generates a finer local mesh $\mathcal{Q}_{h'}(\widetilde{z}_j)$ around $\widetilde{z}_j$.

\item[b)] For one set of sampling points, $\hat{z}_j\in\mathcal{Q}_{h'}(\widetilde{z}_j)$, $j=1,2,\ldots, l_r$, one calculates
\begin{equation}\label{eq:re1}
\widetilde{A}(\hat{x}; k)
=A(\hat{x}; \Omega^{(m)}, k)-\sum_{j=1}^{l_r} e^{ik(d-\hat{x})\cdot \hat{z}_j} A(\hat{x}; \widetilde{\Sigma}_j, k).
\end{equation}

\item[c)] Using $\widetilde{A}(\hat{x}; k)$ in Step b) as the measurement data, one applies Scheme S to locate the significant local maximum points on $\mathcal{T}_h\backslash\cup_{j=1}^{l_r} \mathcal{Q}_{h'}(\widetilde{z}_j)$ of the corresponding indicator function.

\item[d)] Repeat Steps b) and c) by all the possible sets of sampling points from $\mathcal{Q}_{h'}(\widetilde{z}_j)$, $j=1,2,\ldots l_r$. The clustered local maximum points on $\mathcal{T}_h\backslash\cup_{j=1}^{l_r} \mathcal{Q}_{h'}(\widetilde{z}_j)$ are the positions corresponding to the scatterer components of $\Omega^{(s)}$.

\item[e)] One updates the $\widetilde{z}_j$'s to be those sampling points $\hat{z}_j$'s which generate the clustered local maximum points in Step d).

\end{enumerate}

\end{enumerate}

\hrule

\medskip
\medskip

We note that in \eqref{eq:re1}, if the re-sampling points $\hat{z}_j$'s are the exact position points, namely $\hat{z}_j=z_j$, $j=1,2,\ldots, l_r$, then
\[
\sum_{j=1}^{l_r} e^{ik(d-\hat{x})\cdot \hat{z}_j} A(\hat{x}; \widetilde{\Sigma}_j, k)=A(\hat{x}; \Omega^{(r)}, k).
\]
This, together with \eqref{eq:mm1}, implies that $\widetilde{A}(\hat{x}; k)$ calculated according to \eqref{eq:re1} is an approximation to $A(\hat{x}; \Omega^{(s)}, k)$.

Next, we propose an enhanced Scheme M by making use of two far-field measurements which could provide a more robust and accurate locating of the multi-scale scatterers $\Omega^{(m)}$. Indeed, we assume that in $\Omega^{(m)}$, the diameters of the multiple components of $\Omega^{(r)}$ are around $d_1$, whereas the diameters of the multiple components of $\Omega^{(s)}$ are around $d_2$ such that $d_1/d_2$ is relatively large. We choose two wave numbers $k_1$ and $k_2$ such that for $\lambda_1=2\pi/k_1$ and $\lambda_2=2\pi/k_2$, $\lambda_1>d_1$ with $\lambda_1\sim d_1$, and $d_2<\lambda_2<d_1$ with $\lambda_2/d_2$ relatively large. Then, in $A(\hat{x}; \Omega^{(m)}, k_1)$, according to \eqref{eq:mm1} and \eqref{eq:mm2}, $A(\hat{x}; \Omega^{(r)}, k_1)$ is more significant and this will enable Scheme AR to have a more accurate locating of $\Omega^{(r)}$. On the other hand, according to \eqref{eq:mm2}, $A(\hat{x}; \Omega^{(m)}, k_2)$ clearly carries more scattering information of $\Omega^{(s)}$ than that in $A(\hat{x};\Omega^{(m)}, k_1)$. Hence, after the locating of $\Omega^{(r)}$ by using $A(\hat{x}; \Omega^{(m)}, k_1)$, one can use $A(\hat{x}; \Omega^{(m)}, k_2)$ as the measurement data for the second stage in Scheme M to yield a more accurate reconstruction of $\Omega^{(s)}$. In summary, the enhanced Scheme M by making use of two far-field measurements can be formulated as follows.

\medskip

\hrule

\medskip

\noindent {\bf Algorithm:  Enhanced Locating Scheme M}
\medskip

\hrule

\begin{enumerate}[1)]

\item Collect two far-field measurements $A(\hat{x}; \Omega^{(m)}, k_1)$ and $A(\hat{x}; \Omega^{(m)}, k_2)$ corresponding to the multi-scale scatterer $\Omega^{(m)}$.

\item Use $A(\hat{x}; \Omega^{(m)}, k_1)$ as the measurement data for the first stage in Scheme M, namely Steps 2) and 3).

\item Use $A(\hat{x}; \Omega^{(m)}, k_2)$ as the measurement data for the second stage in Scheme M, namely Step 4).

\item Apply the local re-sampling technique following the next sub-steps of Step 4)  in  Scheme M to update $\widetilde{z}_j$'s and to locate the components of the small-size scatterer $\Omega^{(s)}$. Particularly, \eqref{eq:re1} is modified to be

\begin{equation}\label{eq:re1n}
\widetilde{A}(\hat{x}; k_2)
=A(\hat{x}; \Omega^{(m)}, k_2)-\sum_{j=1}^{l_r} e^{ik_2(d-\hat{x})\cdot \hat{z}_j} A(\hat{x}; \widetilde{\Sigma}_j, k_2).
\end{equation}

\end{enumerate}

\hrule

\medskip
\medskip

\section{Numerical experiments and discussions}\label{sect:5}

In this section, we present some numerical results to illustrate salient features of our new schemes using augmented far field data set  as well as its ability to image multiple multi-scale scatterers by the novel Scheme M with the local re-sampling technique.

Three geometries will be considered for the scatterer components in our numerical experiments. They are given by revolving bodies through rotating the following 2D shapes in the $x$-$y$ plane around the $x$-axis
\begin{eqnarray*}
\mathbf{Circle:} &\quad & \{ (x,y) : x=\cos(s), \  y=\sin(s), \  0\le s\le 2\pi \},\\
\mathbf{Peanut:} &\quad & \{ (x,y) : x=\sqrt{3 \cos^2 (s) + 1}\cos(s), \  y=\sqrt{3 \cos^2 (s) + 1}\sin(s), \  0\le s\le 2\pi \},\\
\mathbf{Kite:} &\quad & \{ (x,y) : x=\cos(s)+0.65\cos (2s)-0.65, \  y=1.5\sin(s), \  0\le s\le 2\pi \}.
\end{eqnarray*}
In the sequel, they are denoted by \textbf{B}, \textbf{P} and \textbf{K}, respectively, for short.
The candidate  data set $\widetilde{\mathscr{S}}_h$  includes  far field data of all three reference components \textbf{B}, \textbf{P} and \textbf{K}, and is further lexicographically augmented by  a collection of a priori known orientations and sizes. More precisely, the augmented data set is obtained by rotating \textbf{P} and \textbf{K} in the  $x$-$y$ plane every $\pi/4$ radian\footnote{There are only four different orientations for \textbf{P} due to its  symmetry.}  as shown in Figs.~\ref{fig:Scatterer-shape-peanut} and \ref{fig:Scatterer-shape-kite}, respectively,  and by scaling  \textbf{B}, \textbf{P} and \textbf{K} by one fifth, one half, one, twice and five times.
\begin{figure}[b]
\hfill{}\includegraphics[width=0.23\textwidth]{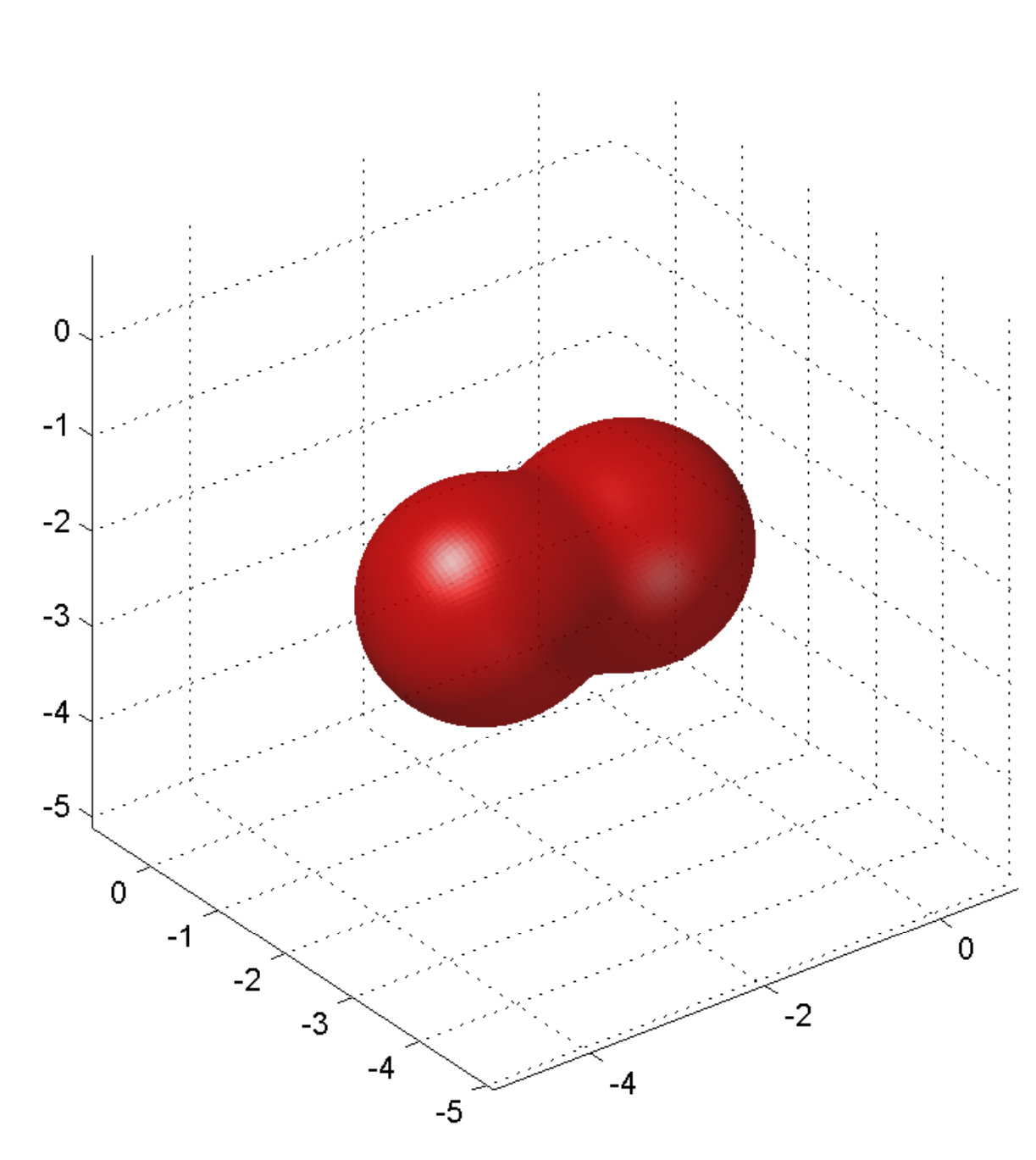}\hfill{}\includegraphics[width=0.23\textwidth]{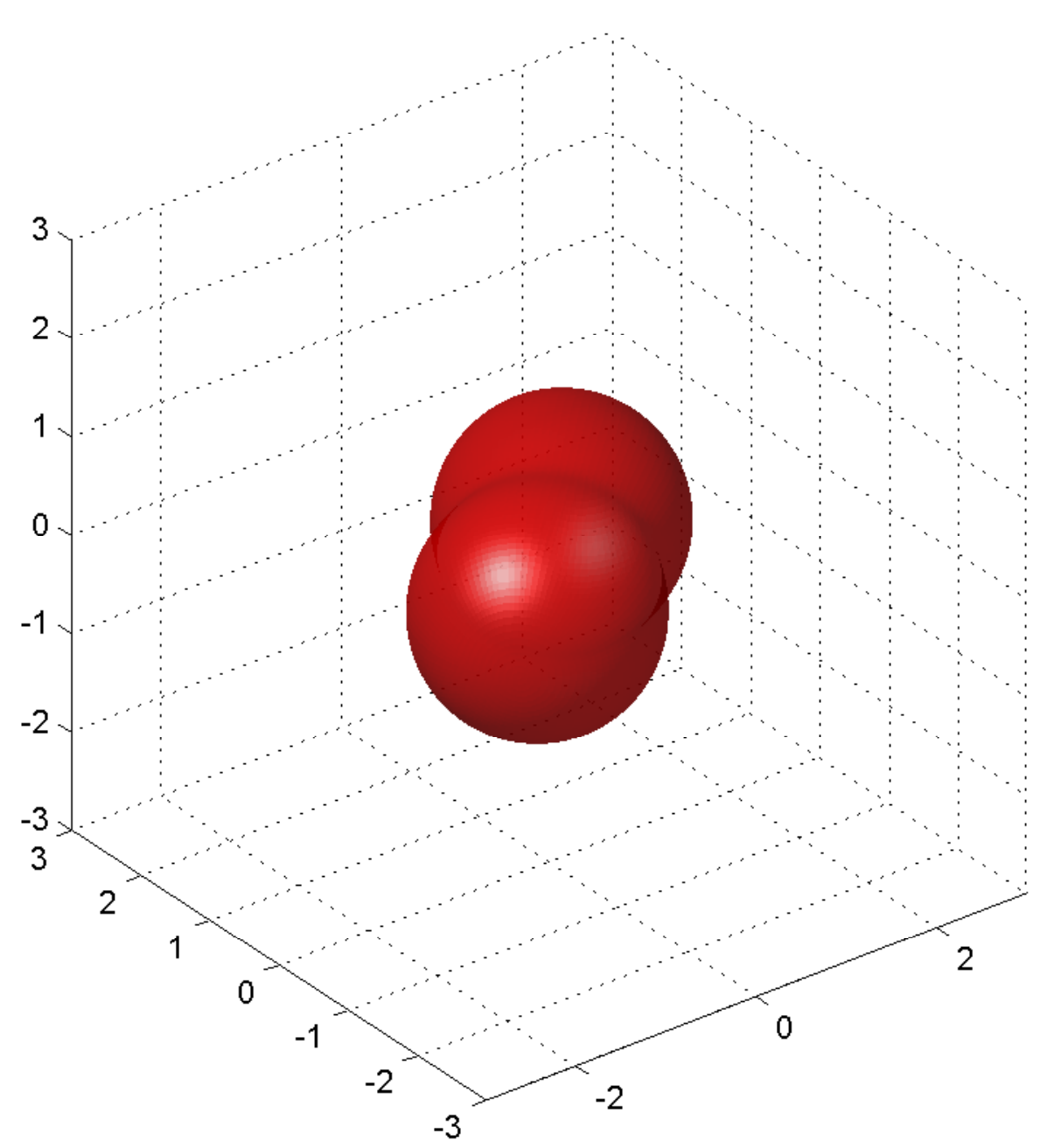}\hfill{}\includegraphics[width=0.23\textwidth]{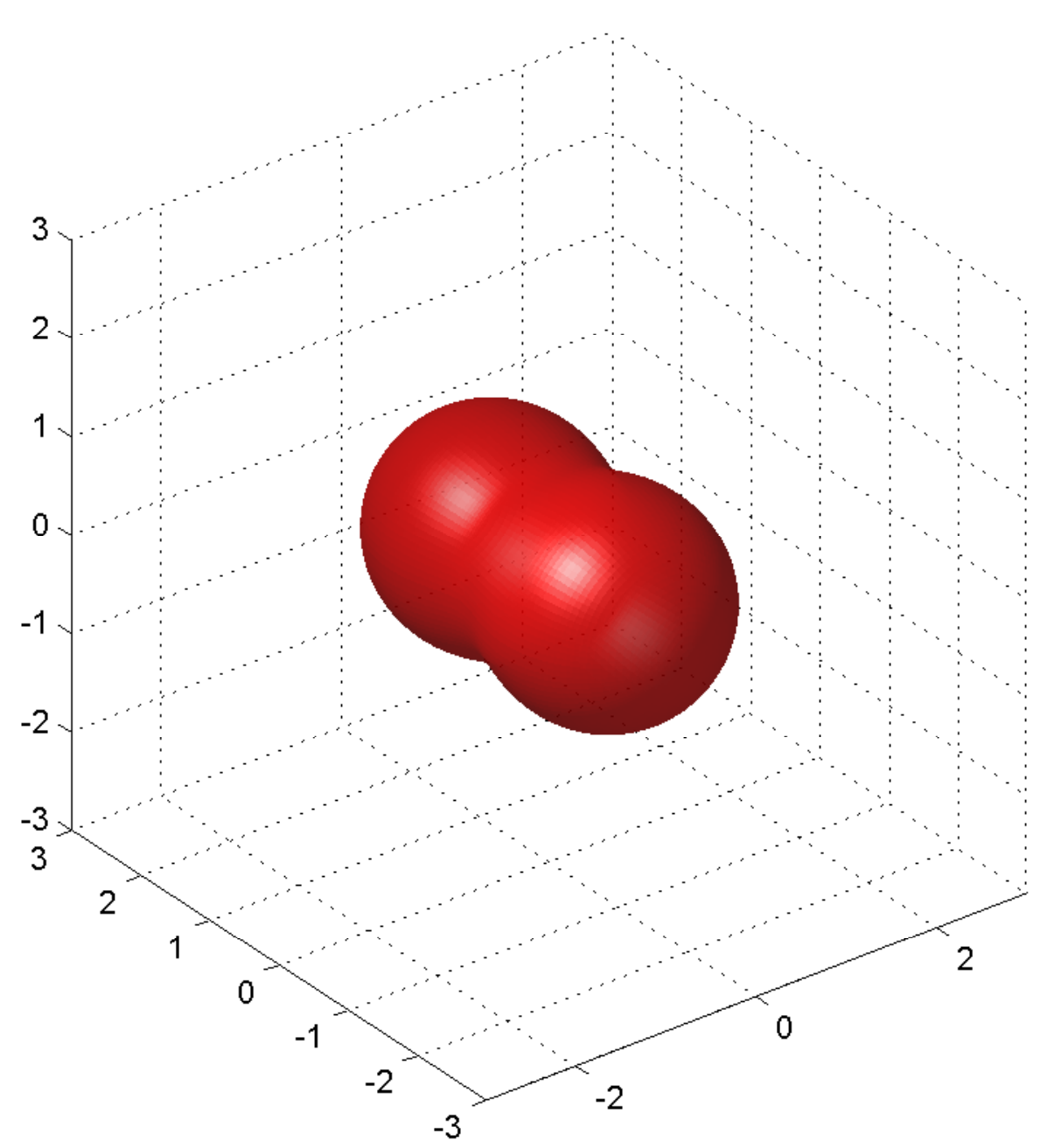}\hfill{}\includegraphics[width=0.23\textwidth]{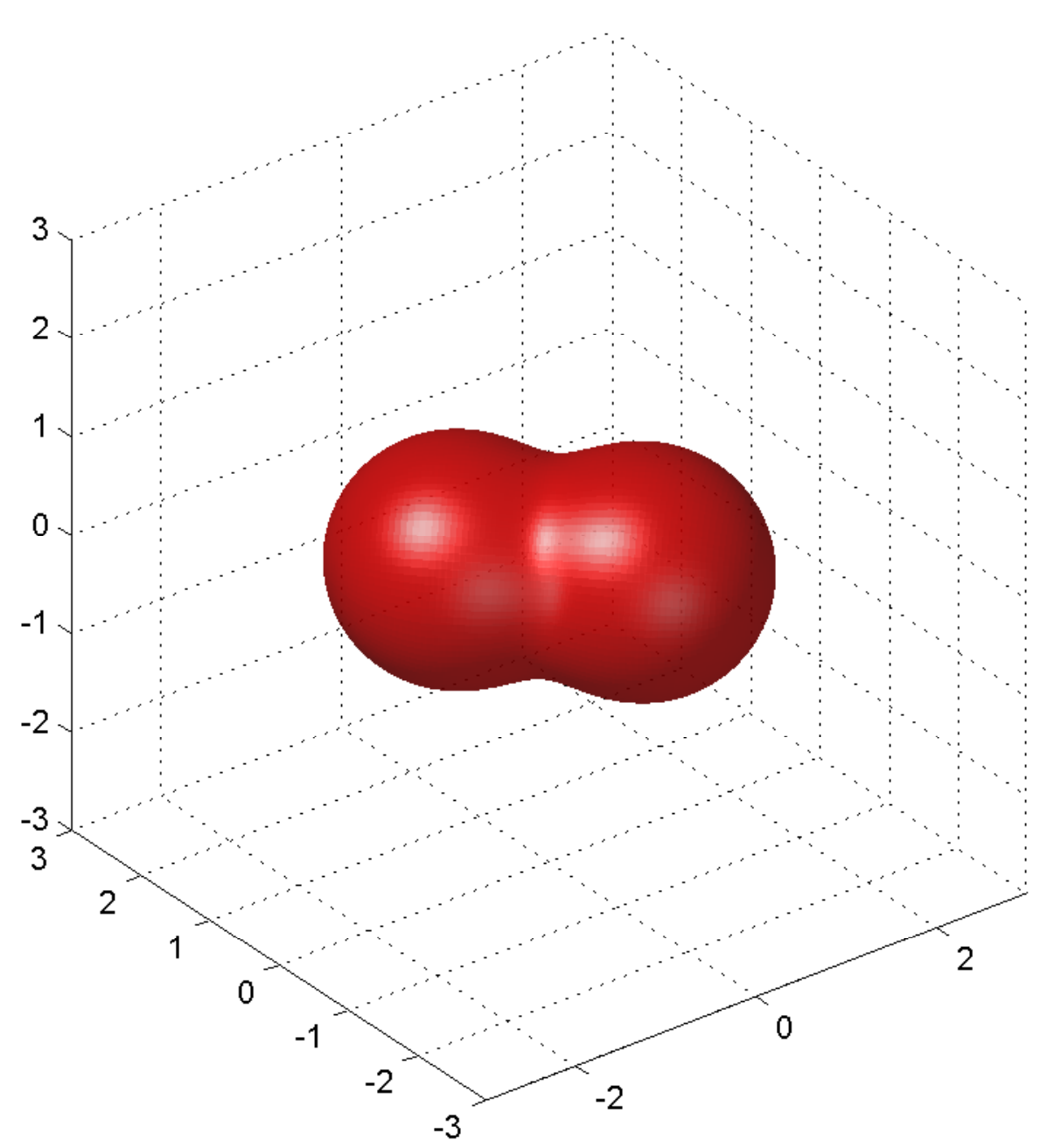}\hfill{}

\caption{\label{fig:Scatterer-shape-peanut}Scatterer component Peanut with four orientations.}
\end{figure}

\begin{figure}[pt]
\hfill{}\includegraphics[width=0.23\textwidth]{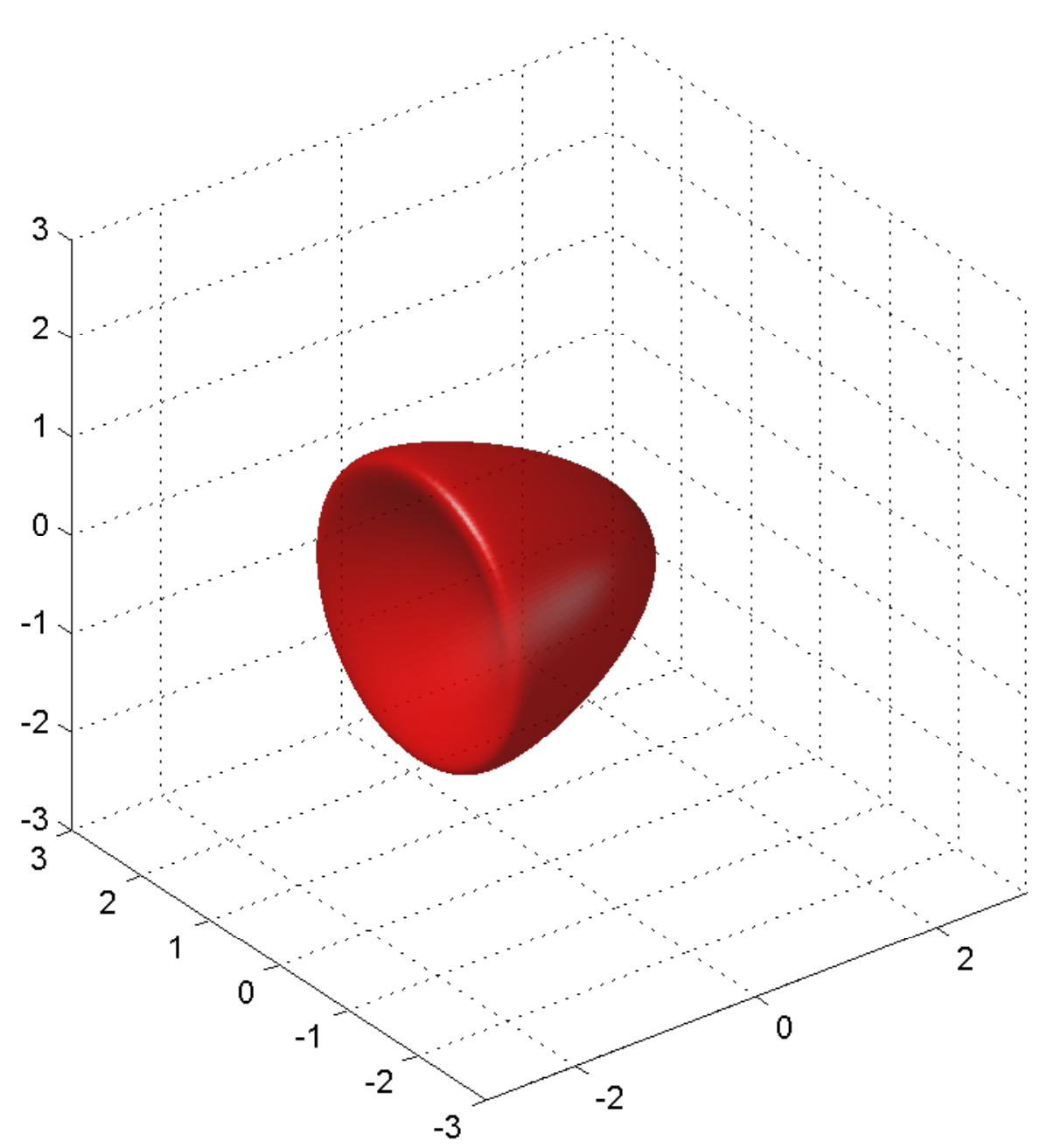}\hfill{}\includegraphics[width=0.23\textwidth]{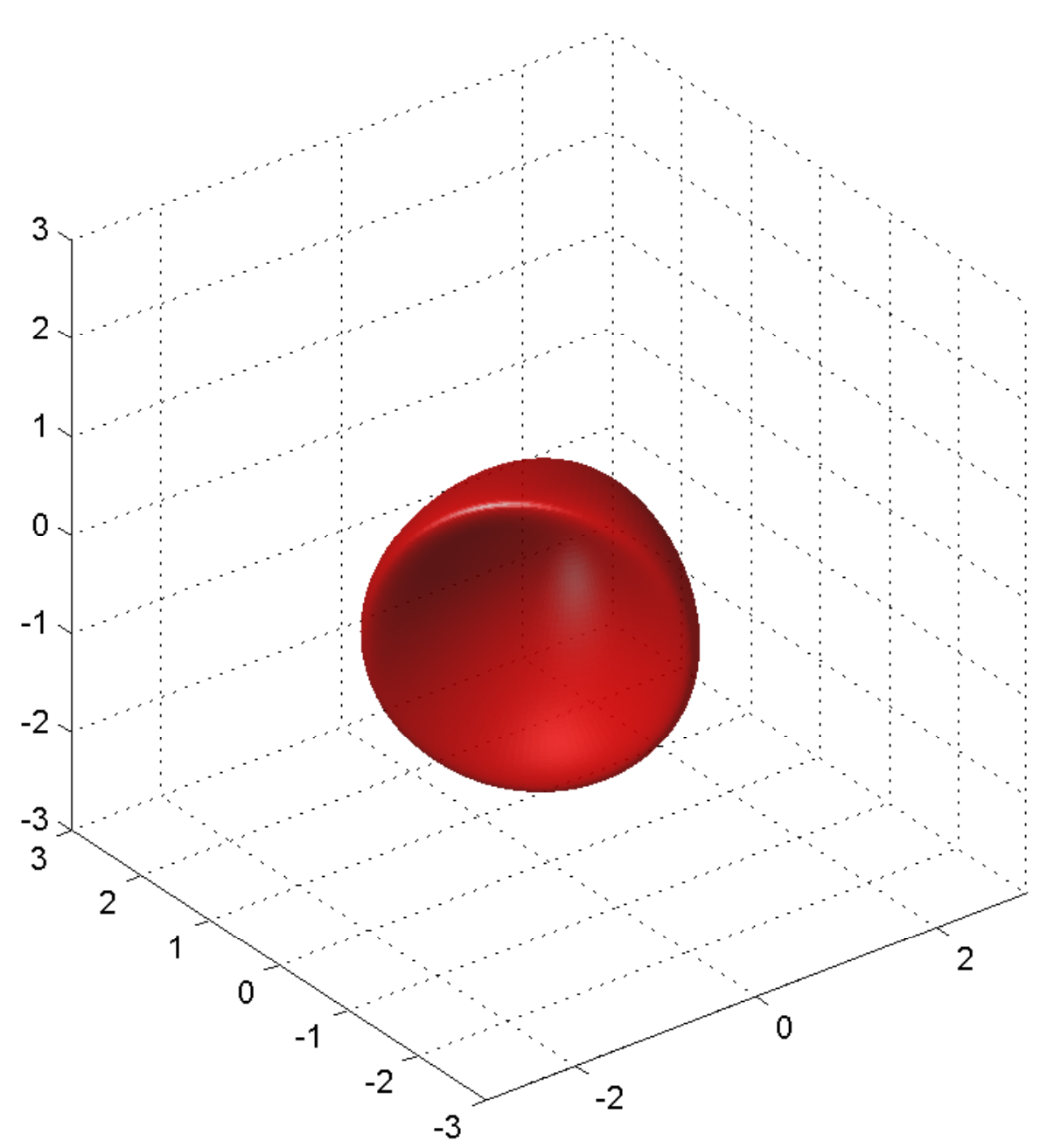}\hfill{}\includegraphics[width=0.23\textwidth]{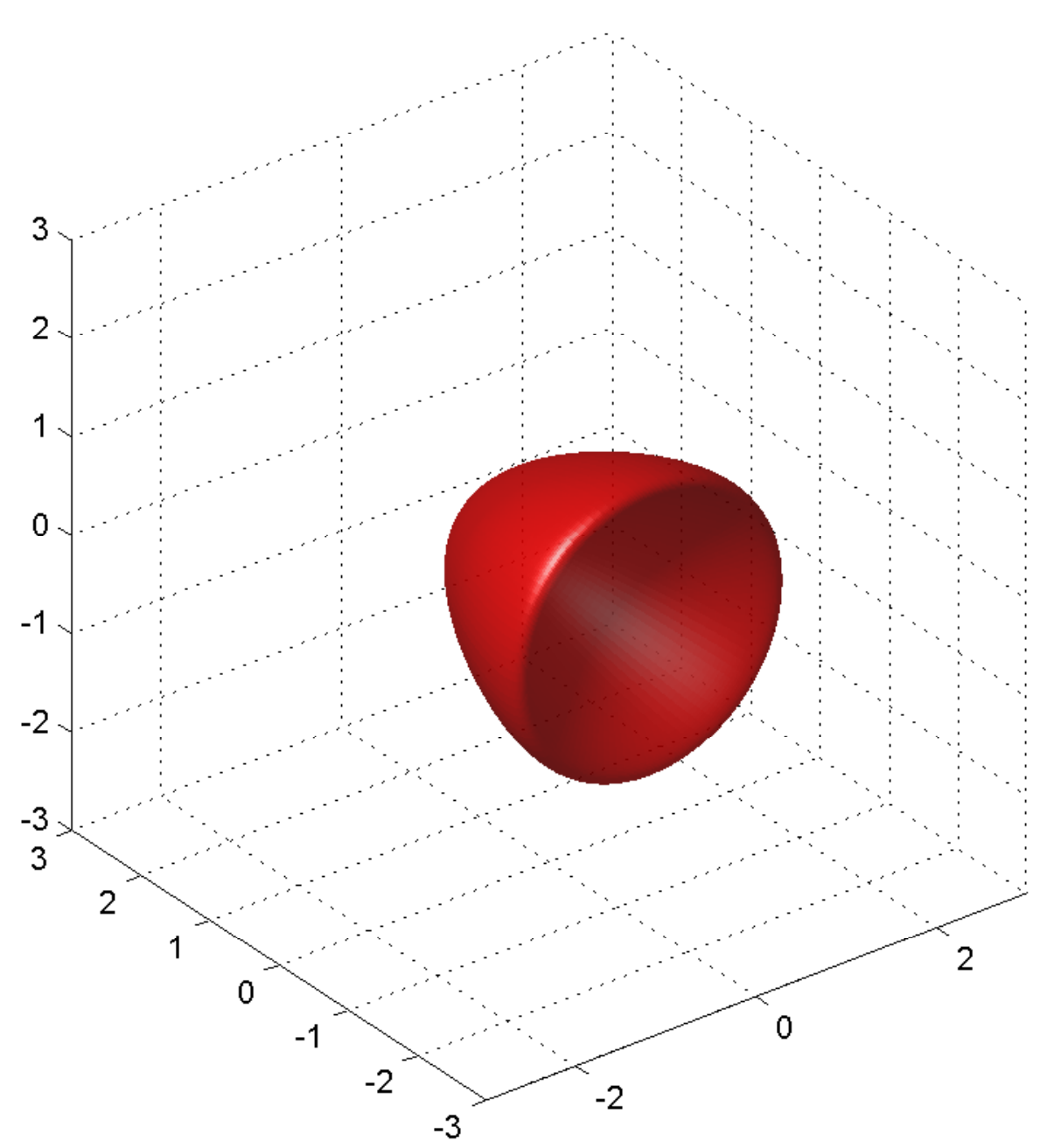}\hfill{}\includegraphics[width=0.23\textwidth]{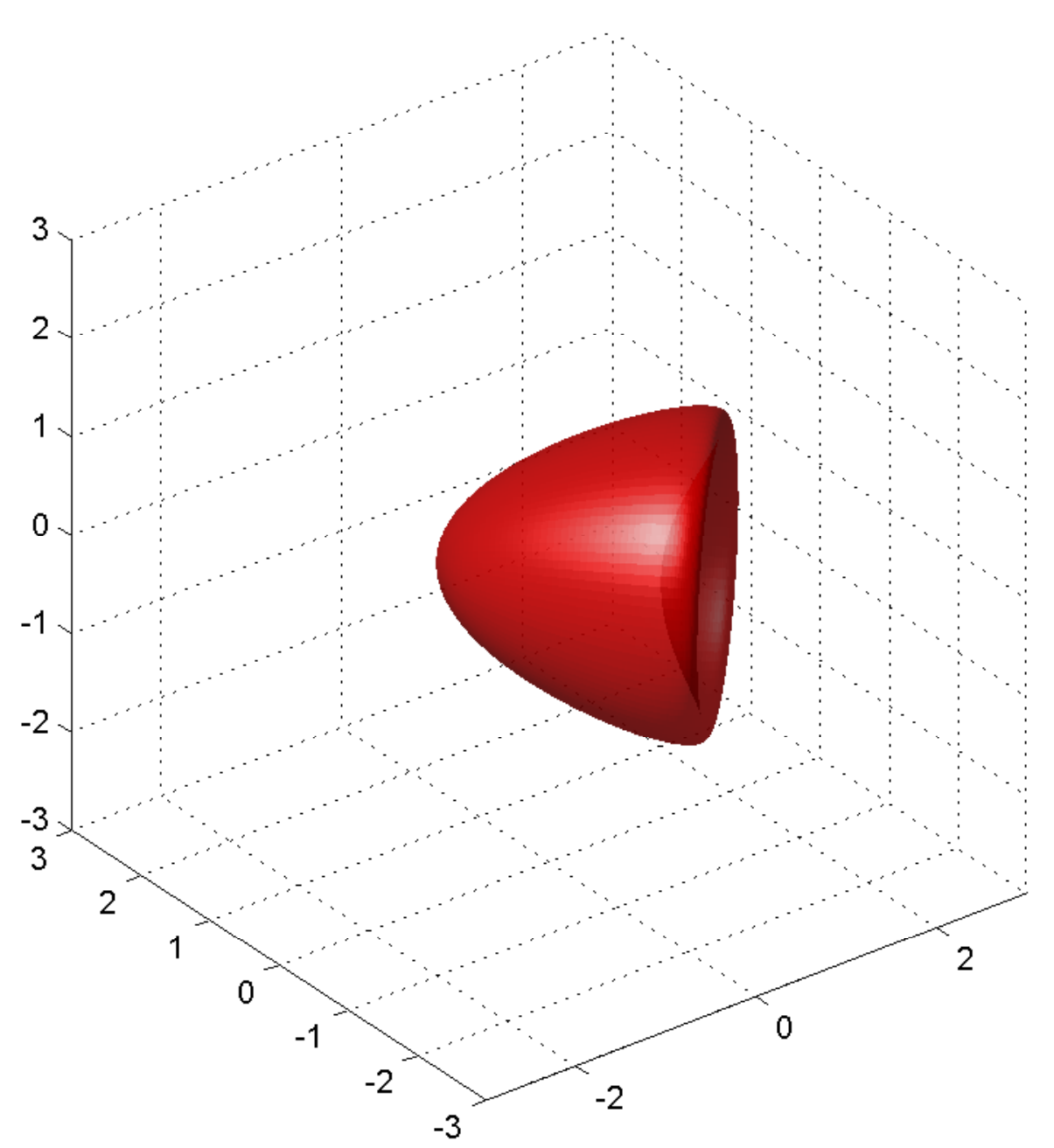}\hfill{}

\hfill{}\includegraphics[width=0.23\textwidth]{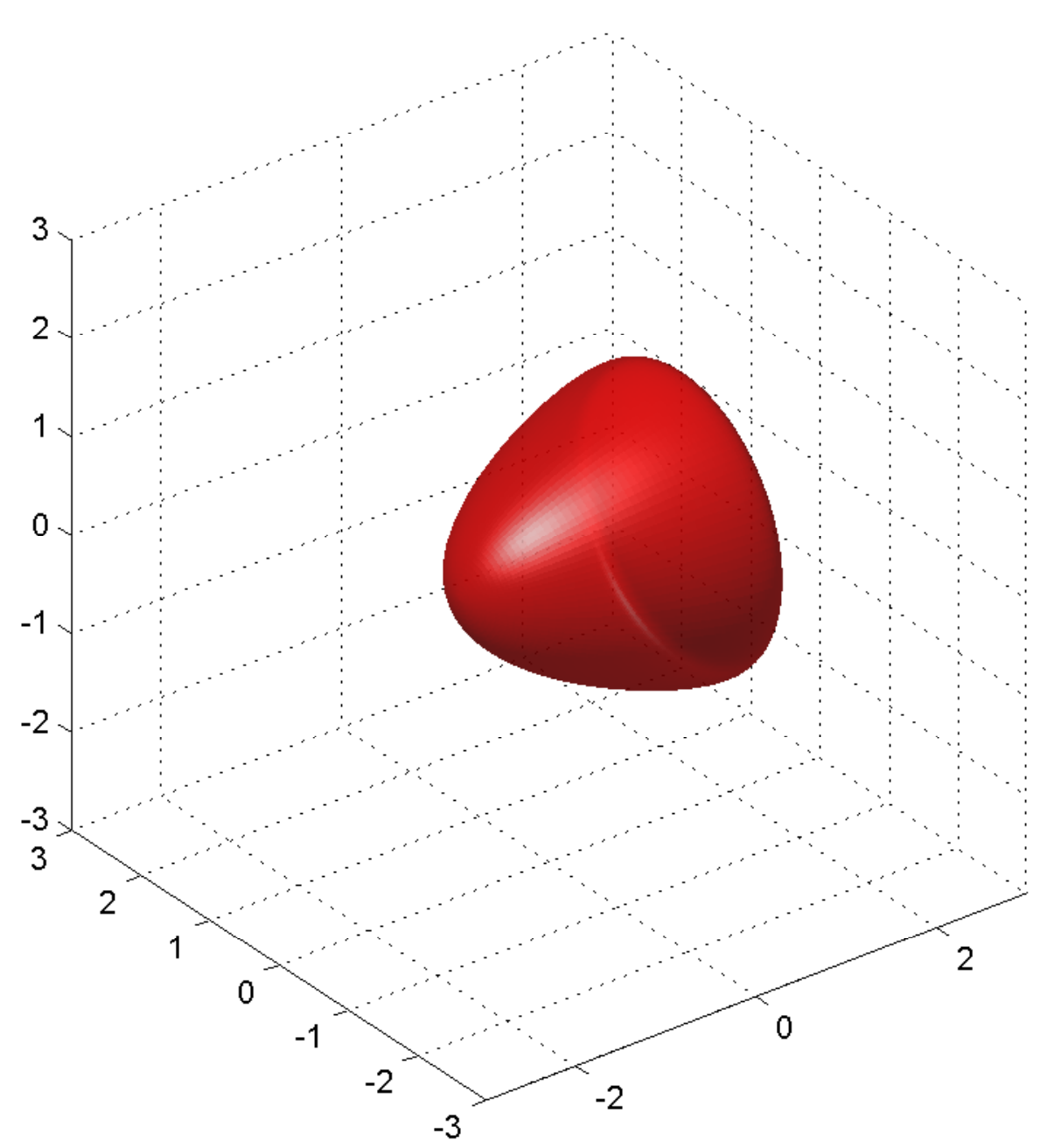}\hfill{}\includegraphics[width=0.23\textwidth]{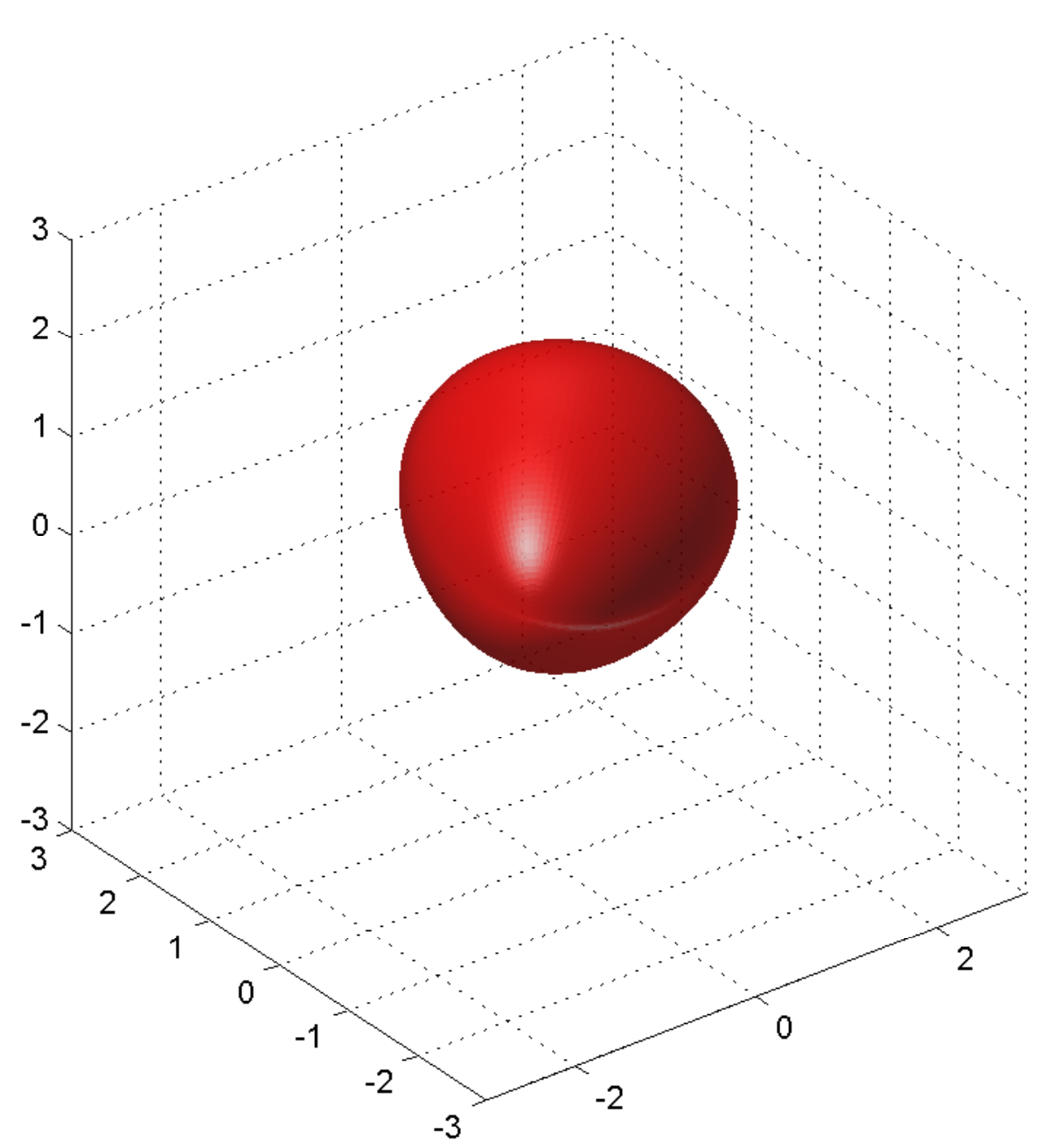}\hfill{}\includegraphics[width=0.23\textwidth]{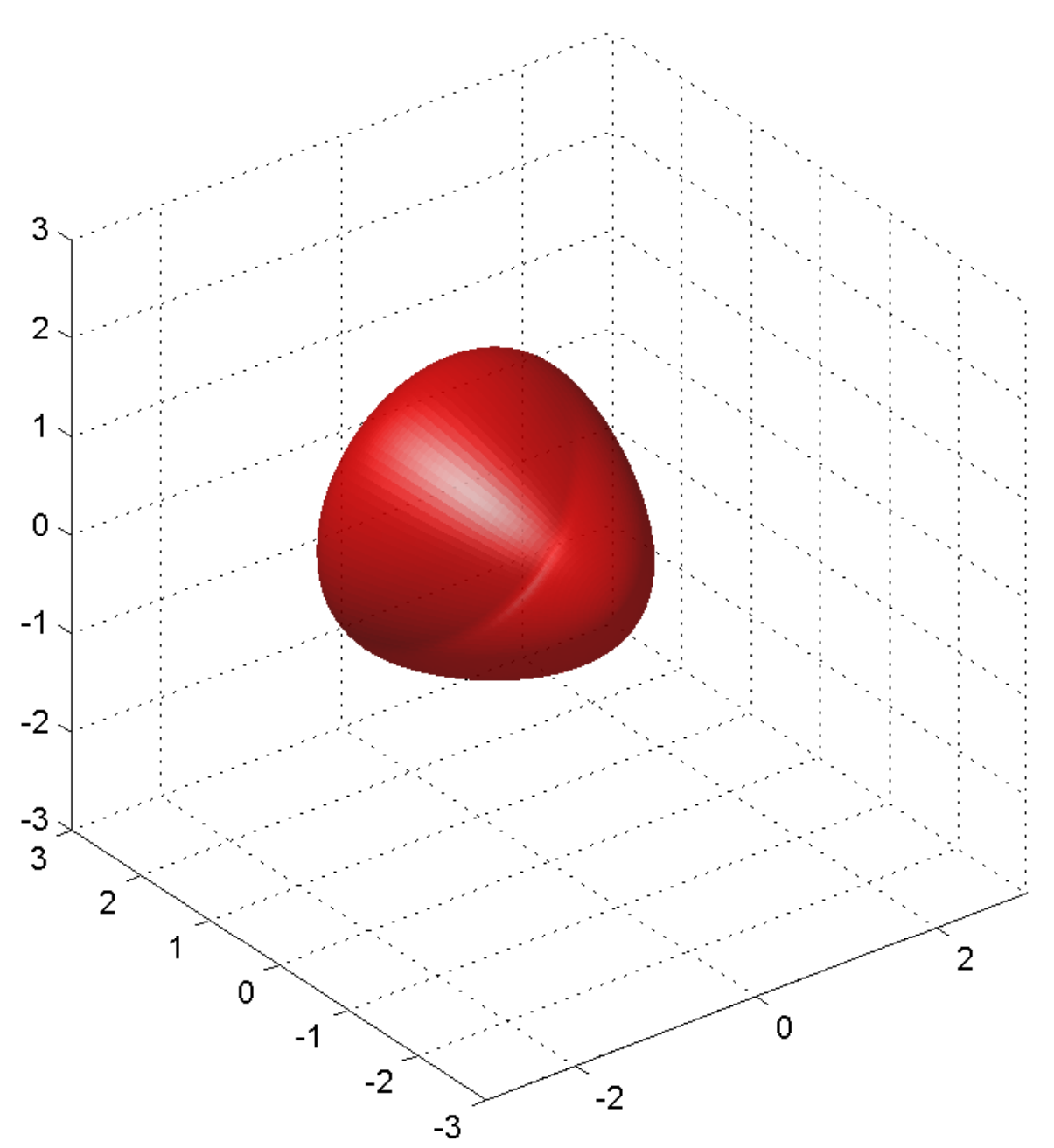}\hfill{}\includegraphics[width=0.23\textwidth]{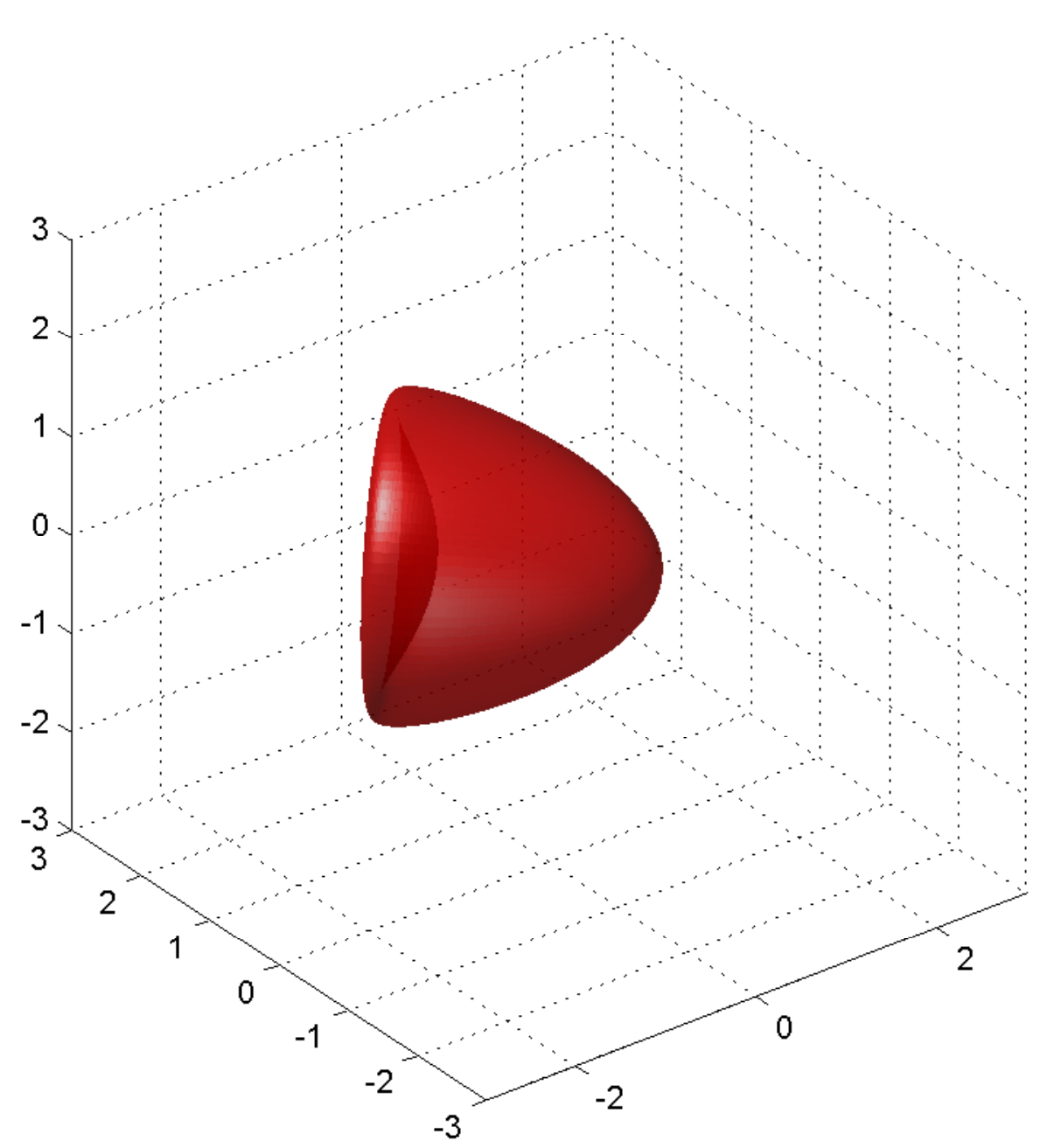}\hfill{}

\caption{\label{fig:Scatterer-shape-kite}Scatterer component Kite with eight orientations.}
\end{figure}

In  the examples below, as assumed earlier,
we set $\varepsilon_0=\mu_0=1$ and $\sigma_0=0$ outside
the scatterer, and hence the wavelength is unitary in the homogeneous background.  Unless otherwise specified, all
the scatterer components are  either PEC conductors or
inhomogeneous media with all other parameters the same as those in the homogeneous
background except $\varepsilon=4$. Our near-field data are  obtained by solving the Maxwell system \eqref{eq:Maxwell general} using the quadratic
$H(\mathrm{curl})$-conforming edge element discretization in a spherical domain centered at the origin and holding inside all the scatterer components.  The computational domain is enclosed
by a PML layer to damp the reflection. Local adaptive
refinement scheme within the inhomogeneous scatterer is adopted to
enhance the resolution of the scattered wave.  The far-field data
are approximated by the integral equation representation \cite[p.~181, Theorem~3.1]{PiS02}  using the spherical Lebedev quadrature (cf.\!\cite{Leb99}).
We refine the mesh successively
till the relative maximum error of successive groups of far-field
data is below $0.1\%$. The far-field patterns on the finest mesh
are used as the exact data. The electric far-field patterns $A(\hat{x},\Omega)$, $\Omega = \Omega^{(r)}$ or $\Omega^{(m)}$, are observed at 590
Lebedev quadrature points distributed on the unit sphere
$\mathbb{S}^{2}$ (cf.\!\cite{Leb99} and references therein). The exact far-field data
$A(\hat{x},\Omega)$ are corrupted point-wise by the formula
\begin{equation}
A_{\delta}(\hat{x},\Omega)=A(\hat{x},\Omega)+\delta\zeta_1\underset{\hat{x}}{\max}|A(\hat{x},\Omega)|\exp(i2\pi
\zeta_2)\,,
\end{equation}
where $\delta$ refers to the relative noise level, and both  $\zeta_1$ and $\zeta_2$ follow the
uniform distribution ranging from $-1$ to $1$. The values of the indicator functions have  been
normalized between $0$ and $1$ to highlight the positions
identified.

Some experimental settings are defined as follows. In our tests,
we shall always take the incident direction $d=(1,0,0)^T$ and the polarization $p=(0,0,1)^T$.
In all  our tests, the noise level is  $3\%$.
To improve the accuracy and robustness of imaging results using Scheme AR and Enhanced Scheme M,
we adopt two full augmented data sets associated with two detecting EM waves with two proper wave numbers,
which will be clearly specified later.

Two inverse scattering benchmark problems are considered here.
The first one \textbf{PK} is to image two regular-size scatterer components with  kite- and peanut-shape, respectively. In this case, we reconstruct the scatterer components with correct orientations and sizes by the augmented data set using Scheme AR. The second example \textbf{KB} is to image a combined  scatterer consisting of multiple multi-scale components, an enlarged kite of \textbf{K} by two times and a relatively small ball of \textbf{B} scaled to one half from the unit one.
The size ratio between the two components is about six.

\subsection{Scheme AR}

\paragraph*{Example PK.}

In this example, we try to locate with Scheme AR  a kite component \textbf{K}  located at $(2,2,2)$ with azimuthal angle $\pi/4$ radian,
and a peanut component \textbf{P}  located at $(-2,-2,-2)$ with azimuthal angle $3\pi/4$ radian as shown as in Fig.~\ref{fig:True-scatterer}(a)
and  its projection on the $x-y$, $y-z$ and $z-x$ planes shown in Fig \ref{fig:True-scatterer}(b)-(d), respectively.

 \begin{figure}[bp]
\hfill{}\includegraphics[width=0.23\textwidth]{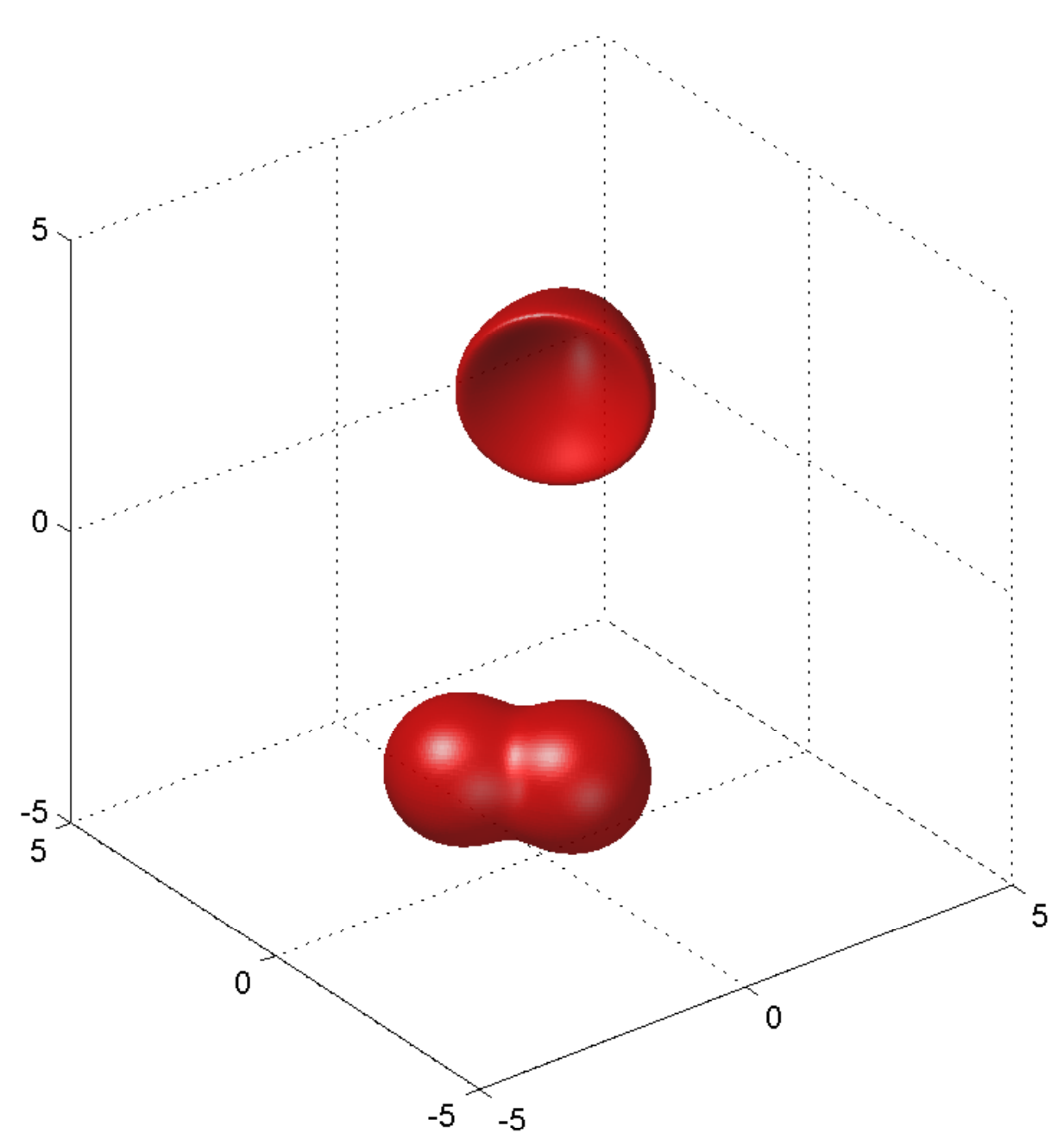}\hfill{}\includegraphics[width=0.23\textwidth]{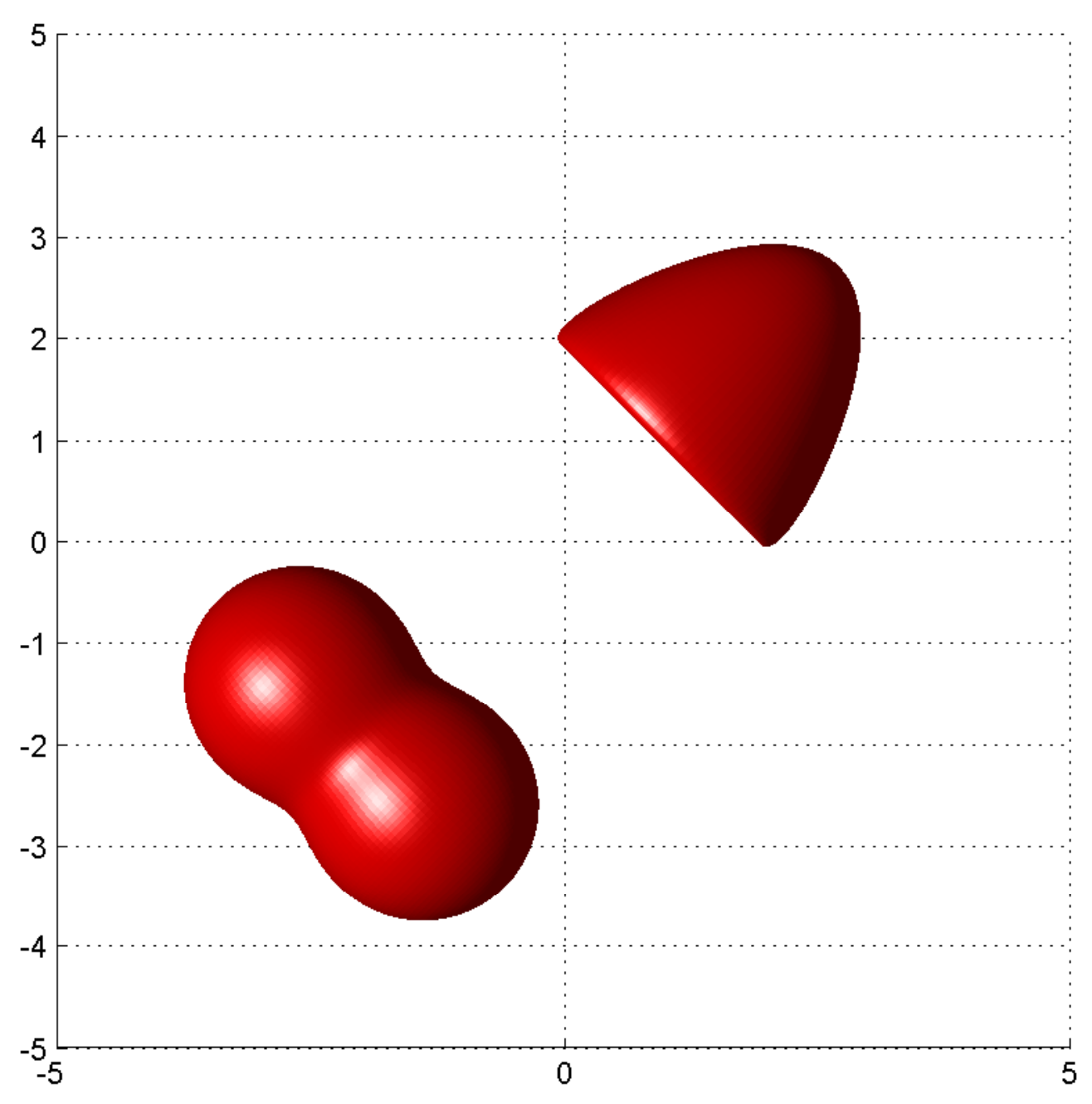}\hfill{}\includegraphics[width=0.23\textwidth]{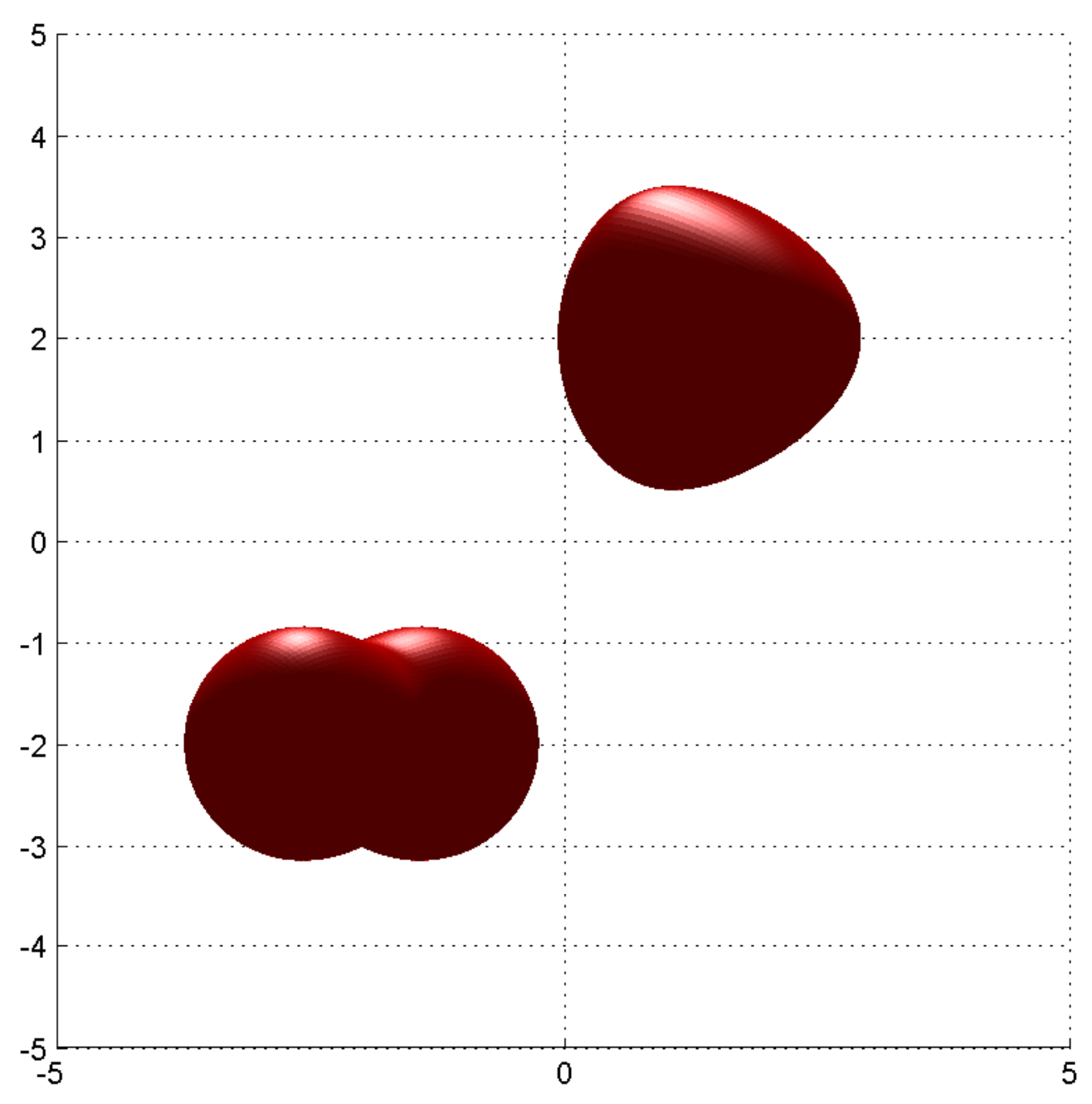}\hfill{}\includegraphics[width=0.23\textwidth]{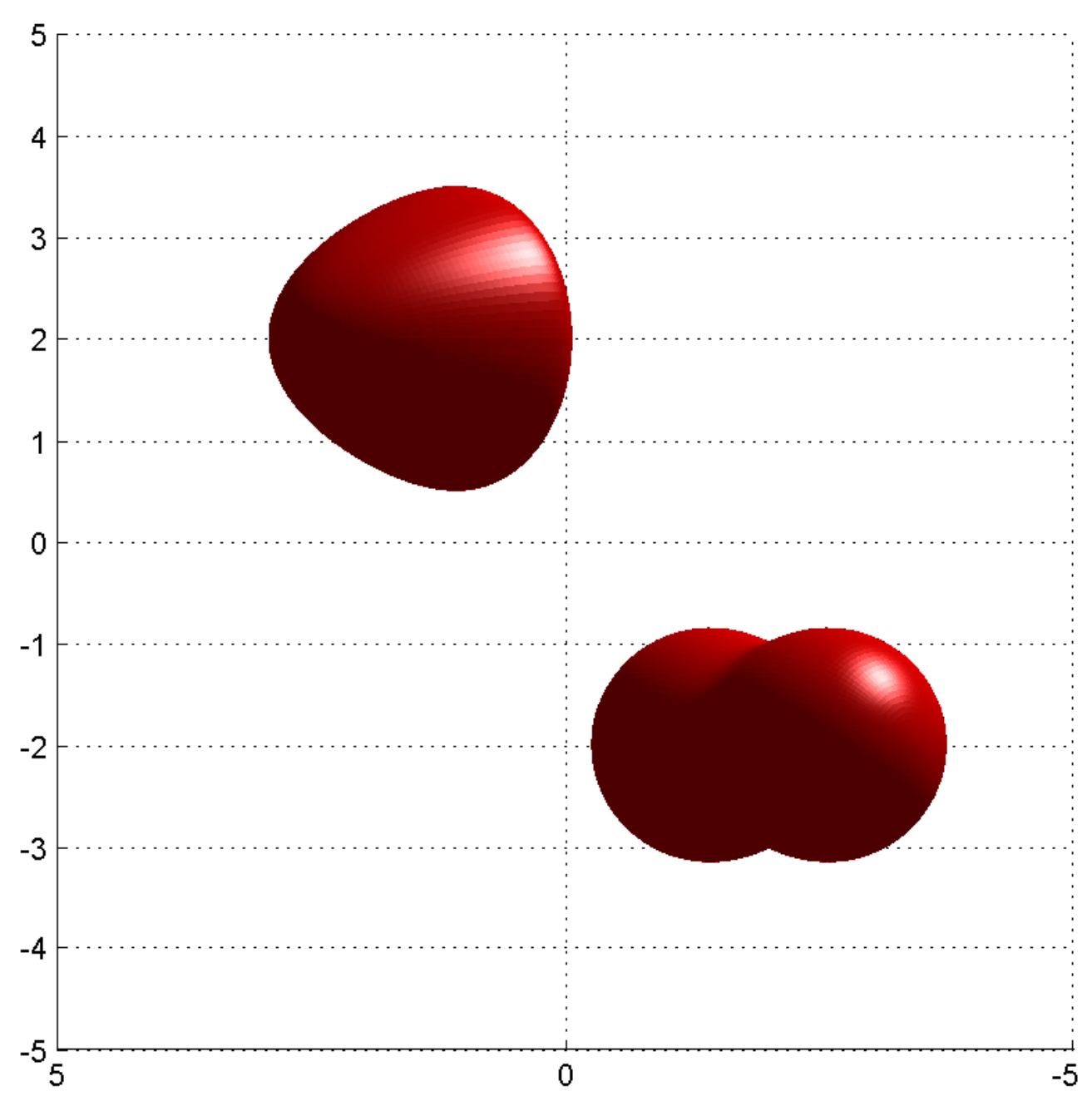}\hfill{}

\hfill{}(a)\hfill{}\hfill{}(b)\hfill{}\hfill{}(c)\hfill{}\hfill{}(d)\hfill{}
\caption{\label{fig:True-scatterer}True scatterer for Example \textbf{PK}. }
\end{figure}

As remarked earlier, we choose the two scatterer components to be inhomogeneous media. There are two considerations for such a choice. First, we developed Scheme AR in Section~\ref{sect:3} mainly for locating PEC obstacles, but we also gave the extension to locate medium components if the generic situation described in Remark~\ref{rem:MediumObstacle2} is fulfilled. Second, we would like to illustrate the wide applicability of Scheme AR, and we refer to \cite{LiLiuShangSun} for numerical results on recovering multiple PEC obstacles by Scheme R. We implement Scheme AR in a two-stage imaging procedure as follows:

\paragraph{Scheme S.}

We first set $k=1$, which amounts to sending a detecting EM wave of wavelength at least twice larger than each component of the scatterer. With the collected far-field data, we implement Scheme S to find how many components to be recovered and locate the rough positions of those scatterer components.

The imaging result at this coarse stage is shown in Fig.~\ref{fig:Final-image-of-Coarse}, indicated by
the characteristic behavior of the function $I_s^j(z)$ (cf.~\eqref{eq:indicator function}) in Scheme S.
Note that no reference spaces are needed up to this stage.
It can be observed that the indicator function achieves local maxima in the region where
there exists a scatterer component, either kite or peanut. The rough position of the peanut is highlighted in Figs.~\ref{fig:Final-image-of-Coarse}(a)
which indicate a possible scatterer component somewhere around the highlighted region.  In  Figs.~\ref{fig:Final-image-of-Coarse}(b), we see that the rough position of the kite could also be found.
But its dimer brightness as shown in  Figs.~\ref{fig:Final-image-of-Coarse}(b) tells us that one cannot figure out its shape and size up to this stage.

Then we could incorporate the suspicious regions
into a stack of cubes, as in Figs.~\ref{fig:Final-image-of-Coarse}(c)
and (d). And the computation of the next stage, i.e., Scheme
AR, is just performed on these cubes, which are shown exclusively in Figs.~\ref{fig:Final-image-of-Coarse}(e)
and (f). It is emphasized that this preprocessing stage can be skipped and one can directly implement the Scheme AR as described in the next stage to locate the kite \textbf{K} and the peanut \textbf{P}. However, by performing this preprocessing stage, the computational costs can be significantly reduced, and the robustness and resolution can be enhanced for Scheme AR, as will be performed in the next stage. 

\begin{figure}
\hfill{}\includegraphics[width=0.4\textwidth]{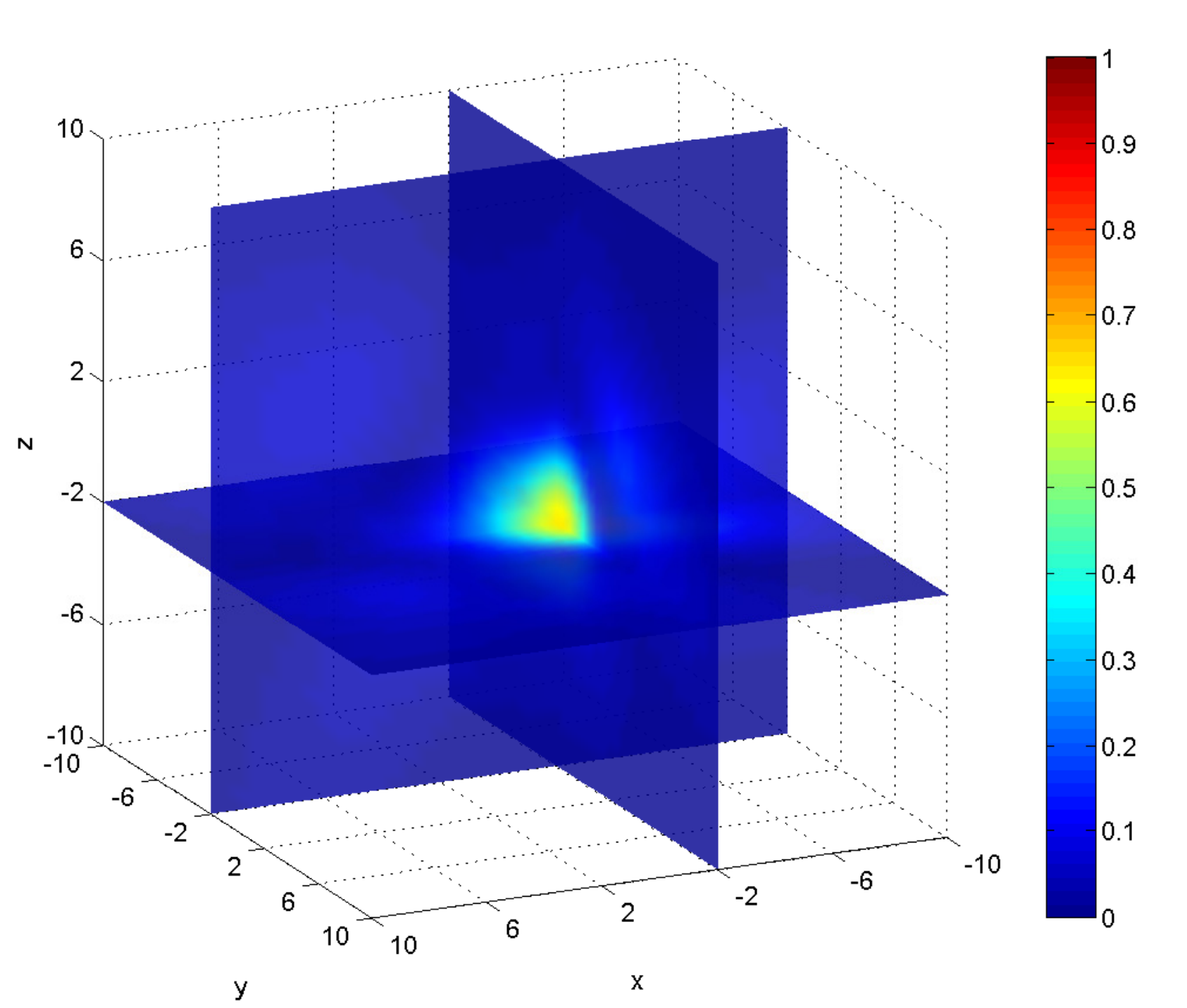}\hfill{}\includegraphics[width=0.4\textwidth]{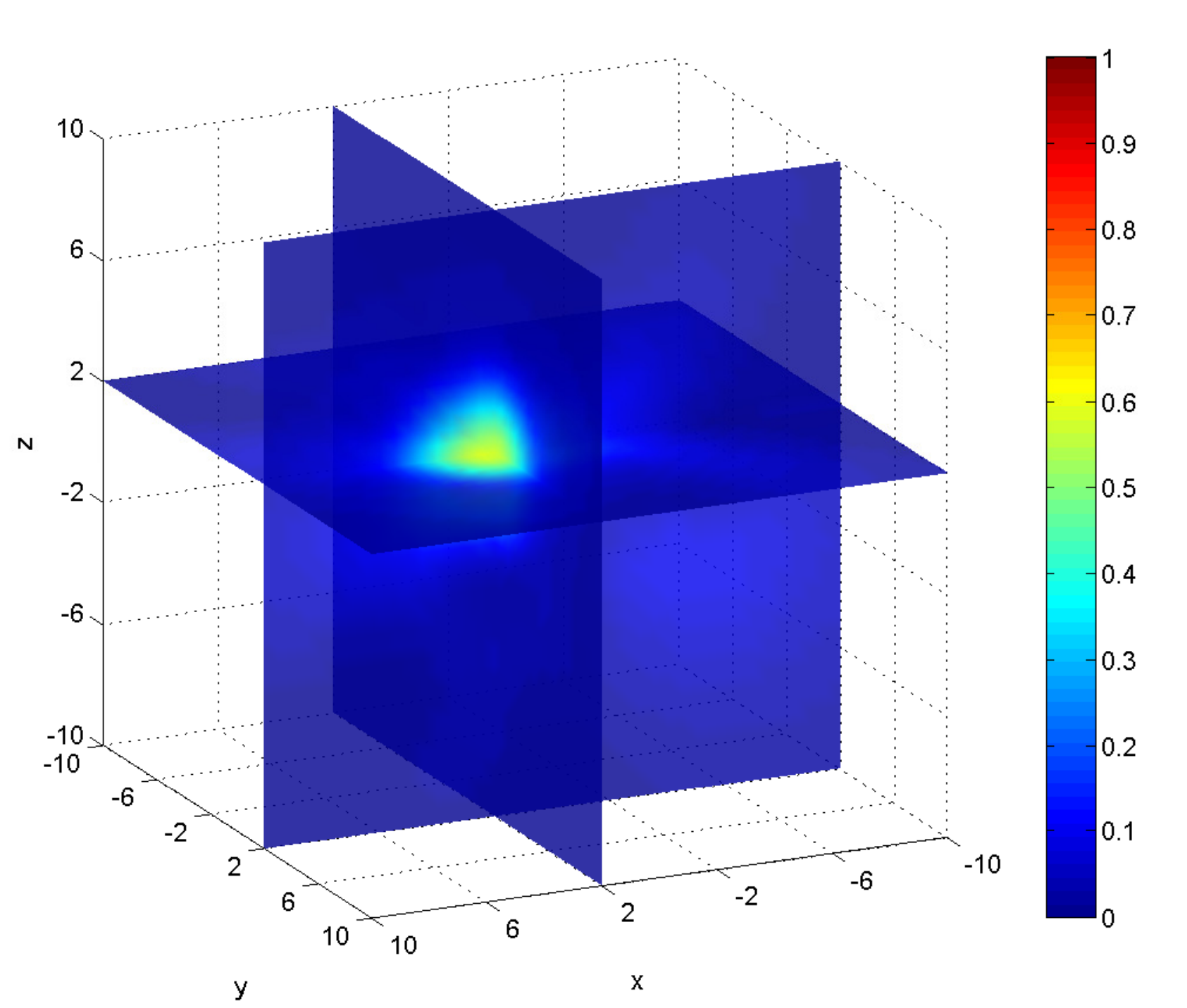}\hfill{}

\hfill{}(a)\hfill{}\hfill{}(b)\hfill{}

\hfill{}\includegraphics[width=0.4\textwidth]{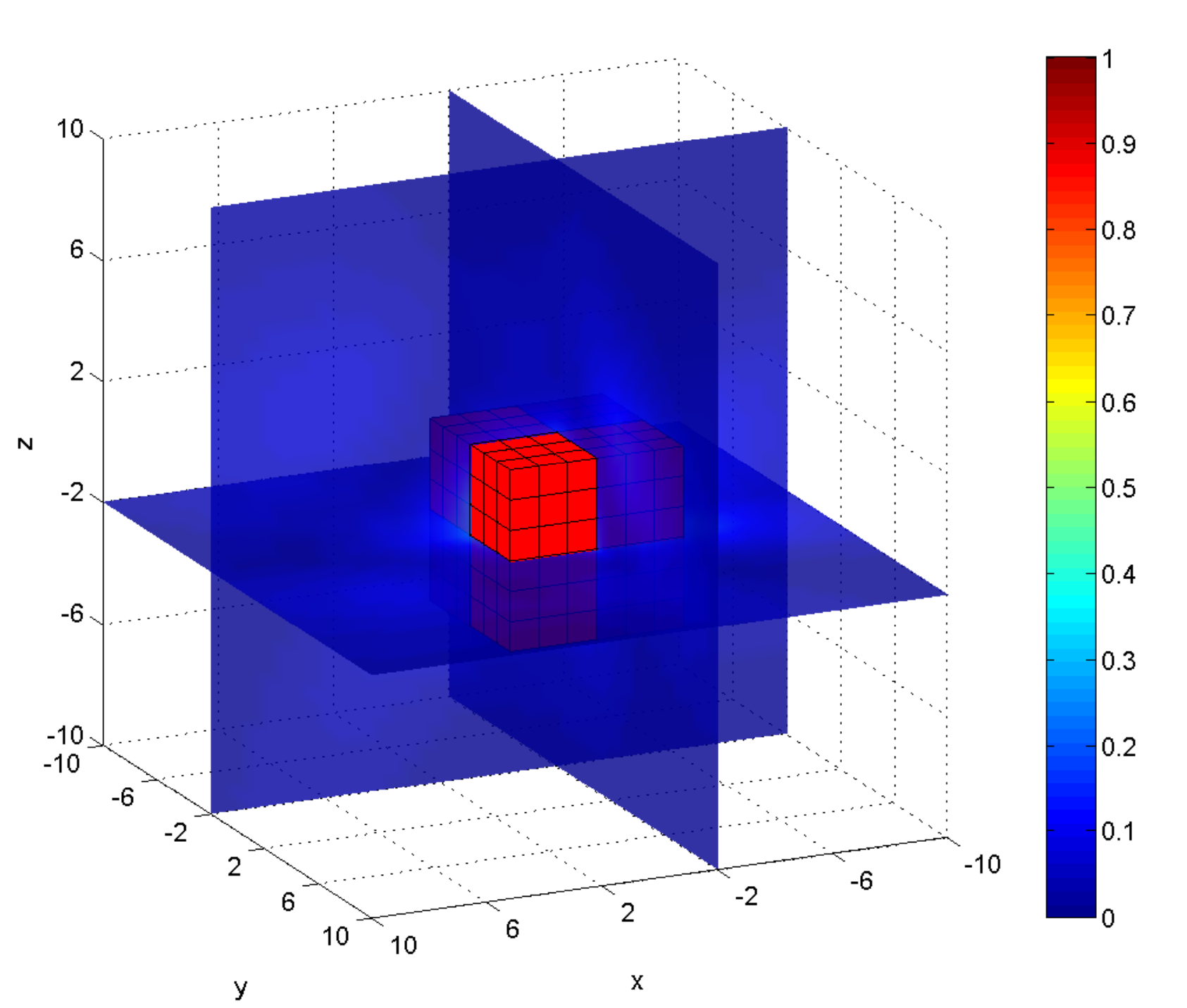}\hfill{}\includegraphics[width=0.4\textwidth]{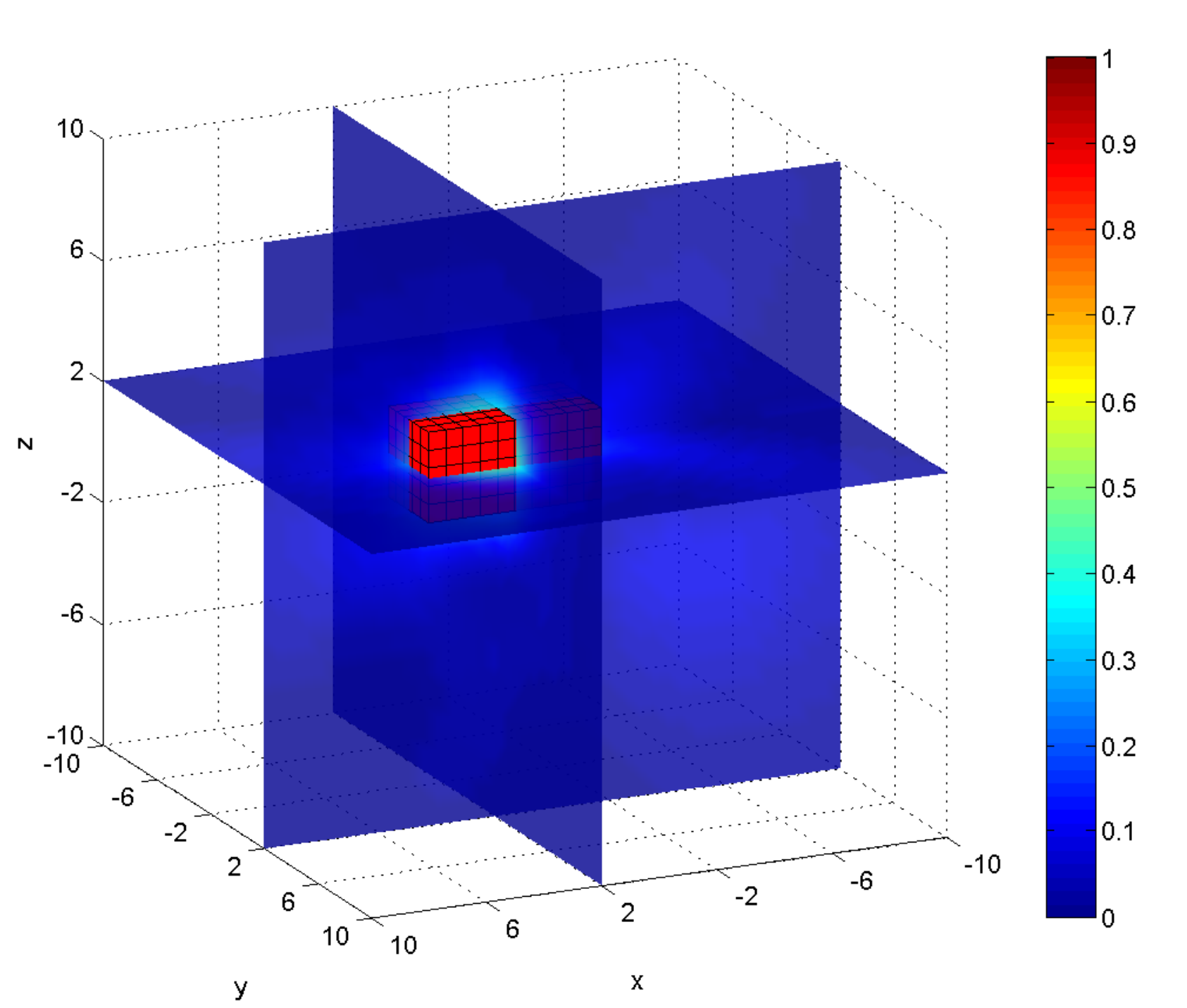}\hfill{}

\hfill{}(c)\hfill{}\hfill{}(d)\hfill{}

\hfill{}\includegraphics[width=0.4\textwidth]{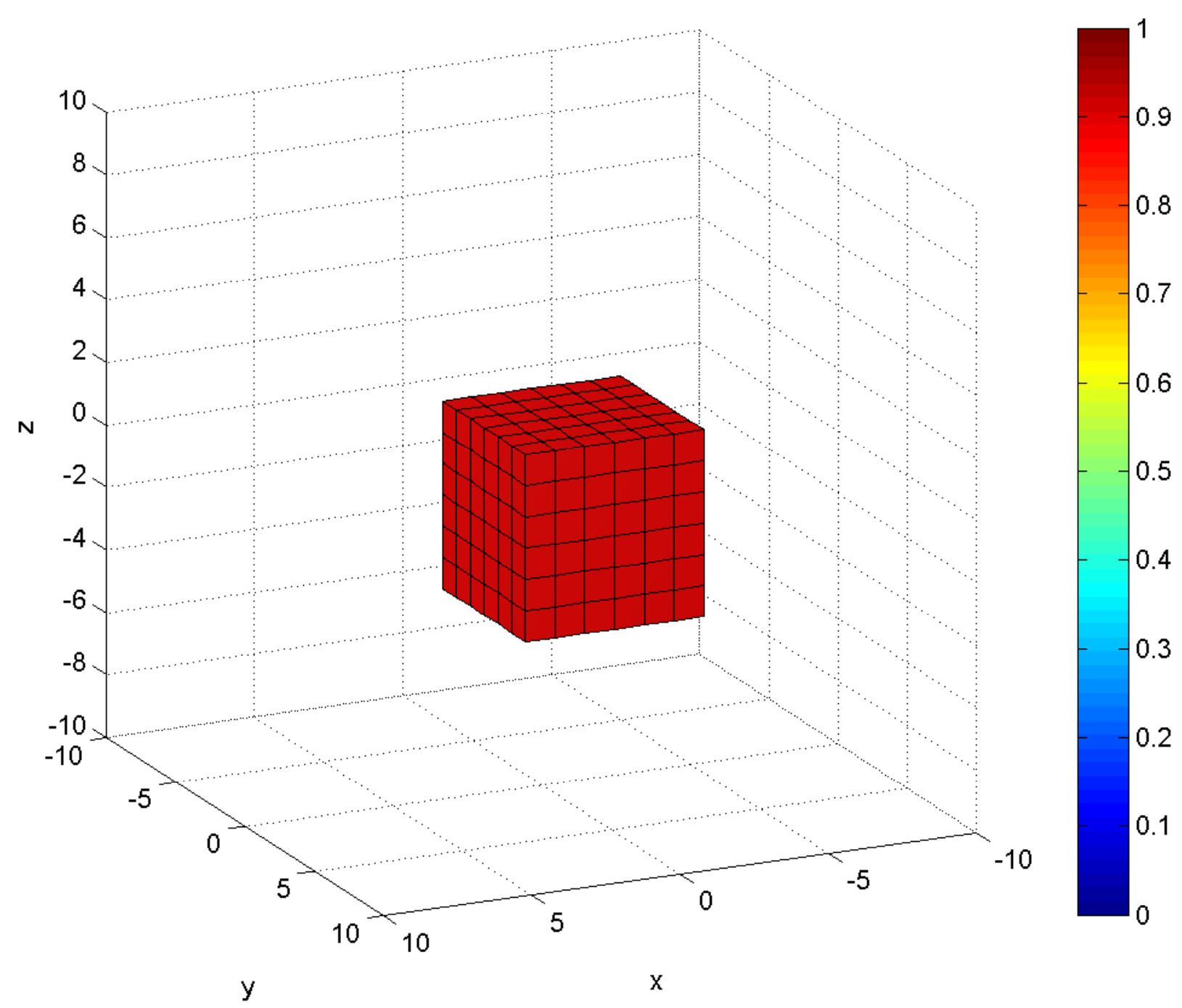}\hfill{}\includegraphics[width=0.4\textwidth]{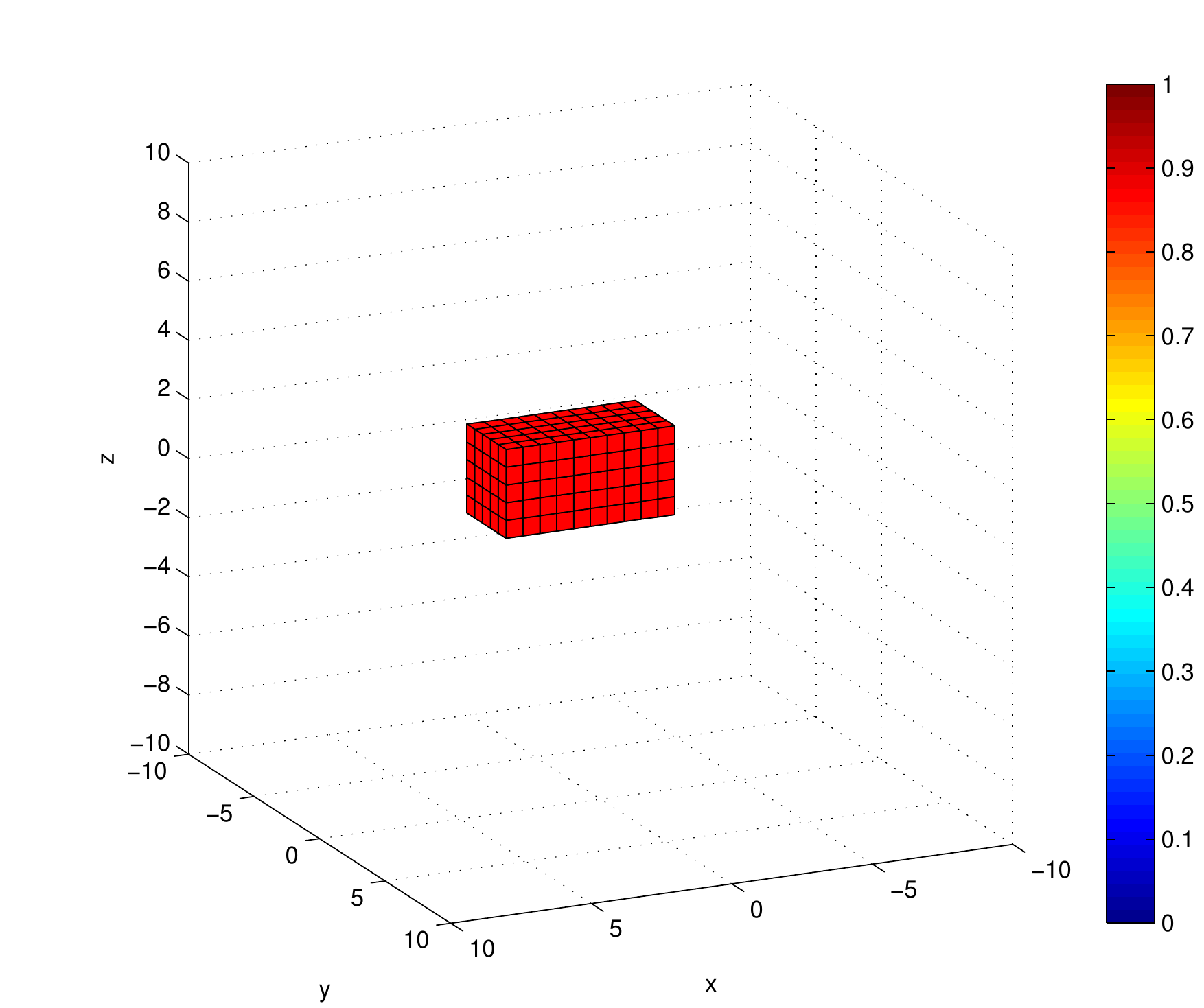}\hfill{}

\hfill{}(e)\hfill{}\hfill{}(f)\hfill{}

\caption{\label{fig:Final-image-of-Coarse}Identification in the coarse/preprocessing stage  in Example \textbf{PK}.}
\end{figure}

\paragraph{Scheme AR.}

In this stage, we take $k=5$. With the collected far-field data, we implement Scheme AR to determine the location, shape, orientation and size of each scatterer component.

When we use the far-field data of the reference peanut with $3\pi/4$ azimuthal angle and unitary scale as the test data 
in the indicator function $I_r^j(z)$ (cf.~\eqref{eq:indicator regular}),
the  distribution of the indicator 
function  is shown in Fig.~\ref{fig:Fine-Stage-Identification-Peanut}(a). Then we take maximum of the indicator values and find a 
much precise location $(-2.1,\ -2.1,\ -2.1)$ of the peanut, as in Fig.~\ref{fig:Fine-Stage-Identification-Peanut}(b). 
Based on that position, we plot the proper shape, orientation and size based on the information carried with the far field data employed 
and plot the imaging result  in Fig.~\ref{fig:Fine-Stage-Identification-Peanut}(c). 
Its projection on the orthogonal cut   planes across its location are shown in Fig.~\ref{fig:Fine-Stage-Identification-Peanut}(d)-(f). 
It can be concluded that the position identified is quite good and reasonable.

After excluding the peanut component, we apply Scheme AR to 
the local mesh  around the Kite component. When
the far-field data of the reference kite with $\pi/4$ azimuthal angle and unitary scale is adopted in the indicator function $I_r^j(z)$ 
(cf.~\eqref{eq:indicator regular}), the value distribution of
the indicator function is shown in Fig.~\ref{fig:Fine-Stage-Identification-Kite}(a). Then we take maximum of the indicator values and find the
the location $(2.2,\ 2.2,\ 2.2)$  of the kite, as in Fig.~\ref{fig:Fine-Stage-Identification-Kite}(b). 
As previous, we plot the exact shape, orientation and size and show three orthogonal cut planes across the location identified 
in Fig.~\ref{fig:Fine-Stage-Identification-Kite}(c)-(f).
The identified location is very close to the exact position of the kite.

\begin{figure}
\hfill{}\includegraphics[width=0.32\textwidth]{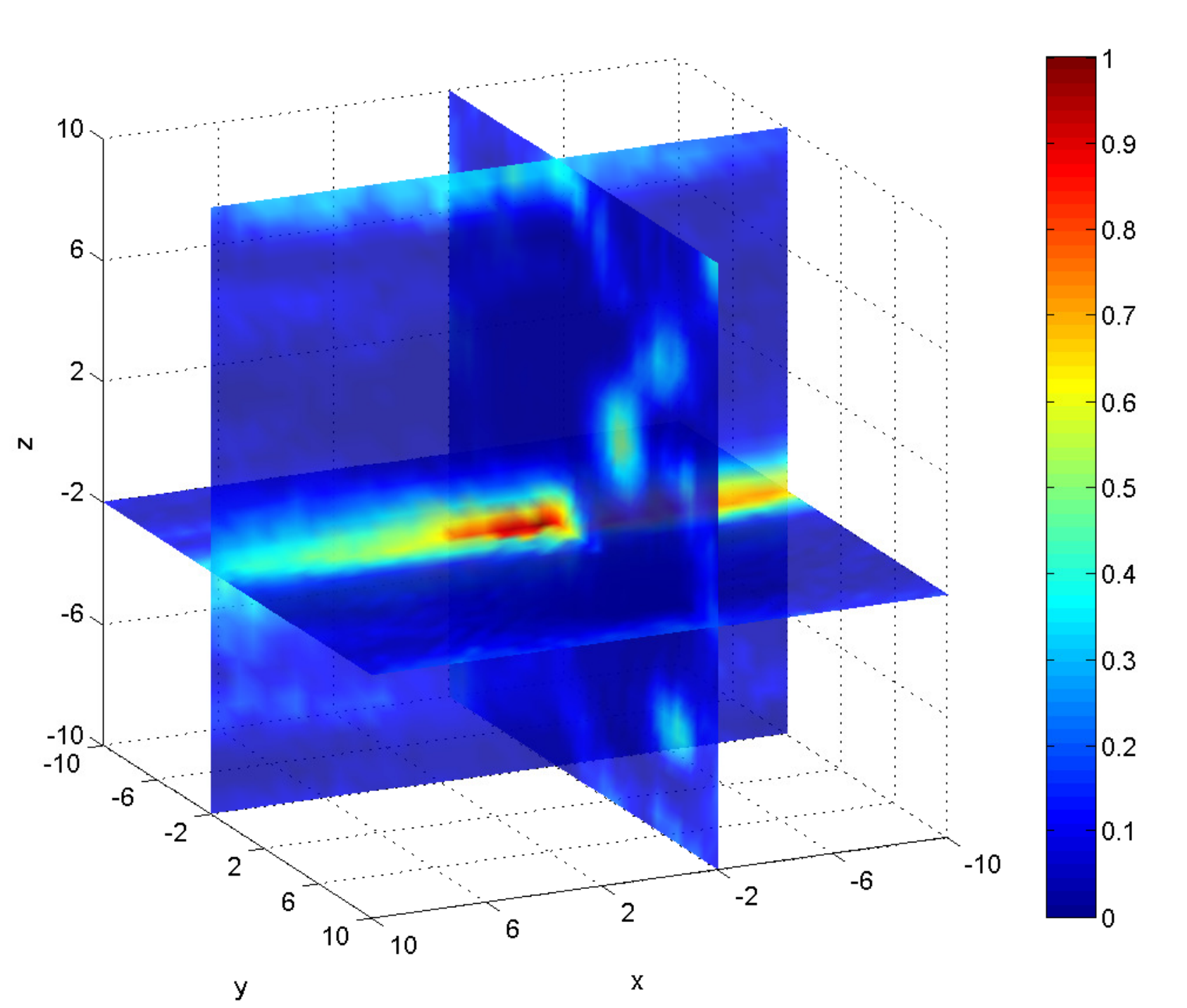}\hfill{}
\includegraphics[width=0.32\textwidth]{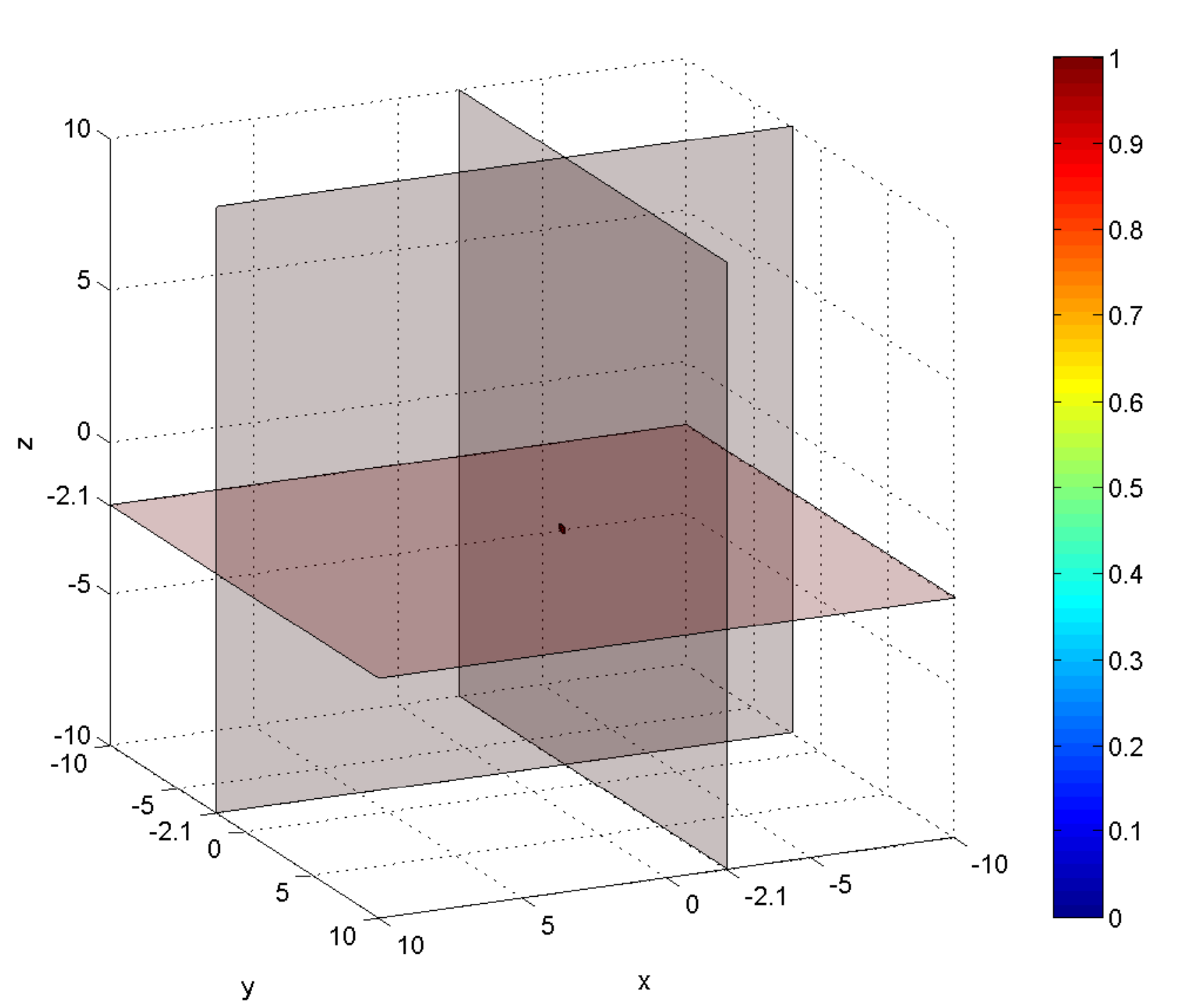}\hfill{}\includegraphics[width=0.32\textwidth]{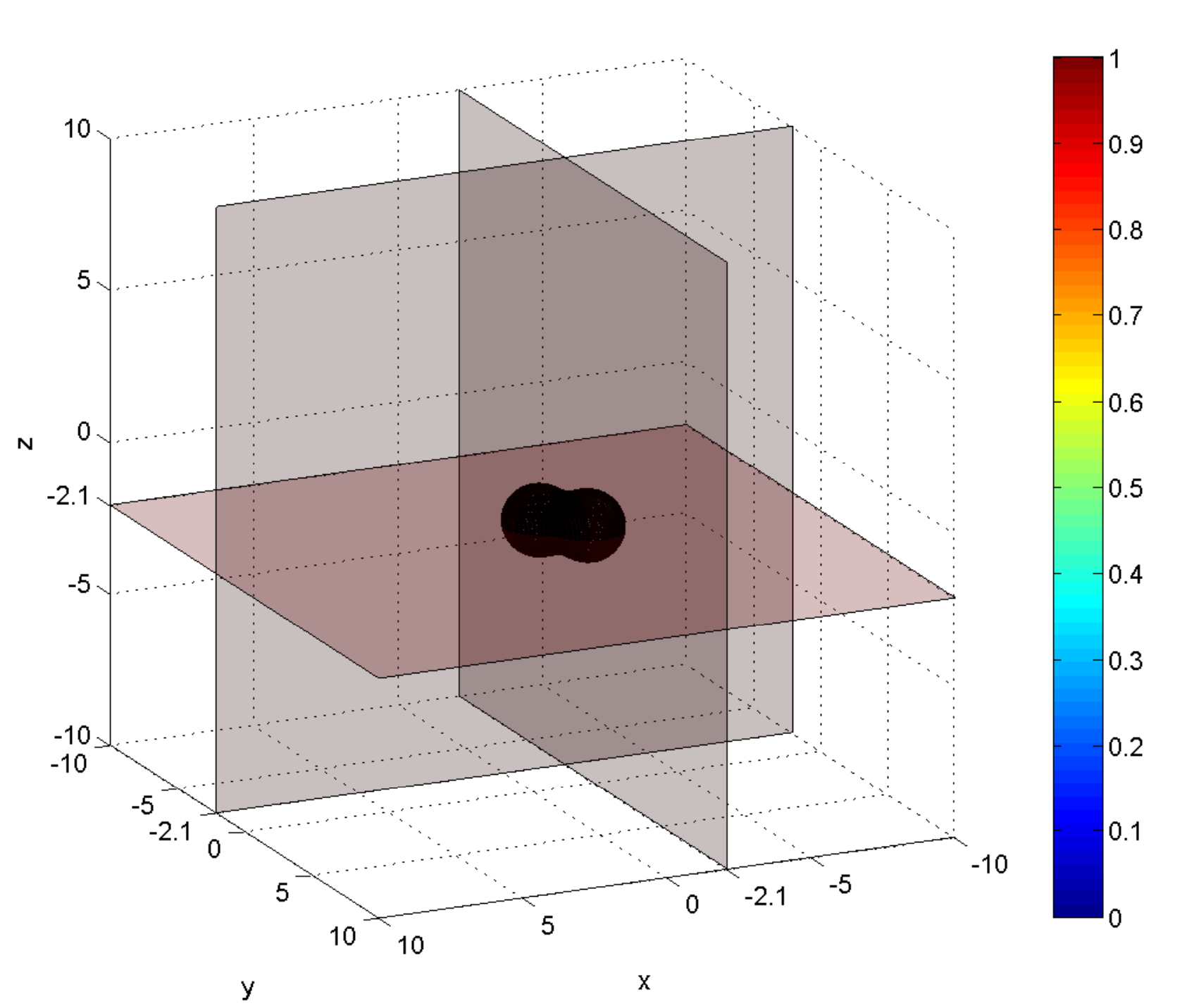}\hfill{}

\hfill{}\!\!\!\!\!\!\!\!\!\!(a)\hfill{}~~~~~~~~~~(b)\hfill{}~~~~~~~(c)\hfill{}\hfill{}

\hfill{}\includegraphics[width=0.32\textwidth]{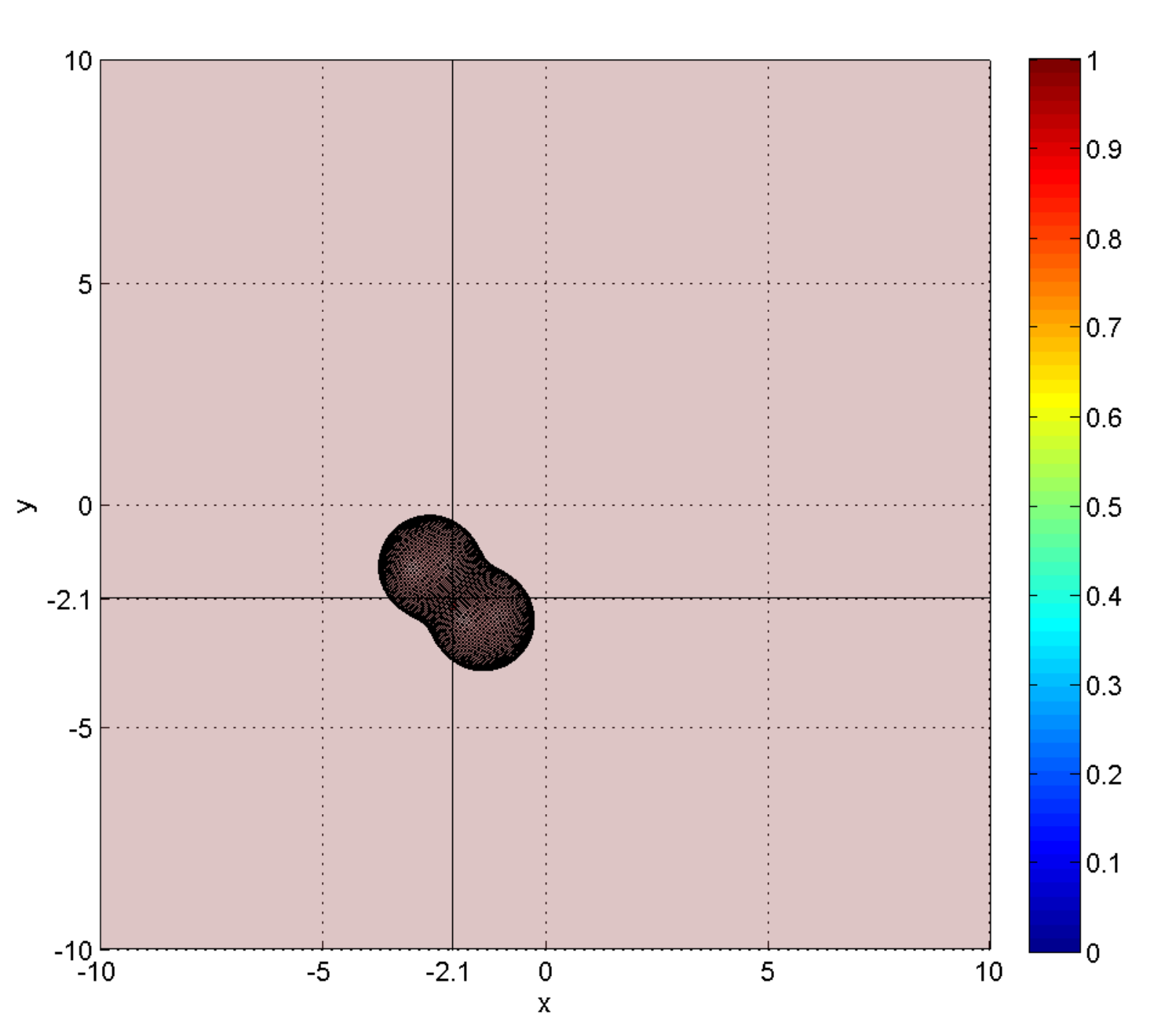}\hfill{}\includegraphics[width=0.32\textwidth]{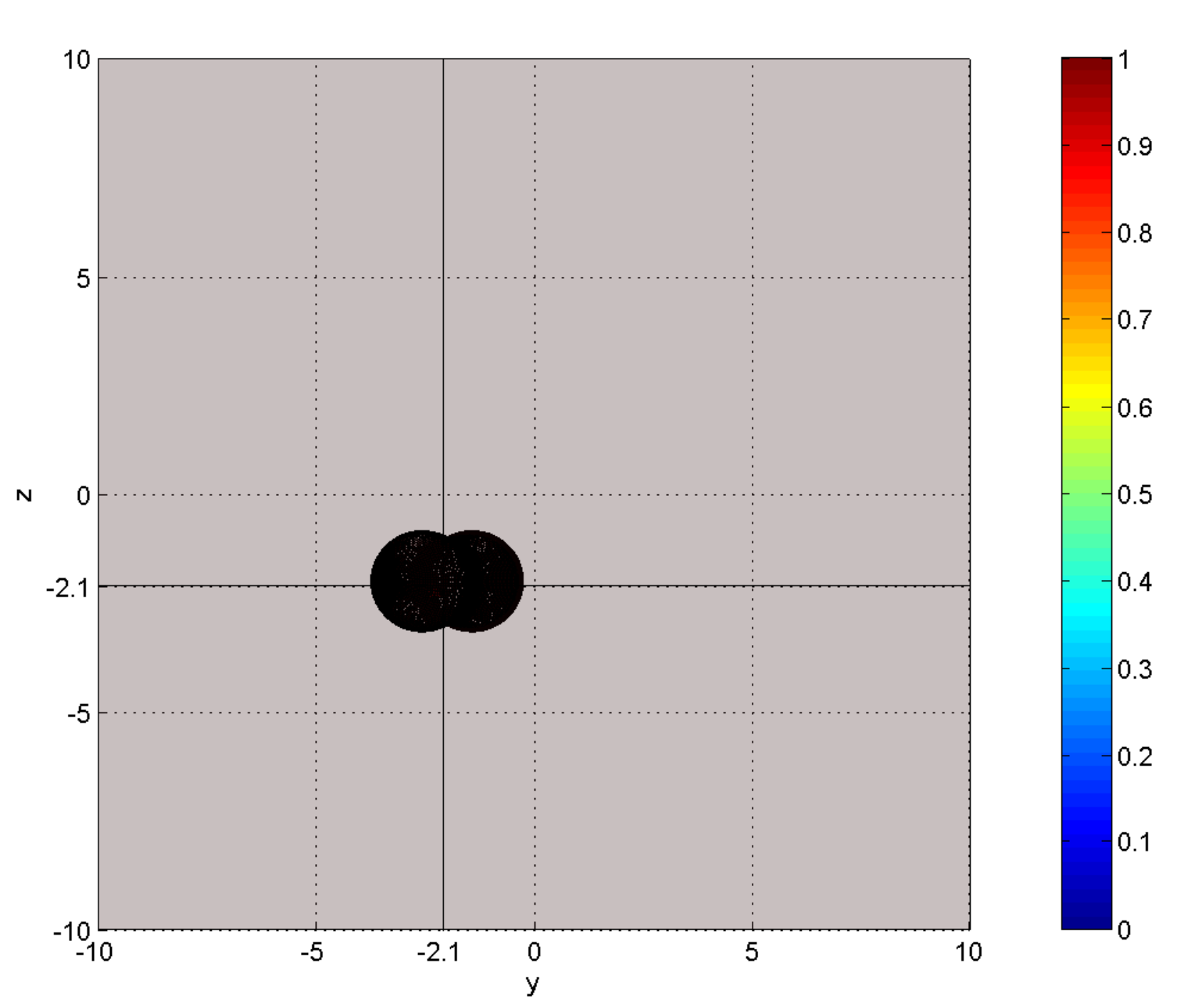}\hfill{}\includegraphics[width=0.32\textwidth]{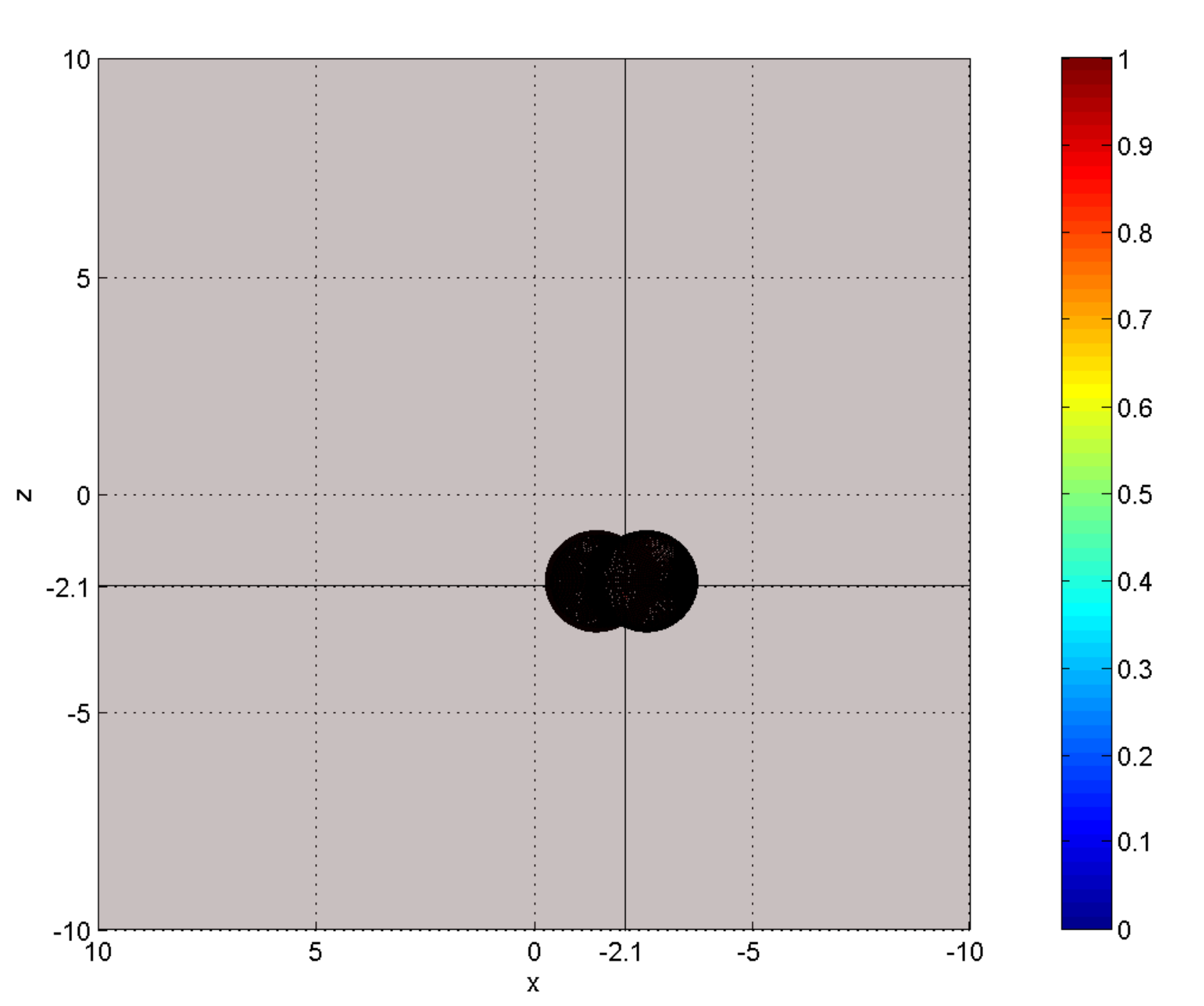}\hfill{}

\hfill{}\!\!\!\!\!\!\!\!\!\!(d)\hfill{}~~~~~~~~~~(e)\hfill{}~~~~~~~(f)\hfill{}\hfill{}

\caption{\label{fig:Fine-Stage-Identification-Peanut} Fine stage identification of the Peanut component in Example \textbf{PK}:
(a) the multi-slice plot of the indicator function; (b) rough position
by take maximum of indicator function; (c) the reconstructed component after the determination of the orientation of the peanut; (d)-(f) projections of the reconstruction in (c).}
\end{figure}

\begin{figure}
\hfill{}\includegraphics[width=0.32\textwidth]{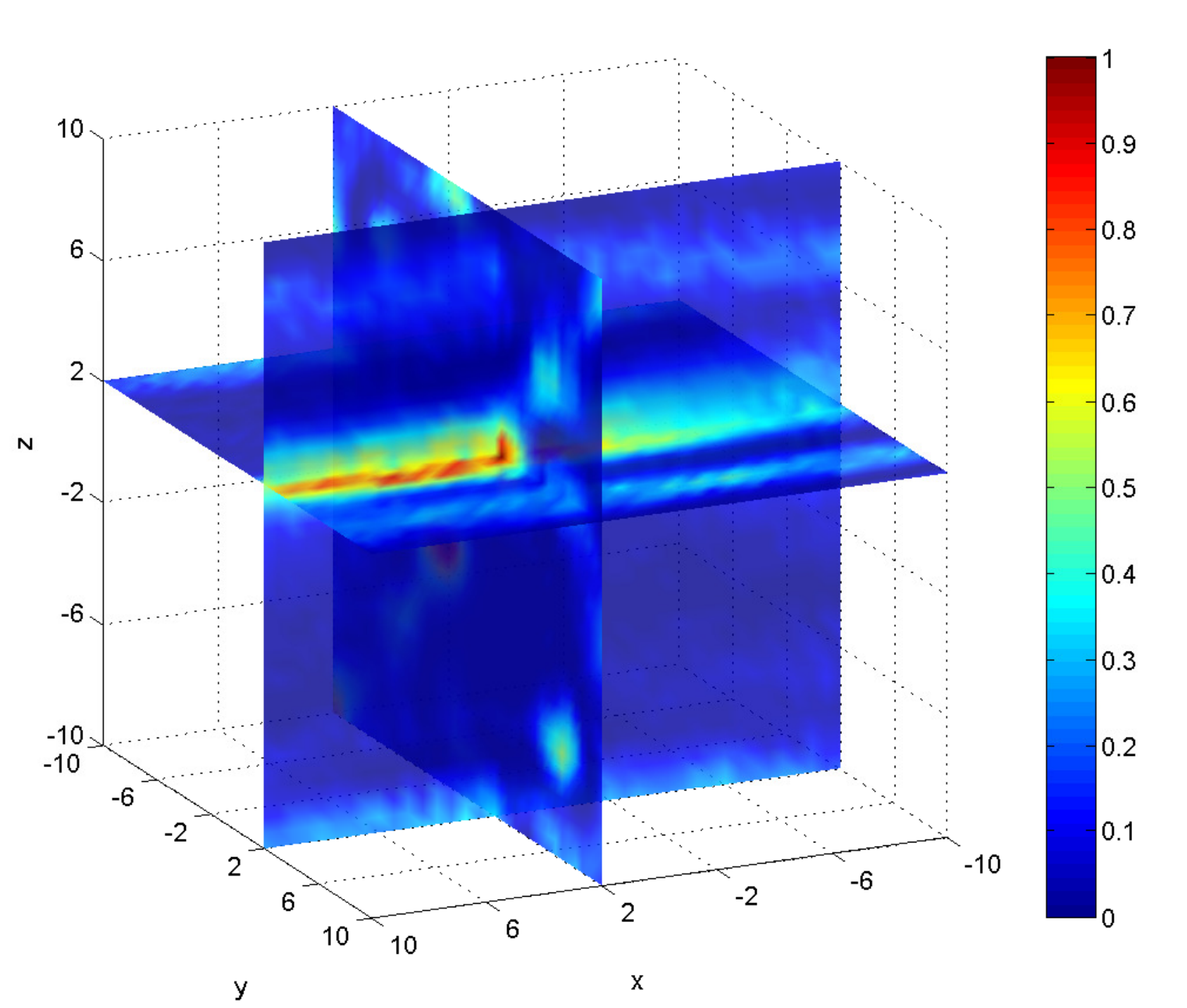}\hfill{}\includegraphics[width=0.32\textwidth]{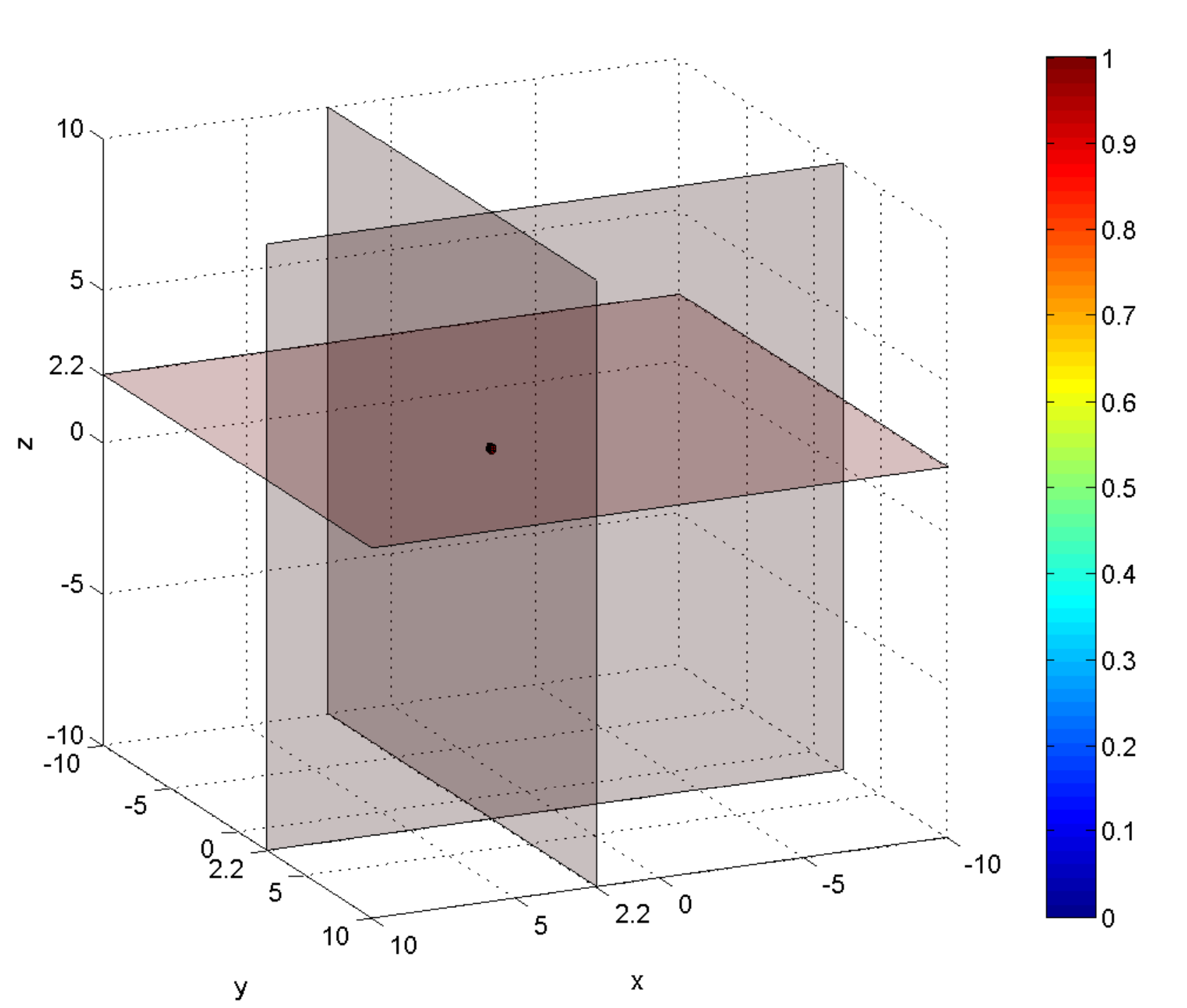}\hfill{}\includegraphics[width=0.32\textwidth]{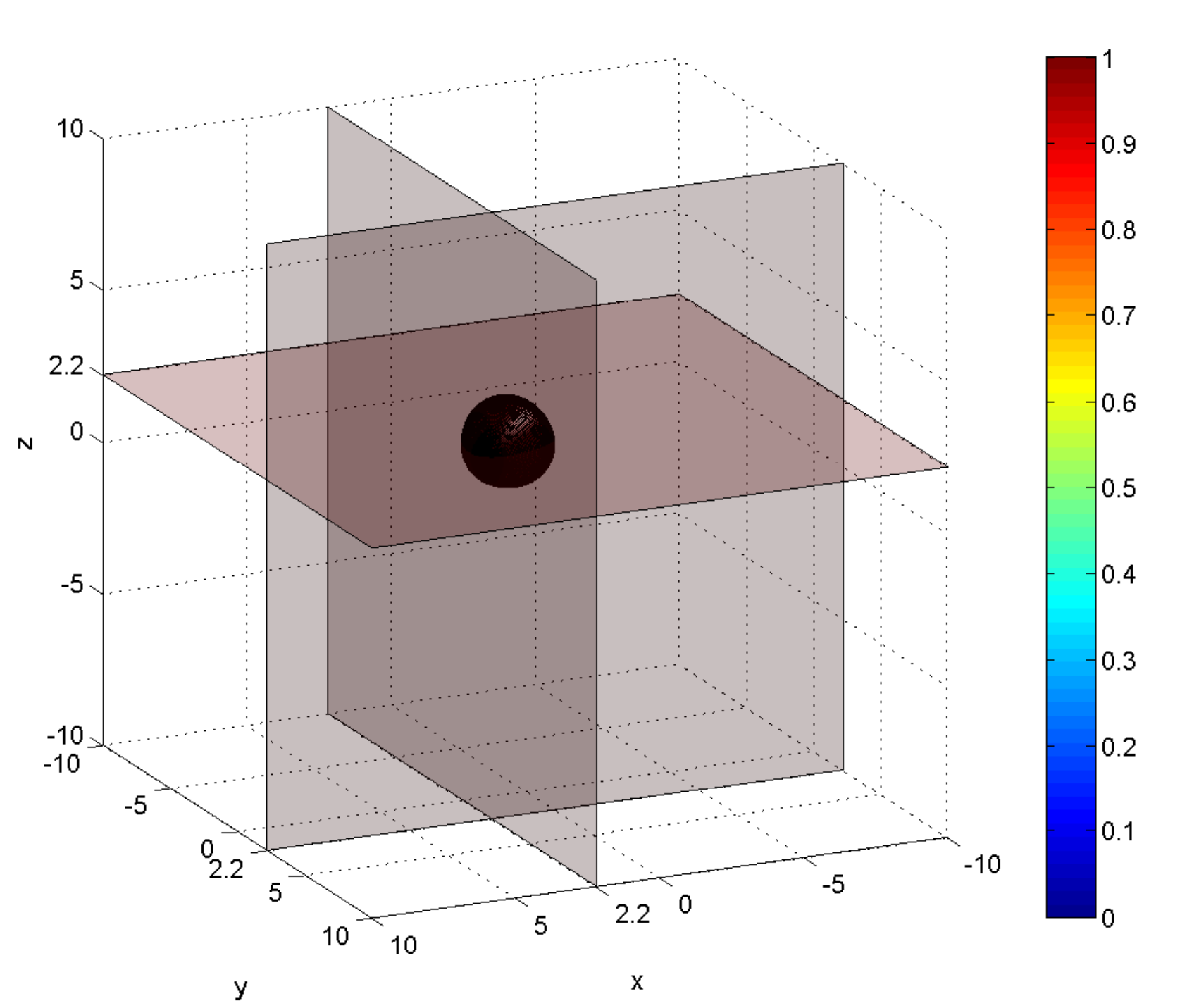}\hfill{}

\hfill{}\!\!\!\!\!\!\!\!\!\!(a)\hfill{}~~~~~~~~~~(b)\hfill{}~~~~~~~(c)\hfill{}\hfill{}

\hfill{}\includegraphics[width=0.30\textwidth]{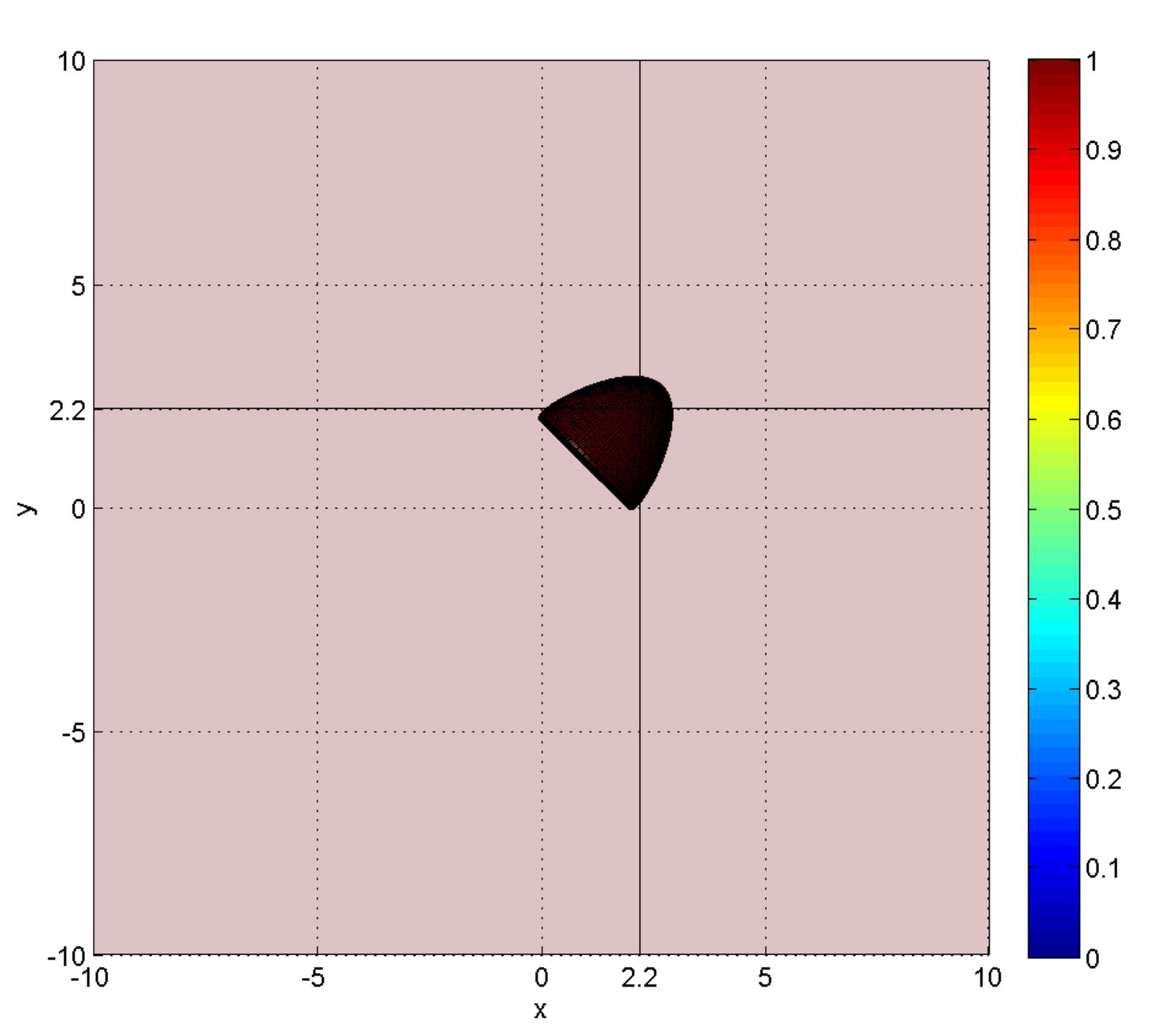}\hfill{}\includegraphics[width=0.32\textwidth]{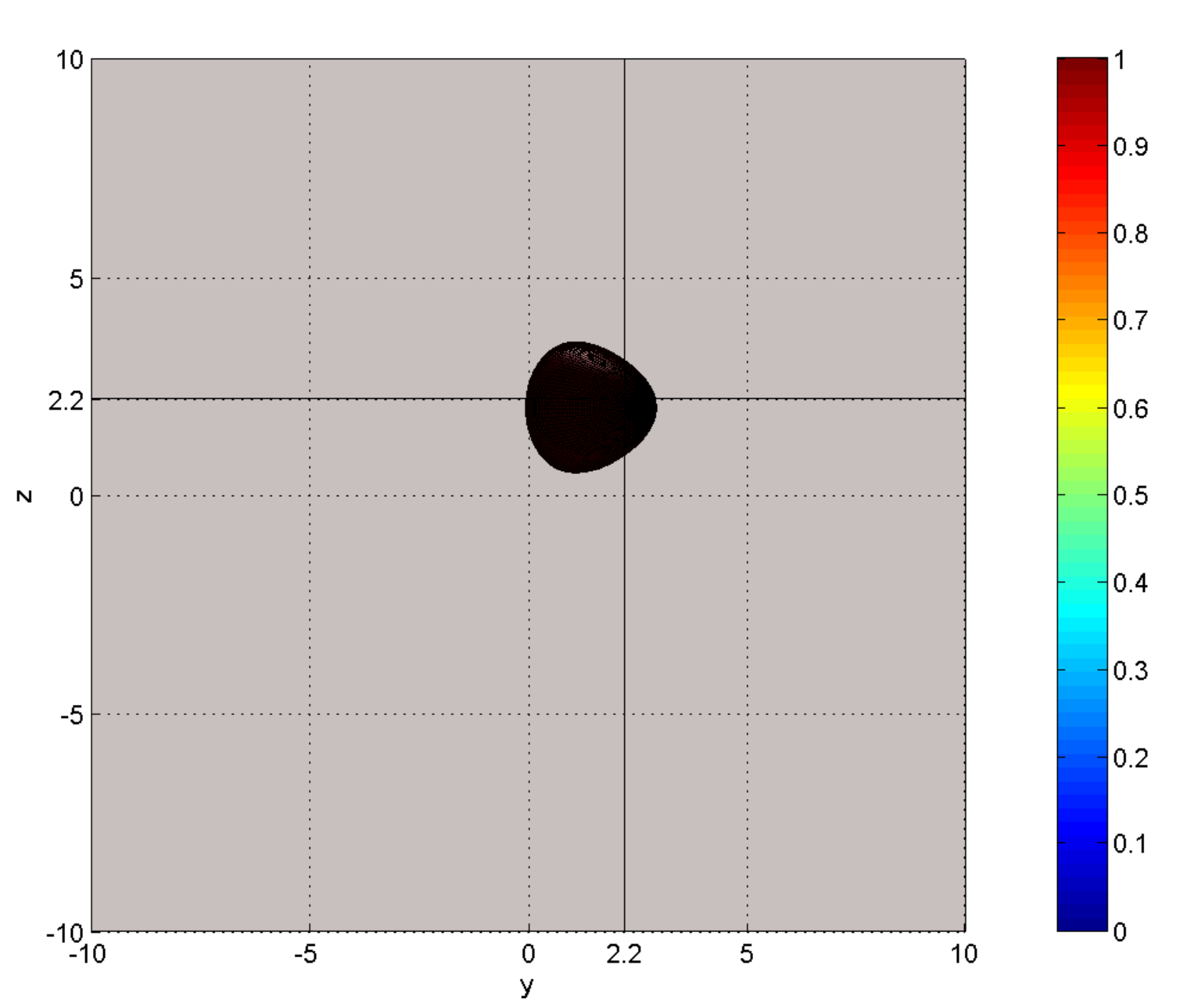}\hfill{}\includegraphics[width=0.32\textwidth]{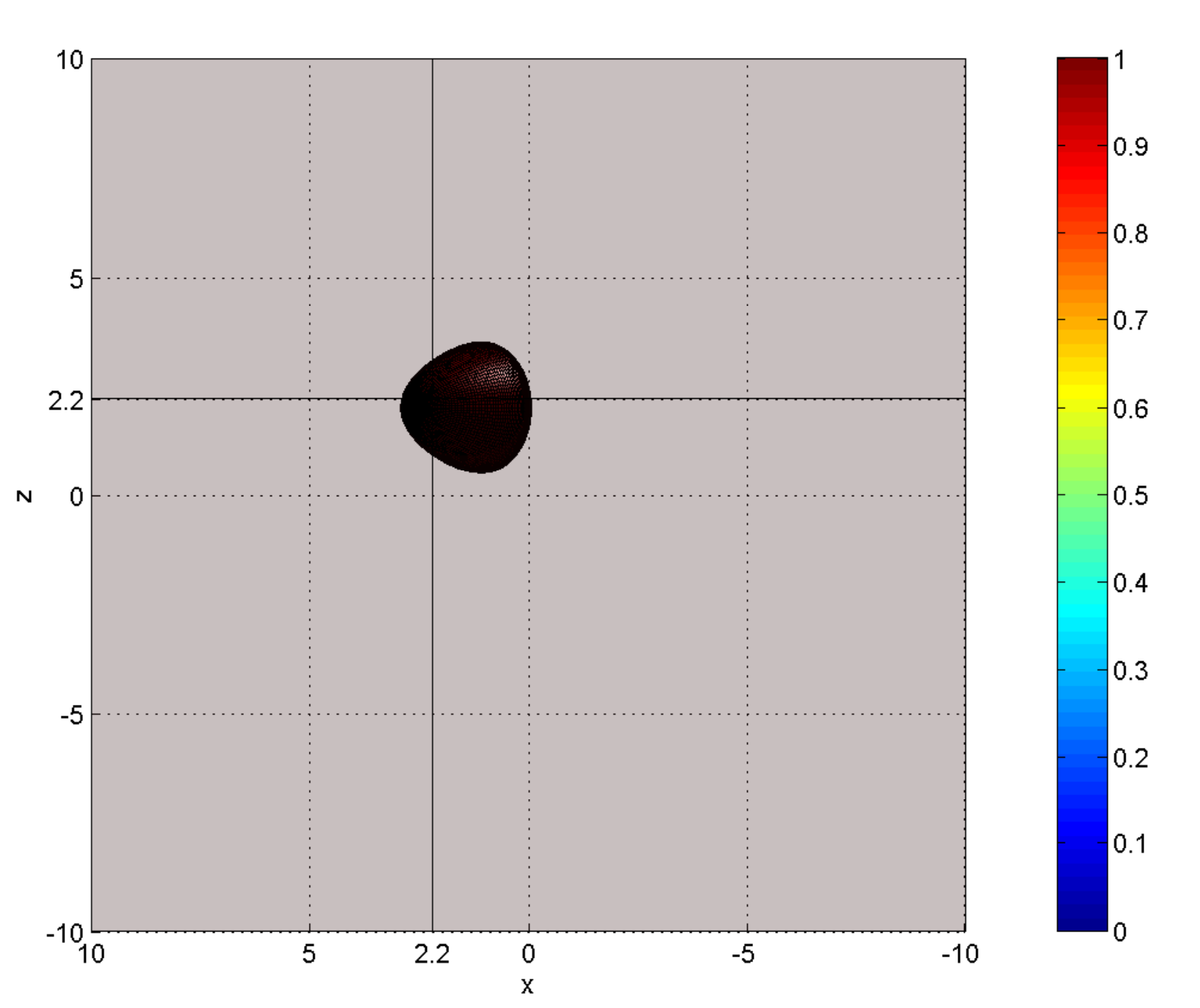}\hfill{}

\hfill{}\!\!\!\!\!\!\!\!\!\!(d)\hfill{}~~~~~~~~~~(e)\hfill{}~~~~~~~(f)\hfill{}\hfill{}

\caption{\label{fig:Fine-Stage-Identification-Kite} Fine stage Identification of the Kite component in Example \textbf{PK}: (a) the multi-slice plot of the indicator function; (b) rough position by taking
maximum of the indicator function; (c) the reconstructed component after the determination of the orientation and size of the kite; (d)-(f) projections of the reconstruction in (c).}

\end{figure}


\subsection{Enhanced Scheme M}

\paragraph*{Example KB.}  In this example we try to locate multiple multi-scale scattering components using Enhanced Scheme M.
The exact scatterer is composed of a kite-shaped scatterer enlarged by two times from the reference one and a ball scatterer scaled by a half from the unit one.
The kite is chosen to be a PEC obstacle, whereas the ball is an inhomogeneous medium. 
The exact scatterer is shown in Fig.~\ref{fig:True-scatterer-mlts}, where the 3D kite-shaped  component is located at
$(0,\,0,\,-4)$ and  the ball component is located
at $(0,\,0,\,9)$ with radius a half unit.

\begin{figure}[bp]
\hfill{}\includegraphics[width=0.23\textwidth]{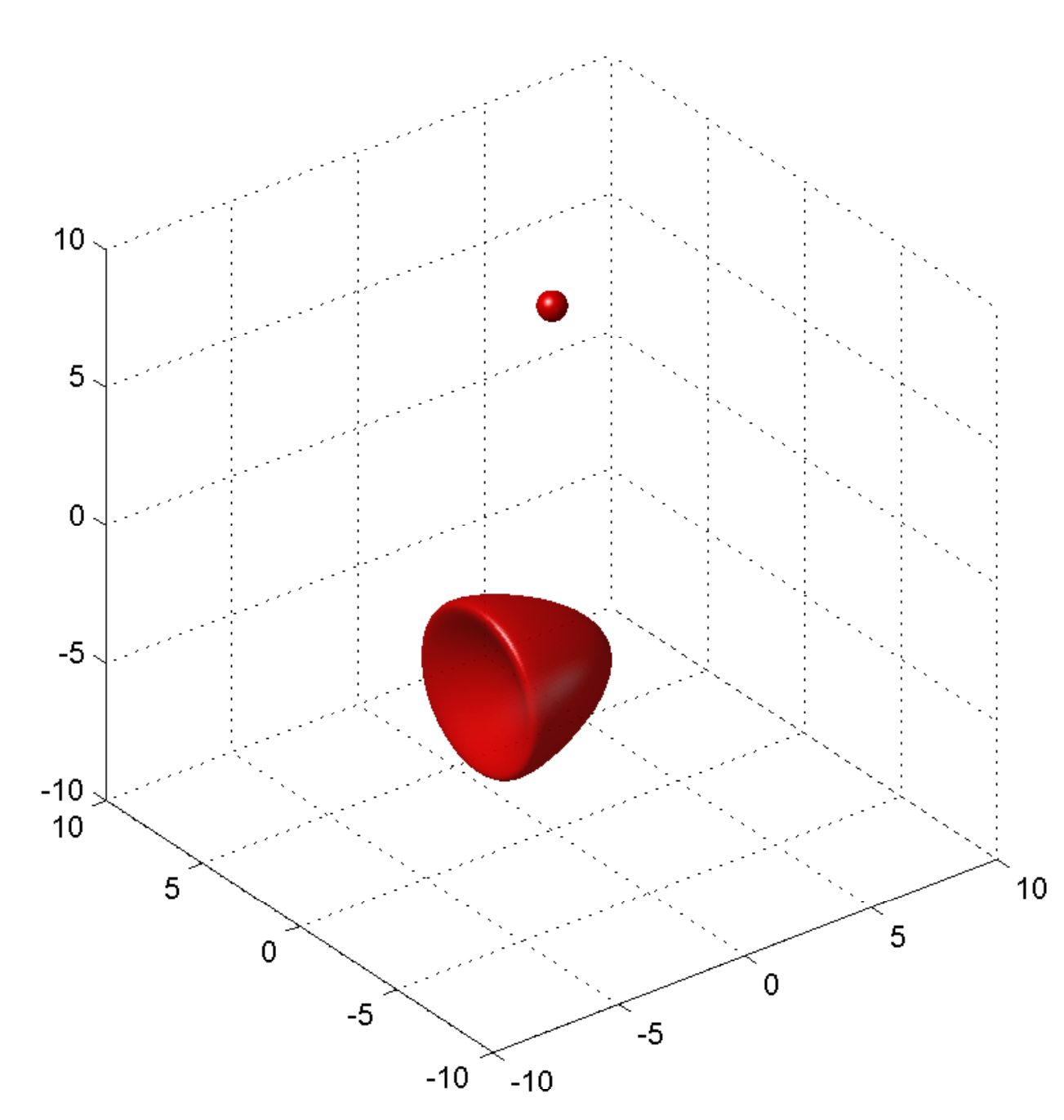}\hfill{}\includegraphics[width=0.23\textwidth]{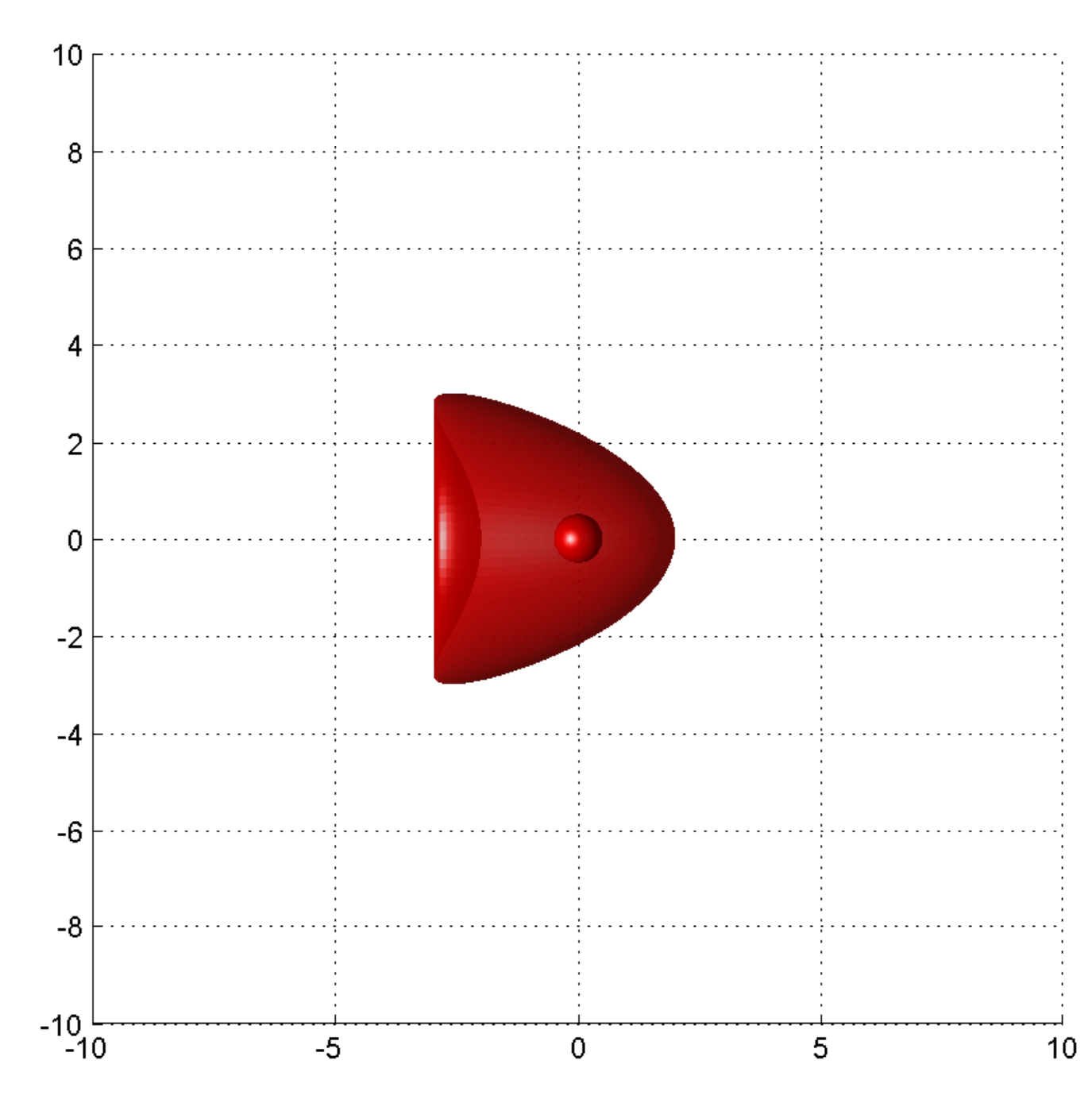}\hfill{}\includegraphics[width=0.23\textwidth]{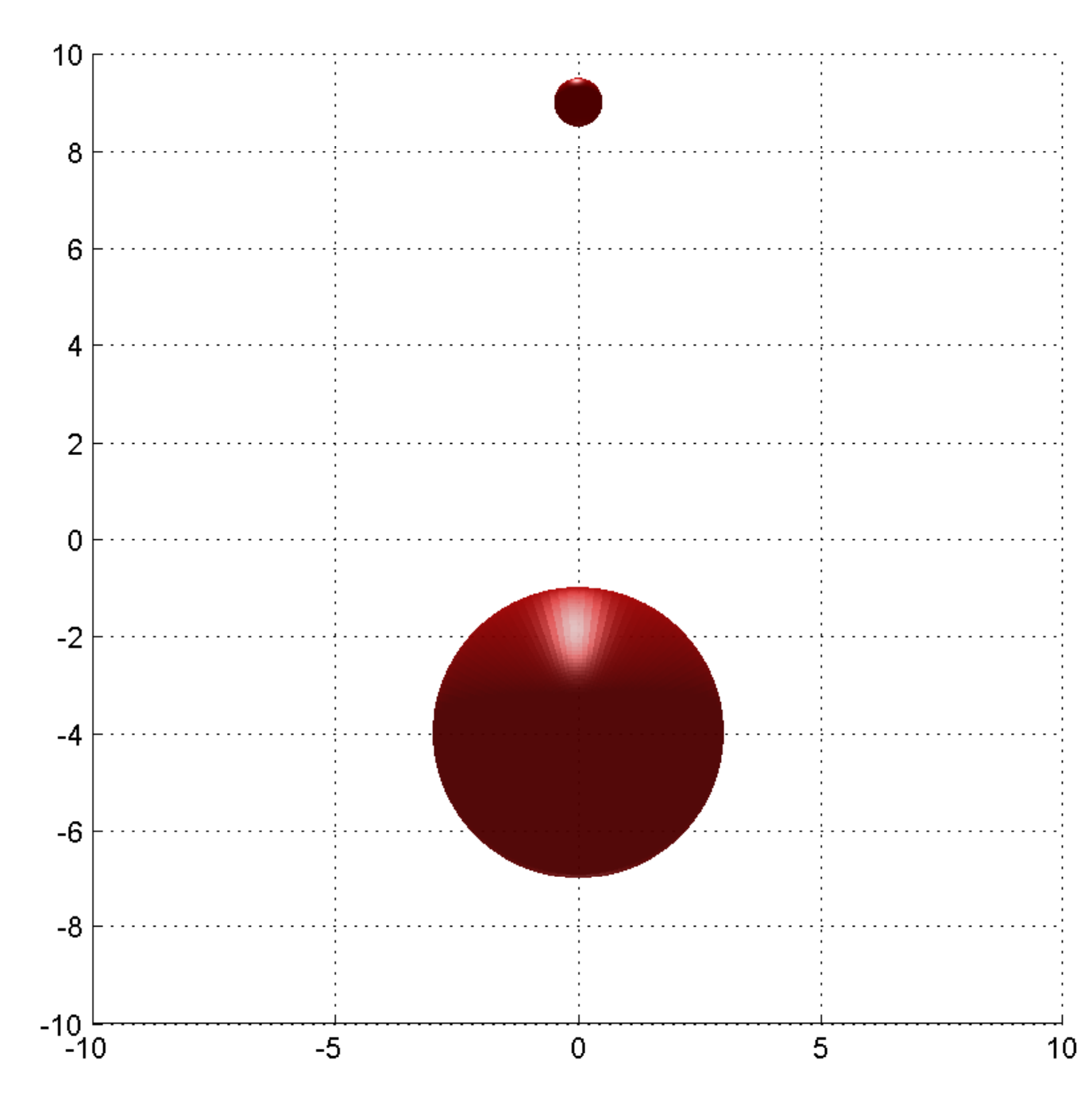}\hfill{}\includegraphics[width=0.23\textwidth]{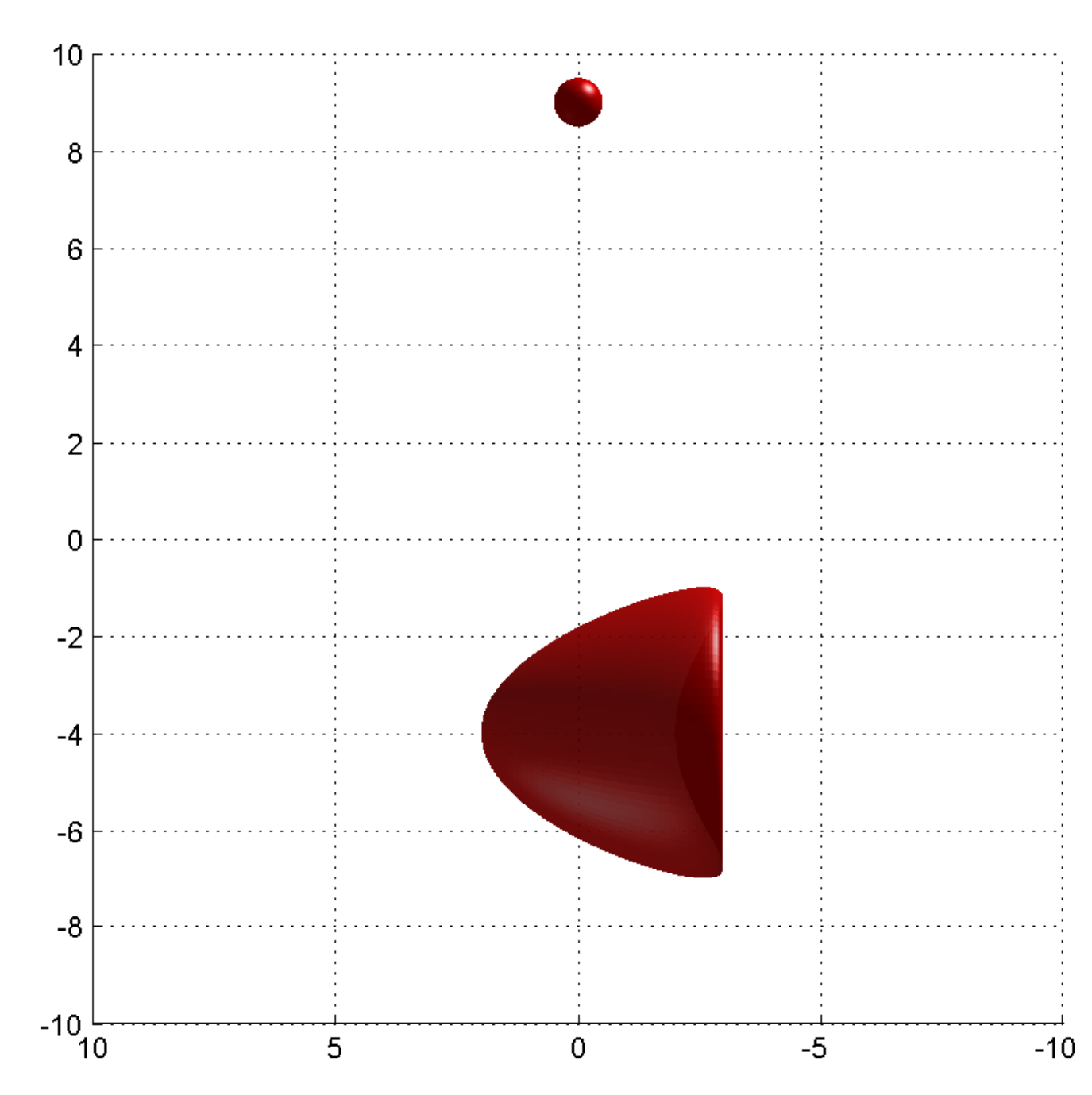}\hfill{}

\hfill{}(a)\hfill{}\hfill{}(b)\hfill{}\hfill{}(c)\hfill{}\hfill{}(d)\hfill{}
\caption{\label{fig:True-scatterer-mlts}True scatterer of Example \textbf{KB}. }
\end{figure}

Now we employ Enhanced Scheme M to detect the unknown scatterers by applying
Scheme AR first and then Scheme S with the local re-sampling technique. In the first stage of Scheme AR, the far-field data used are collected by illuminating the scatterer by an incident EM waves of $k=\pi$. In the second stage for Scheme S, the far-field data used are collected by illuminating the scatterers by an detecting EM waves of  $k=2\pi/5$.
For $k=\pi$, we enrich our augmented reference space $\widetilde{\mathscr{S}}$ by the far-field data corresponding to each reference components with different orientations and sizes on $590$ Lebedev quadrature points on the unit sphere.

\paragraph{Scheme AR.}

We first apply Scheme AR to the multi-scale scatterers. When
the far-field data of the reference kite with vanishing azimuthal angle and double size is adopted in the indicator function $I_r^j(z)$
(cf.~\eqref{eq:indicator regular}), 
the local maximum behavior of the indicator function is shown in Fig.~\ref{fig:Fine-Stage-Identificationt-mlts}, (a).
Using Scheme AR, we obtain a rough position of the kite component by taking
the coordinates at which the indicator function achieves the maximum, namely  $(0,\,0,\,-4.3056)$ as shown in Fig.~\ref{fig:Fine-Stage-Identificationt-mlts}(a). 
Its shape, orientation and size are superimposed by the message carried in the far-field
data and plotted in Fig.~\ref{fig:Fine-Stage-Identificationt-mlts}(b), where we reverse the $x$-axis for ease of visualization.

\paragraph*{Local re-sampling technique}

The detected position from Scheme AR in the previous step is an approximate position of the kite component due
to the noise. In order to implement the local re-sampling technique, we set a local searching
region around the obtained position point, namely $(0,\,0,\,-4.3056)$. In this test, we choose a stack of $10$-by-$10$-by-$10$ cubes centered at $(0,\,0,\,-4.3056)$
with total side length $1$, namely within the precision of half wave length, as shown in Fig.~\ref{fig:Fine-Stage-Identificationt-mlts}, (c)
and (d). Then we  subtract the  the far-field pattern associated with the regular-size component from the total one following \eqref{eq:re1} by testing every searching node in the
cubic mesh points.

\paragraph{Scheme S.}

The rest of the job is to  follow Step 4) in Enhanced Scheme M to test every suspicious points among the cubic grid points as shown in Fig.~\ref{fig:Fine-Stage-Identificationt-mlts}(c).  Fig.~\ref{fig:Recogonize little scatterer} shows a gradual evolution process as we
move gradually the sampling grid point from the nearly correct $z_{0}=(0,\,0,\,-4.0056)$
to a perturbed position $z_{0}=(0,\,0,\,-4.1056))$, which helps us update the position of the regular-size \textbf{K} component to be $z_{0}=(0,\,0,\,-4.0056)$ but also determine the location of the small-size \textbf{B} component.  From this example, we see that the identified position of the small ball component 
is no longer available if the position of the regular-size component is slightly perturbed. For the current test, the tolerance of the perturbation is within $0.05$.
Hence, a nice by-product from the local re-sampling technique is that it helps improve significantly  the position of the regular-size component.
The operation in this stage is essentially very cheap since only a few local grid points are involved and the
re-sampling procedure only computes inner product of the subtracted far-field data with the test data in \eqref{eq:indicator function}. Moreover,
efficiency can be further improved by implementing the algorithm in parallel.

\begin{figure}
\hfill{}\includegraphics[width=0.3\textwidth]{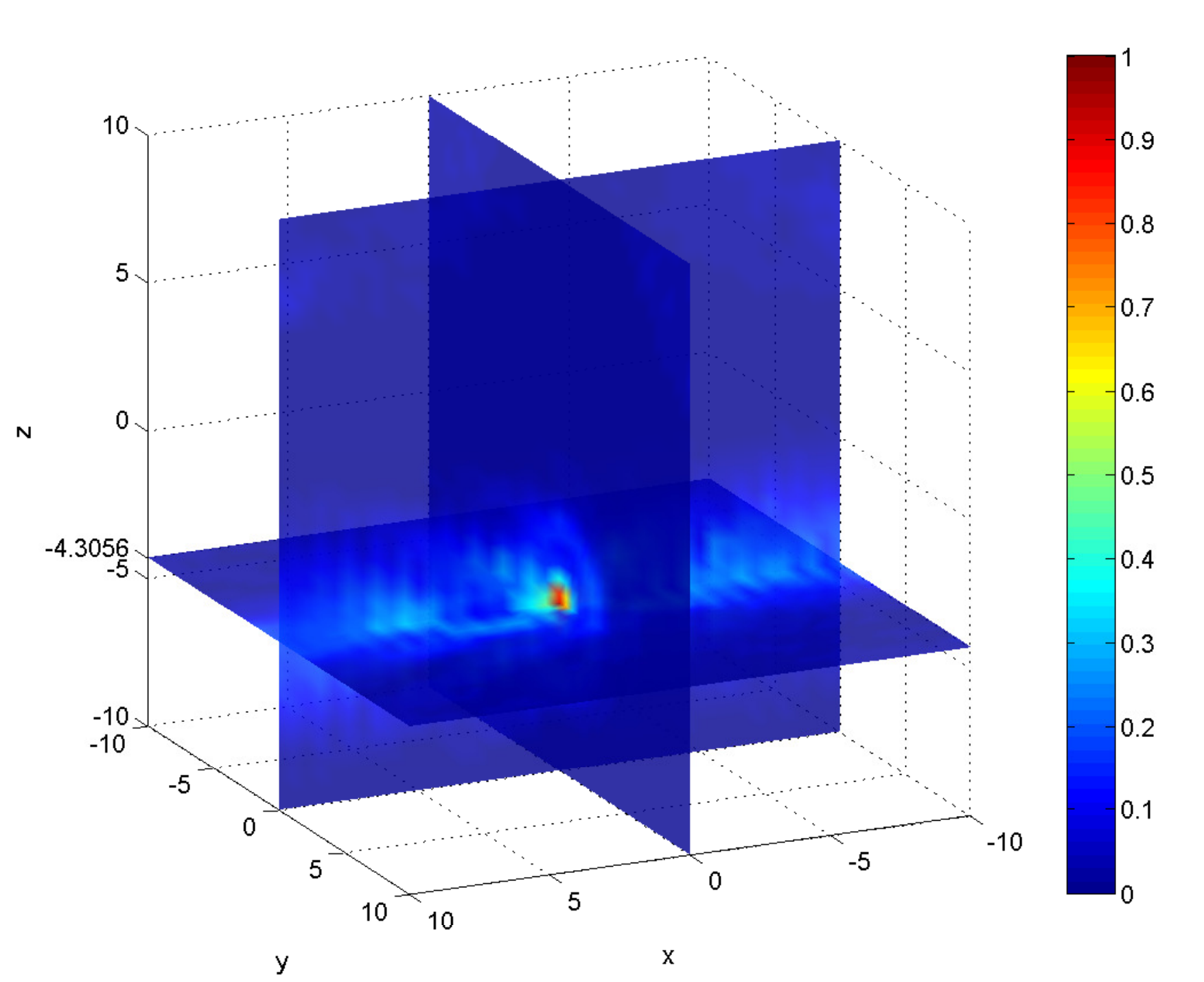}\hfill{}\includegraphics[width=0.3\textwidth]{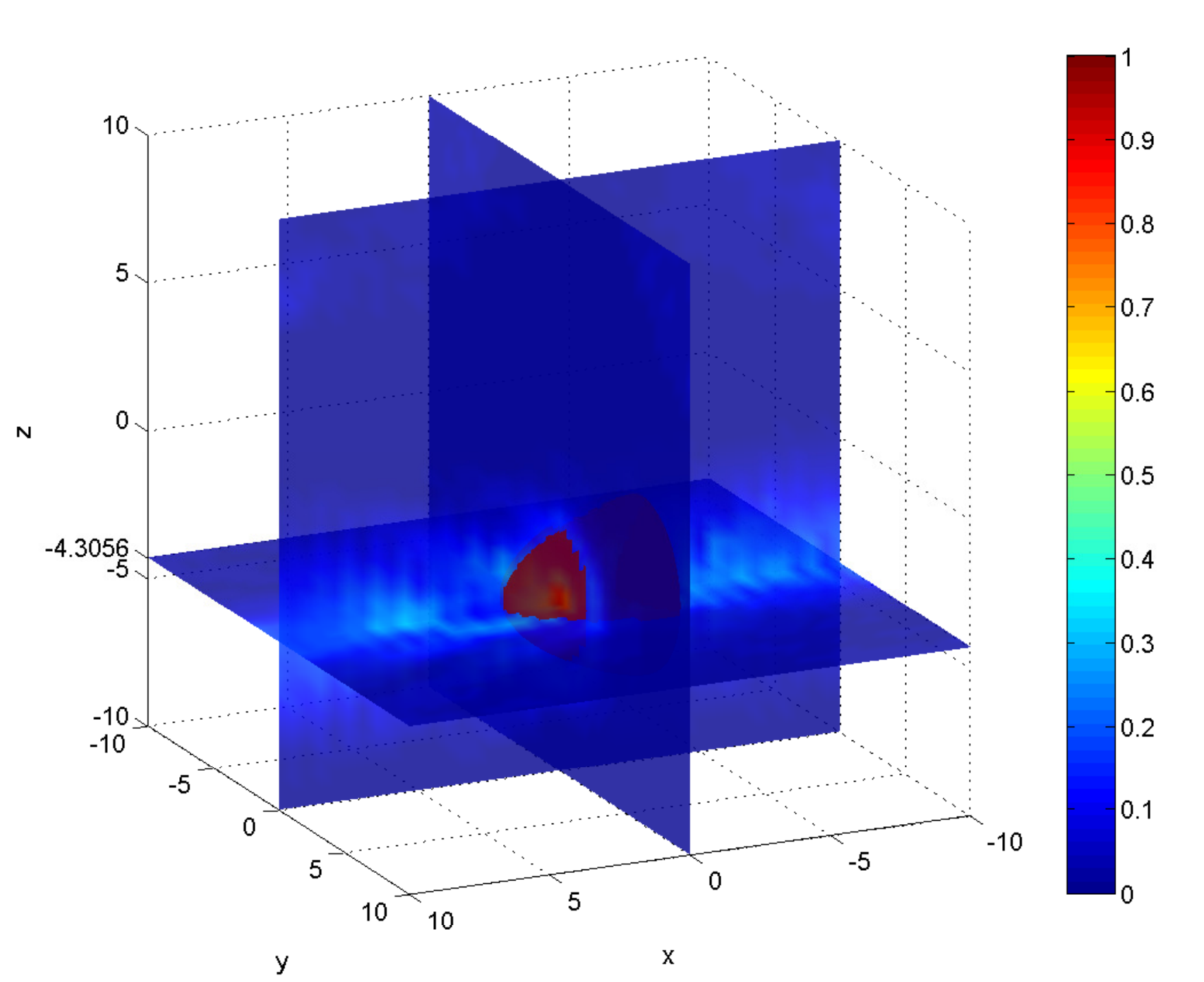}\hfill{}

\hfill{}(a)\hfill{}~~~~~~~~(b)\hfill{}

\hfill{}\includegraphics[width=0.3\textwidth]{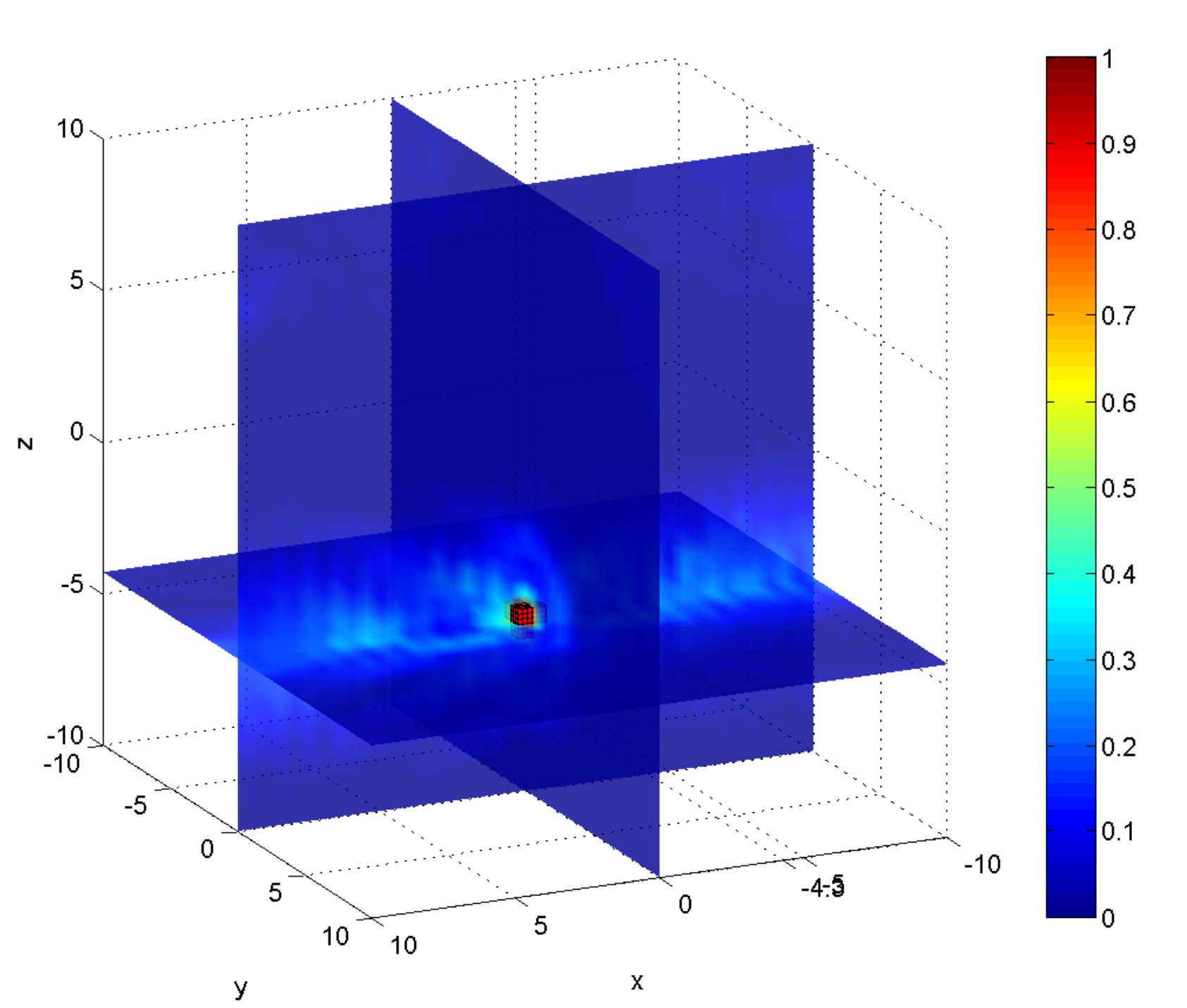}\hfill{}\includegraphics[width=0.3\textwidth]{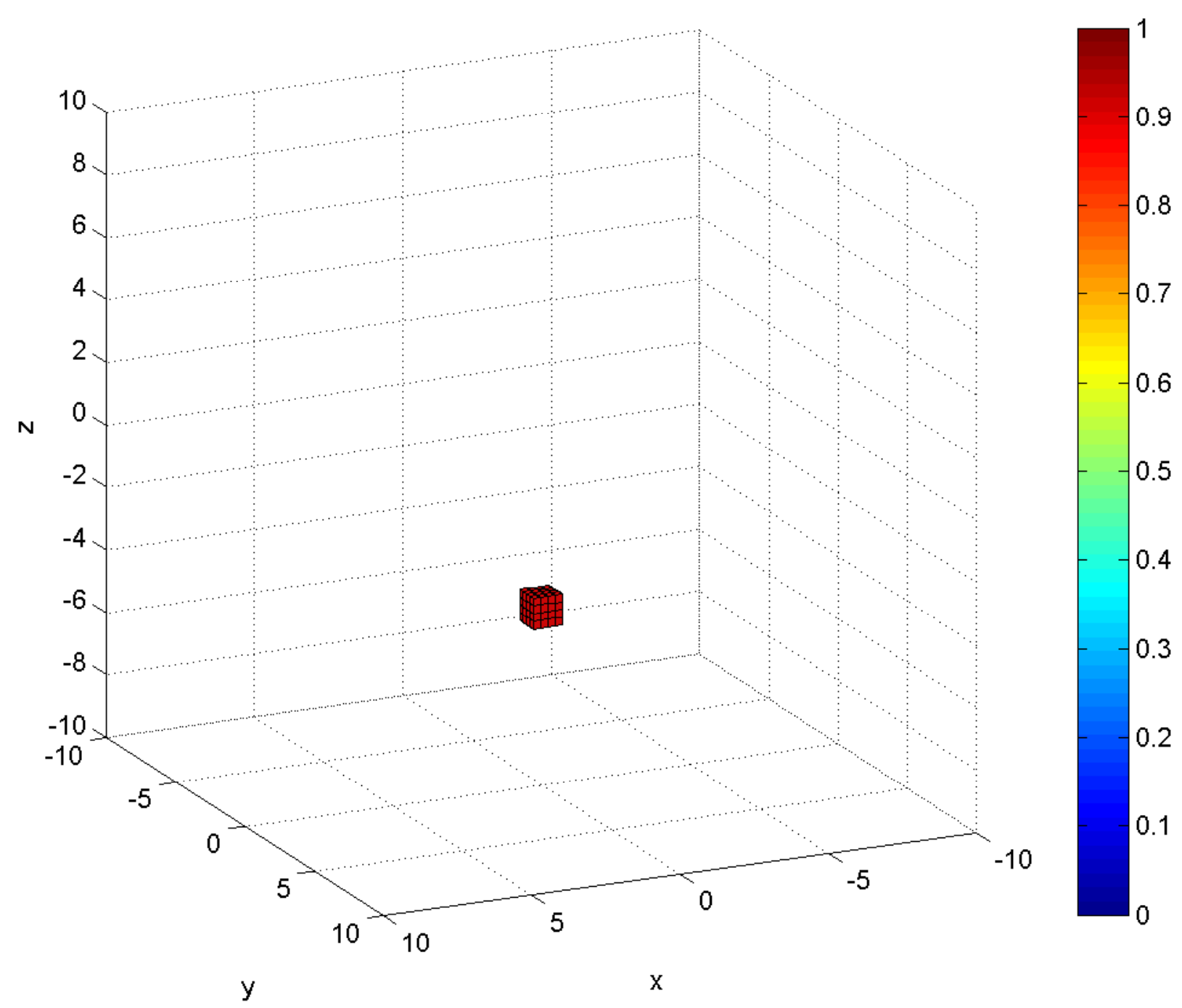}\hfill{}

\hfill{}(c)\hfill{}~~~~~~~~(d)\hfill{}

\caption{\label{fig:Fine-Stage-Identificationt-mlts} Locating by Scheme AR in Example \textbf{KB}: (a) the multi-slice plot of the indicator function by Scheme AR; (b) the reconstructed component after the determination of
the orientation and size of the kite; (c) a multi-slice plot with re-sampling cubes;
(d) the isolated re-sampling cubes without the background multi-slide plot.}
\end{figure}

\begin{figure}
\hfill{}\includegraphics[width=0.23\textwidth]{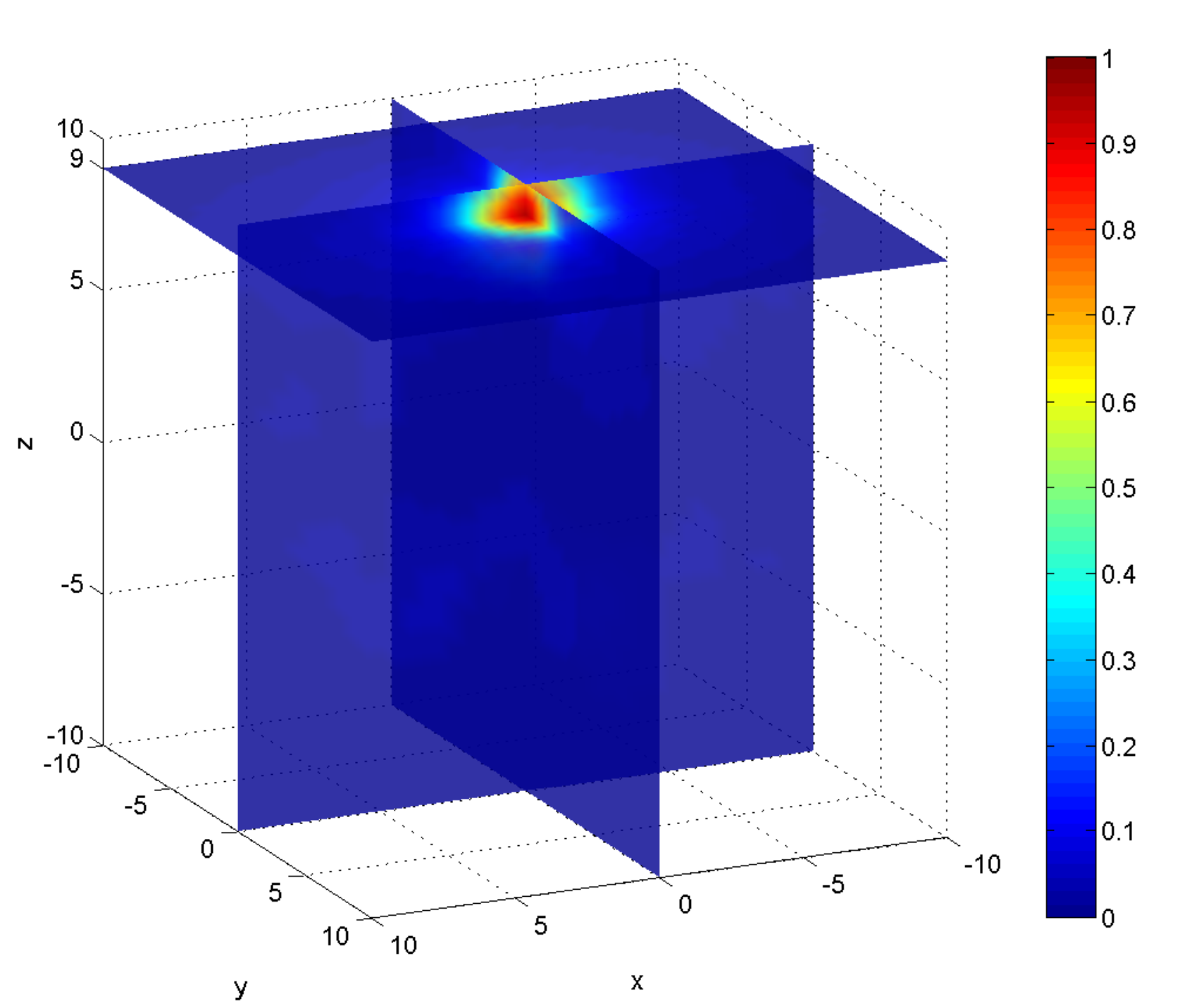}\hfill{}\includegraphics[width=0.23\textwidth]{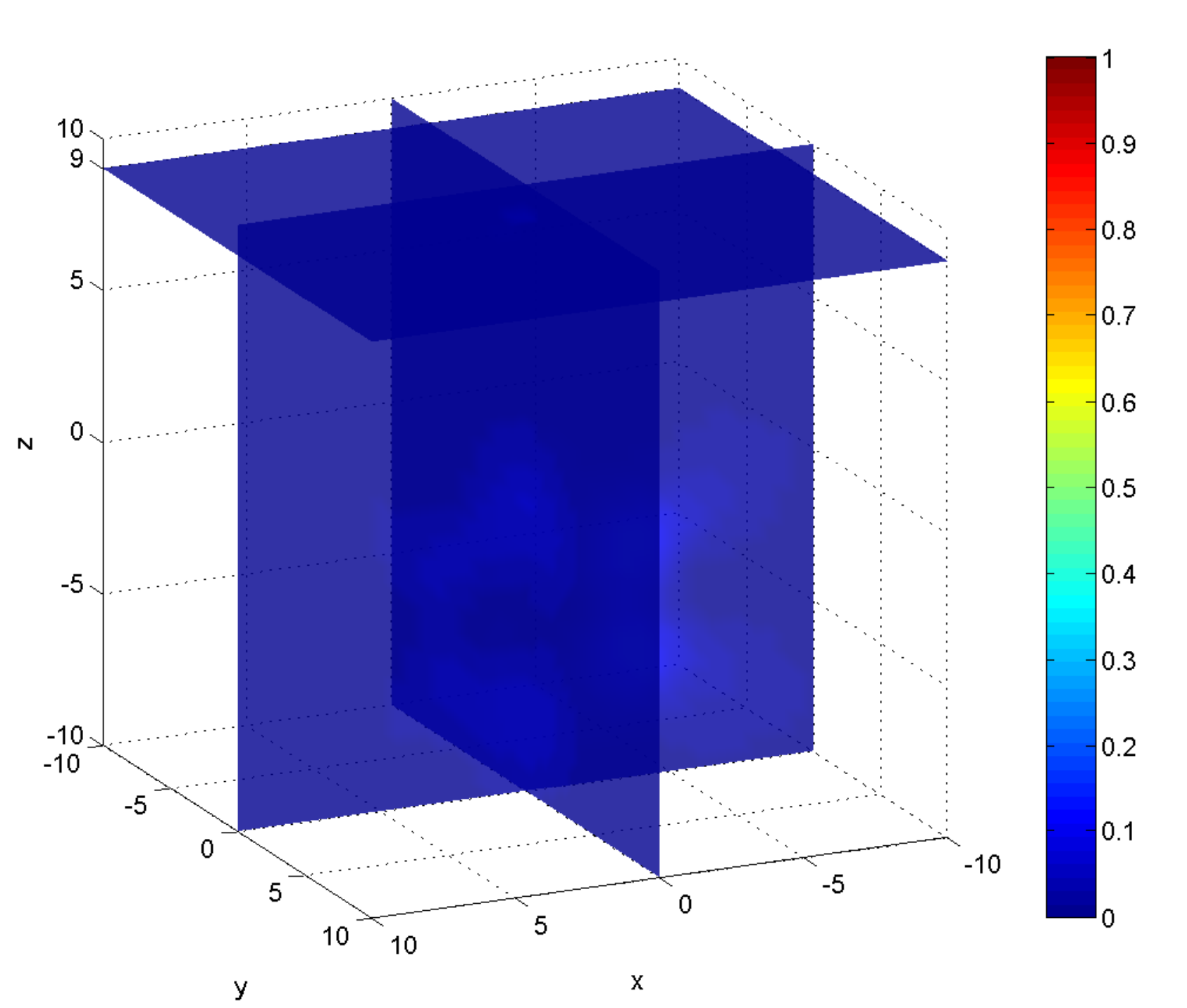}\hfill{}

\hfill{}(a)\hfill{}\hfill{}(b)\hfill{}

\caption{\label{fig:Recogonize little scatterer} Locating the small ball scatterer component in Example \textbf{KB}.
The multi-slice plots of the indicator function   (a) when $z_{0}$ is sufficiently near its actual position ($z_{0}=(0,\,0,\,-4.0056)$), or 
(b) when $z_{0}$ is away from its actual position ($z_{0}=(0,\,0,\,-4.1056))$.}
\end{figure}

\section{Conclusion}

In this paper we have developed several variants of the one-shot method proposed in \cite{LiLiuShangSun}. The methods can be used for the efficient numerical reconstruction of multiple multi-scale scatterers for inverse electromagnetic scattering problems. The methods are based on the local 'maximum' behaviors of the indicating functions aided by a candidate set of a priori known far-field data. Rigorous mathematical justifications are provided and several benchmark examples are presented to illustrate the efficiency of the schemes.

The local re-sampling technique  is shown to be an effective a posteriori position-fine-tuning method, which required  rough information of the position by an preprocessing stage of Scheme AR.  The local re-sampling technique adds only a small amount of computational overhead, but helps  calibrate the positions of the regular-size scatterers and determine the locations of the small-sized scatterers.

The present approaches can be extended in several directions including the one by making use of limited-view measurement data. The extension to the use of time-dependent measurement data would be nontrivial and poses interesting challenges for further investigation. Finally, it would be worthwhile to consider different noise background such as Gaussian and impulsive noise.

\end{document}